\documentclass[12pt]{article}
\usepackage[margin=1in]{geometry}
\usepackage{tabularx}
\usepackage{amsmath}
\usepackage{caption}
\usepackage[T1]{fontenc}
\usepackage{physics}
\usepackage{amssymb}
\usepackage{amsthm}
\usepackage{titling}
\usepackage{bm}
\usepackage[dvipsnames]{xcolor}
\usepackage{comment}
\usepackage{float}
\usepackage{enumerate}
\usepackage{dsfont}
\usepackage{mathtools}
\usepackage{amsfonts}
\usepackage{setspace}
\usepackage{graphicx}
\usepackage{soul}
\usepackage{subfig}
\usepackage{multirow}
\usepackage[symbol]{footmisc}
\usepackage[
backend=biber,
style=alphabetic,
sorting=anyvt
]{biblatex}
\usepackage[final]{hyperref}
\hypersetup{
	colorlinks=true,       
	linkcolor=red,        
	citecolor=red,        
	filecolor=magenta,     
	urlcolor=red        
}
\makeatletter
\renewcommand{\fnum@figure}{Fig. \thefigure}
\makeatother
\usepackage{blindtext}

\usepackage{chngcntr}

\newtheorem{mydef}{Definition}
\newtheorem{lemma}{Lemma}
\newtheorem{theorem}{Theorem}
\newtheorem*{theorem*}{Theorem} 

\newtheorem{proposition}{Proposition}

\newtheorem{corollary}{Corollary}
\renewcommand{\ge}{\geqslant}
\renewcommand{\le}{\leqslant}

\DeclareMathOperator{\Unif}{Unif}
\newcommand{\Mod}[1]{\ (\mathrm{mod}\ #1)}

\thanksmarkseries{alph}

\title{Asymptotically Optimal Adversarial Strategies for the Probability Estimation Framework}
\author{%
Soumyadip Patra\thanks{ \texttt{spatra@uno.edu}}%
\and
Peter Bierhorst\thanks{ \texttt{plbierho@uno.edu}}%
}
\date{{\normalsize Department of Mathematics,\\ University of New Orleans,\\ New Orleans, Louisiana, USA\\[\baselineskip]Dated: \today}}
\begin{document}

\maketitle

\begin{abstract}
The Probability Estimation Framework involves direct estimation of the probability of occurrences of outcomes conditioned on measurement settings and side information. It is a powerful tool for certifying randomness in quantum non-locality experiments. In this paper, we present a self-contained proof of the asymptotic optimality of the method. Our approach refines earlier results to allow a better characterisation of optimal adversarial attacks on the protocol. We apply these results to the (2,2,2) Bell scenario, obtaining an analytic characterisation of the optimal adversarial attacks bound by no-signalling principles, while also demonstrating the asymptotic robustness of the PEF method to deviations from expected experimental behaviour. We also study extensions of the analysis to quantum-limited adversaries in the (2,2,2) Bell scenario and no-signalling adversaries in higher $(n,m,k)$ Bell scenarios.
\end{abstract}

\tableofcontents

\section{Introduction}

Randomness has proven to be a valuable resource for a multitude of tasks, be it computation or communication. In cryptography, access to reliable random bits is essential, since the security of various cryptographic primitives is known to be compromised if the incorporated randomness is of poor quality \cite{10.1109/FOCS.2004.44, 10.1007/978-3-662-44371-2_26, cryptoeprint:2014/623}. In the study of random network modelling, being able to sample random graphs uniformly and (reliably) at random is crucial \cite{Orsini2015}. And for some problems, randomised algorithms are known to vastly outperform their deterministic counterparts \cite{motwani_raghavan_1995}.

A distinction between two notions of randomness, that of \emph{process} and \emph{product}, is discussed in \cite{10.1093/oso/9780198788416.001.0001} (chapter~8). Although both notions are tightly connected, randomness of a process refers to its \emph{unpredictability}, while that of a product refers to a \emph{lack of pattern} in it. An unpredictable process will, with high probability, produce a sequence (a string of bits, say) that is patternless; on the other hand, a seemingly irregular string of bits might not be unpredictable and instead be a probabilistic mixture of pre-recorded information. While product randomness suffices for tasks like Monte Carlo simulations, sampling and those involving randomised algorithms, cryptographic applications involving an adversary necessitate process randomness.

Process randomness, while being non-existent in the strictest interpretation of any classical theory, is permissible in quantum mechanics; an important example of this is quantum non-locality as manifested in a \emph{Bell experiment}. Quintessentially, the set up of a Bell experiment constitutes an entangled quantum system shared between two spatially separated stations $\mathsf{A}$ and $\mathsf{B}$ receiving inputs $\mathsf{x}$ and $\mathsf{y}$, and recording outcomes $\mathsf{a}$ and $\mathsf{b}$, respectively. If after $n$ successive trials the observed correlations between the outcomes conditioned on the settings violate a \emph{Bell inequality} then it can be ruled out that the outcomes were pre-assigned by some probabilistic mixture of deterministic processes. Also, the outcomes are (unpredictably) random, not only to the respective users of the devices at the two stations, but also to an adversary, even to one having a complete understanding of the Bell experiment. This relationship between non-locality in quantum mechanics and its random nature is at the foundation of various device independent random number generation protocols.

Device independence is considered a gold standard in cryptographic tasks such as quantum random number generation and quantum key distribution, in which the respective users are not required to know or trust the inner machinery of their devices, thus treating them as mere black boxes to which they can provide inputs and record outcomes. The only assumption that the experimental setup must satisfy is that the measurement choices of the devices must be uncorrelated with their inner workings. This is the measurement independence assumption, which is ultimately untestable, but is tacitly assumed, arguably, in almost all scientific experiments. The no-signalling condition that the outcome recorded at each station is not influenced by the choice of measurement at the other station, holds throughout the experiment because of a space-like separation between the stations and the impossibility of superluminal signalling in accordance with the special theory of relativity. Furthermore, the adversary trying to simulate the observed statistics may be considered computationally unbounded, a standard that falls under the paradigm of information-theoretic security. Over the years, technological advancement has facilitated loophole-free Bell non-locality experiments which have not only provided experimental validation for ruling out a classical description of nature \cite{PhysRevLett.115.250401, PhysRevLett.115.250402, Hensen2015, PhysRevLett.119.010402}, but have also found practical applications in device independent quantum randomness generation and device independent quantum key distribution \cite{Bierhorst2018, Shalm2021, PhysRevLett.126.050503}.

The probability estimation framework is a broadly applicable framework for performing device independent quantum randomness generation (DIQRNG) upon a finite sequence of loophole-free Bell experiment data, and involves direct estimation of the amount of certifiable randomness by obtaining high confidence bounds on the conditional probability of the observed measurement outcomes conditioned on the measurement settings in the presence of classical side information \cite{PhysRevA.98.040304, PhysRevResearch.2.033465, Bierhorst_2020}. 
Advantageous primarily for its demonstrated applicability to Bell tests with small Bell violations and high efficiency for a finite number of trials, it also can accommodate changing experimental conditions and allows early stoppage upon meeting certain criteria. Also, it can be extended to randomness generation with quantum devices beyond the device independent scenario. 

The probability estimation framework for DIQRNG is provably secure against adversaries who do not possess entanglement with the sources. Security against more general adversaries, with quantum entanglement with the sources, is possible with the quantum estimation framework \cite{Knill2018}, for which the constructions of the probability estimation framework can often be translated to the quantum estimation framework (as was done in~\cite{Zhang2020}), so that progress with the former framework can often be used for the more general latter framework~\cite{Shalm2021}.

The asymptotic optimality of the probability estimation framework was discussed in \cite{PhysRevResearch.2.033465}. The specific result of asymptotic optimality is as follows:~given a sufficiently large number of trials sampling from a fixed behaviour (i.e., a set of quantum statistics), the amount of certified randomness per trial is arbitrarily close to a certain upper limit. Then \cite{PhysRevResearch.2.033465} argues, appealing to convex geometry and the asymptotic equipartition property (AEP), that an adversary can always implement a probabilistic mixture of strategies, independent and identically distributed across successive experimental trials, that generates observed statistics consistent with the fixed behaviour while not needing to generate more than that same upper limit of randomness per trial that is certified by the probability estimation framework. This is important in the sense that the framework certifies all the randomness conceded by the adversary in that particular attack, while also showing that there is no advantage to be gained for the adversary by resorting to (more sophisticated) memory attacks.

In this paper, we provide a full derivation of the asymptotic optimality of the probability estimation framework, filling in some steps omitted by \cite{PhysRevResearch.2.013016}, along the way obtaining a better characterisation of the adversary's optimal probabilistic mixture for generating the observed statistics. A better understanding of the optimal attack in the asymptotic regime will set a benchmark enabling the implementer of the protocol to defend against these attack modes. Making precise the arguments from convex geometry, we explicitly describe the optimal strategy that an adversary (restricted only by the no-signalling condition) can employ with the minimum required number of different strategies in convex mixture to simulate the observed statistics. Our improvement, with a more self-contained approach, upon the result in \cite{PhysRevResearch.2.033465} is to reduce the smallest number of required strategies by one. Specifically, the smallest number of possible strategies is one more than the dimension of the set of admissible distributions of a trial. (We assume the set of admissible probability distributions of a given trial to be closed and convex, where we can take the convex closure when this assumption is not met; then the dimension $\text{dim}(C)$ of a non-empty convex subset $C$ of $X$ is the dimension of the smallest affine subset containing $C$.) Our derivation elucidates how only the classical form of the asymptotic equipartition property is needed for the probability estimation framework, allowing a simplified treatment. We also considered the question of \emph{robustness} of the probability estimation framework, deriving a sufficient condition for a probability estimation factor (optimised at a particular distribution) to certify randomness at a positive rate at a \emph{statistically different} distribution.

We apply our results to the (2,2,2) Bell scenario (the scenario of two parties, two measurement settings, and two outcomes), obtaining an analytic characterisation of the optimal attack of an adversary (restricted only by the no-signalling condition) holding classical side information. We show that the optimal adversarial attack involves a decomposition of the observed statistics in terms of a single extremal no-signalling (super-quantum) correlation and eight local deterministic correlations. The proof of optimality relies upon the fact that equal mixtures of two extremal no-signalling non-local super-quantum correlations are expressible as an equal mixture of four local deterministic correlations. We show that this result does not generalise to higher scenarios such as the (3,2,2), (2,3,2) and (2,2,3) Bell scenarios, thereby indicating that the possibility of an optimal attack involving only a single extremal strategy is only assured in the minimal (2,2,2) Bell scenario. Furthermore, we considered the possibility of an adversary holding classical side information (and hence, restricted to probabilistic attack strategies), but trying to simulate the observed statistics using quantum-achievable probability distributions, while conceding as little randomness as possible. Assuming uniform settings distribution, numerical studies restricted to a two-dimensional slice of the set of quantum-achievable distributions provided some initial evidence that the optimal quantum-achievable attack strategy involves only one extremal quantum correlation, but we were not able to settle this and have phrased it as a conjecture.

The rest of the article is organised as follows: In Section~\ref{sec:PEF}, we review the probability estimation framework where Theorem~\ref{PEF_firstthm} formalises the central idea and Theorem~\ref{PEF_thm2} establishes a lower bound on the smooth conditional min-entropy of the sequence of outcomes conditioned on the settings and side-information. We also present a simplified proof of Lemma~\ref{lemma_conv_PEF}, an important result enabling the algorithm for executing the PEF method, as compared to the proofs in \cite{PhysRevA.98.040304, PhysRevResearch.2.033465}. In Section~\ref{section3}, we present our complete proof of asymptotic optimality, study the implications for finding an optimal adversarial attack strategy, and derive a result on robustness. In Section~\ref{s:application}, we apply our results to the (2,2,2) Bell scenario obtaining an analytic characterisation of the optimal attack strategy for an adversary restricted only by the no-signalling condition. The optimal attack comprises of a decomposition of the observed statistics in terms of a single Popescu-Rohrlich (PR) correlation and (up to) eight local deterministic correlations. We show that for a higher number of parties, settings, and/or outcomes, a crucial result from the (2,2,2) Bell scenario concerning equal mixtures of extremal non-local no-signalling correlations does not hold, and infer that the optimal attack may require more than one non-local distribution in general.
Returning to the (2,2,2) scenario, we discuss a conjecture that the optimal strategy to mimic the observed statistics by means of a probabilistic mixture of quantum-achievable correlations constitutes only a single extremal quantum correlation and (up to) eight local deterministic correlations.

\section{The Probability Estimation Framework}\label{sec:PEF}

The probability estimation method relies on the probability estimation factor (PEF), which is a function assigning a score to the results of a single trial of a quantum experiment, with higher scores corresponding to more randomness. The paradigmatic application is to a Bell non-locality experiment comprising multiple spatially separated parties providing inputs (measurement settings) to measuring devices and recording outputs (observed outcomes); an experimental trial's results then consist of both the choice of inputs and the recorded outputs for that trial. After many repeated trials the product of the PEFs from all the trials is used to estimate the probability of outcomes conditioned on the settings. 

For the examples considered in Section~\ref{s:application}, we will consider the canonical scenario of two measuring parties Alice and Bob each selecting respective binary measurement settings $X$ and $Y$ and recording respective binary outcomes $A$ and $B$, which we refer to as the (2,2,2) Bell scenario. For now we treat things in a general manner as is done in~\cite{PhysRevA.98.040304} and~\cite{Knill2018}, modelling the trial settings for all parties and outcomes for all parties with single random variables $Z$ and $C$, respectively, taking values from respective finite-cardinality sets $\mathcal Z$ and $\mathcal C$. When applied to the (2,2,2) Bell scenario, $C$ comprises the ordered pair $(A,B)$ and $Z$ comprises the ordered pair $(X,Y)$.

\begin{figure}[H]
    \centering
    \includegraphics[scale=0.66]{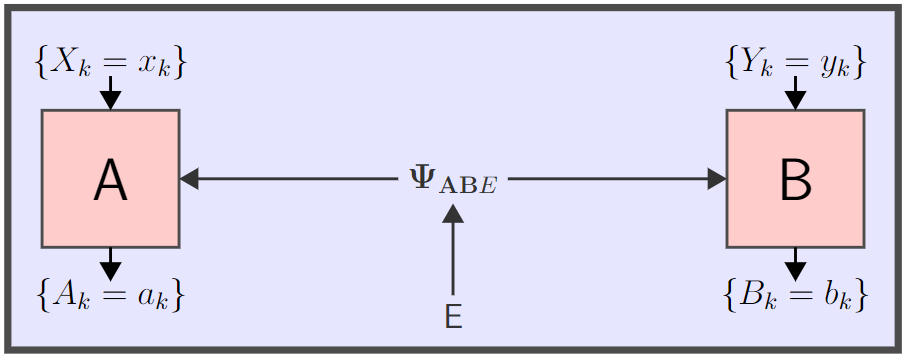}
    \caption{A schematic representation of the set-up for device independent randomness generation in a two-party experiment. The outer rectangular box represents a secure location. The adversary $\mathsf{E}$ has perfect knowledge of the processes inside the secure location but cannot tamper with them. The state $\Psi_{\mathbf{AB}E}$ represents the resource shared between the two parties. $X_{k},\,Y_{k}$ are the trial inputs and the $A_{k},\,B_{k}$ are the trial outcomes for the $k$th trial.}
    \label{fig:AlicenBob}
\end{figure}

The results of a sequence of $n$ time-ordered trials are represented by the sequences $\mathbf{C}=\{C_{i}\}_{i=1}^{n}, \mathbf{Z}=\{Z_{i}\}_{i=1}^{n}$; and so, $(\mathbf{C,Z})$ realises values $(\mathbf{c,z})\in\mathcal{C}^{n}\times\mathcal{Z}^{n}$, where $\mathcal{C}^{n},\mathcal{Z}^{n}$ are the $n$-fold Cartesian products of $\mathcal{C,Z}$. A PEF is then a real-valued function of $C$ and $Z$ satisfying certain conditions, while the product of PEFs from all trials will be a function of $\mathbf{C}$ and $\mathbf{Z}$. High values of the PEF product will correlate with low values of $\mathbb P(\mathbf{C}|\mathbf{Z})$, the conditional probability of the outcomes given the settings.

To define PEFs, we introduce the notion of a \textit{trial model}: A set $\Pi$ encompassing all joint probability distributions of settings and outcomes which are compatible with basic assumptions about the experiment. One important trial model that we consider is $\Pi_{\mathrm{Q}}$, consisting of joint distributions of $(C,Z)$ for which the conditional distribution of $C$ conditioned on $Z$ can be realised by a measurement on a quantum system. Here we introduce the convention, used throughout, of using lower case Greek letters with random variables as arguments to denote distributions, i.e., $\mu(C,Z)$ and $\mu(C\lvert Z)$ denote the joint distribution of $(C,Z)$ and the conditional distribution of $C$ given $Z$, respectively. Another important trial model is $\Pi_{\mathrm{NS}}$ (NS stands for ``no-signalling"), consisting of distributions for which probabilities of measurement outcomes at one location are independent of measurement settings at the other distant locations. (This is more clearly understood in considering the Alice-Bob example, where one of the no-signalling conditions is that $\sum_b\mu(A=a, B=b|X=x,Y=y)=\sum_b\mu(A=a, B=b|X=x,Y=y')$ for all $a,b,x$ and $y\ne y'$.) A third important trial model is the set $\Pi_{\mathrm{L}}$ of distributions for which the conditional distribution of outcomes conditioned on settings are \emph{local}, which means they can be expressed as convex mixtures of local deterministic behaviours. In the bipartite setting, the conditional distribution $\mu_{\mathrm{LD},\lambda}(A,B\lvert X,Y)$, also referred to as a behaviour, is local deterministic if $\mu_{\mathrm{LD},\lambda}(A=a,B=b\lvert X=x,Y=y)=[\![a=f(x,\lambda)]\!][\![b=g(y,\lambda)]\!]$ (where the notation $[\![\cdots]\!]$ represents the function that evaluates to $1$ if the condition within holds, $0$ otherwise). In words, the outcomes are functions of the local settings and the local hidden variable $\lambda$ which can be understood to be a list of outcomes for all possible settings. A formal definition involving more parties and an arbitrary (albeit same) number of outcomes and settings for each party can be found in~\eqref{e:def_LD}. The sets $\Pi_{\mathrm{L}},\Pi_{\mathrm{Q}}\text{ and }\Pi_{\mathrm{NS}}$ satisfy the following strict inclusions:
\begin{equation*}
\Pi_{\mathrm{L}}\subsetneq\Pi_{\mathrm{Q}}\subsetneq\Pi_{\mathrm{NS}}.
\end{equation*}
Certain distributions in $\Pi_{\mathrm{Q}}$ and $\Pi_{\mathrm{NS}}$ violate a Bell inequality and are known to contain randomness, they are contained in $\Pi_{\mathrm{Q}}\setminus\Pi_{\mathrm{L}}$ and $\Pi_{\mathrm{NS}}\setminus\Pi_{\mathrm{L}}$, respectively. It is precisely the inability to decompose such distributions into deterministic ones, as in $\Pi_{\mathrm{L}}$, that implies the presence of randomness. The objective of the PEF approach is to quantify the randomness contained in such distributions. As trial models specify the joint distribution $\mu(C,Z)$, and for the above examples of trial models we gave only the conditional distributions $\mu(C|Z)$, one must also specify the marginal distribution of the settings $\mu(Z)$. For the discussions of $\Pi_{\mathrm{Q}}$ and $\Pi_{\mathrm{NS}}$ in subsequent sections, any fixed distribution satisfying $\mu(Z=z)>0$ for all $z\in\mathcal{Z}$ is permitted. An example of a fixed settings distribution is the equiprobable distribution $\Unif(Z)$ defined as $\Unif(z)=1/\abs{\mathcal{Z}}$ for all $z\in\mathcal{Z}$. 

As a discrete probability distribution is effectively an ordered list of numbers in $[0,1]$ (the probabilities), trial models are always subsets of $\mathbb R^N$, where $N$ is fixed by the cardinality of $\mathcal C$ and $\mathcal Z$. This enables us to use a geometric approach to study these sets, which prove to be invaluable for some arguments.

We can now define PEFs. We use the notation $\mathbb{E}_{\mu}[\ldots]$ and $\mathbb{P}_{\mu}(\ldots)$ to denote expectation and probability, respectively, with respect to a distribution $\mu$; and for the sake of notational concision we sometimes omit commas in distributions or functions of more than one random variable, for instance, $\mu(CZ)$ and $f(CZ)$ must be understood to mean $\mu(C,Z)$ and $f(C,Z)$. 

\begin{mydef}[Probability Estimation Factor]\label{PEF_def}
A probability estimation factor (PEF) with power $\beta>0$ for the model of distributions $\Pi$ is a function $F\colon\mathcal{C}\times\mathcal{Z}\to\mathbb{R}^{+}$ of the random variables $(C,Z)$ such that for all $\sigma(CZ)\in \Pi$,
$\mathbb{E}_{\sigma}[F(CZ)\sigma(C\lvert Z)^{\beta}]\le 1$ holds.
\end{mydef}

In the expression above, $\sigma(C\lvert Z)$ denotes a random variable that is a function of the random variables $C$ and $Z$: $\sigma(C\lvert Z)$ is the random variable that assumes the standard conditional probability (according to $\sigma$) of $C$ taking the value $c$ conditioned on $Z$ taking the value $z$; it is assigned the value zero if the probability $\sigma(Z=z)$ is zero. The parameter $\beta$ can be any positive real value. We then note that the constant PEF $F(cz)=1$ for all $(c,z)\in\mathcal{C}\times\mathcal{Z}$ is a valid PEF for any choice of $\beta>0$. We will notice in the subsequent sections, however, that the parameter does have an effect on the method employed for choosing useful PEFs for the purpose of randomness certification; and in practice we choose the value of $\beta$ that corresponds to the maximum randomness certification.

Prior to defining a PEF we introduced the notion of a trial model. For the application of probability estimation to the outcomes of an experiment, which is a sequence of $n$ time-ordered trials, we introduce the notion of an \emph{experiment model}: It is a set $\Theta$ constraining the joint distribution of $\mathbf{C}$,$\mathbf{Z}$ and $E$, constructed as a chain of individual trial models $\Pi$; it consists of joint distributions $\mu(\mathbf{CZ}\lvert E=e)$ conditioned on the event $\{E=e\}$, where $E$ is the random variable denoting the adversary's side information and realising values $e$ from the finite set $\mathcal{E}$. It satisfies the following two assumptions:
\begin{align}\label{e:Expt_model_assumptions}
\mu(C_{i+1}Z_{i+1}\lvert \mathbf{C}_{\le i}=\mathbf{c}_{\le i},\mathbf{Z}_{\le i}=\mathbf{z}_{\le i},E=e) \in \Pi,\,\forall \,\mathbf{c}_{\le i}\in\mathcal{C}^{i},\,\mathbf{z}_{\le i}\in\mathcal{Z}^{i},\,e\in\mathcal{E},\nonumber \\
\mu(Z_{i+1},\mathbf{C}_{\le i}\mathbf{Z}_{\le i}\lvert E=e) = \mu(Z_{i+1}\lvert E=e)\mu(\mathbf{C}_{\le i}\mathbf{Z}_{\le i}\lvert E=e),\,\forall\,e\in\mathcal{E}.
\end{align}
In \eqref{e:Expt_model_assumptions}, $\mathbf{C}_{\le i},\mathbf{Z}_{\le i}$ denote the outcomes and measurement settings for the first $i\in[n]$ trials, where $[n]\coloneqq\{1,2,\ldots,n\}$, with $\mathbf{c}_{\le i},\mathbf{z}_{\le i}$ denoting their respective realisations. The random variables $C_{i+1},Z_{i+1}$ are the outcomes and settings for the $(i+1)$'th trial. The first condition in \eqref{e:Expt_model_assumptions} formalises the assumption that the (joint) probability of the $(i+1)$'th outcome and setting, conditioned on the outcomes and settings for the first $i$ trials and each realised value $E=e$ of the adversary's side information, belongs to the $(i+1)$'th trial model, i.e., it is compatible with the conditions dictated by the trial model. The second condition states that for each $E=e$ the setting for the \emph{next} trial is independent of the outcomes and settings of the \emph{past and present} trials. Our second condition is a stronger assumption than the corresponding assumption given in \cite{PhysRevA.98.040304}, which is as follows: The joint distribution $\mu$ of $\mathbf{CZ}E$ is such that $Z_{i+1}$ is independent of $\mathbf{C}_{\le i}$ conditionally on both $\mathbf{Z}_{\le i}$ and $E$. It is a straightforward exercise to check that our stronger assumption implies the one stated in~\cite{PhysRevA.98.040304}. While the weaker assumption is sufficient for the following result, we find the stronger assumption operationally clearer as an assumption that the future settings are independent of ``everything in the past" for each realisation of $e$.

For the rest of the paper we adopt the abbreviated notation of $\mu_{y}(X)$ for $\mu(X\lvert Y=y)$. The following theorem, appearing as Theorem $9$ in Appendix C in~\cite{PhysRevA.98.040304}, formalises the central idea behind the framework of probability estimation. We include a proof for this theorem in Appendix~\ref{a:PEFproof_lbscme} for completeness.

\begin{theorem}\label{PEF_firstthm}
Suppose $\mu\colon\mathcal{C}^{n}\times\mathcal{Z}^{n}\times\mathcal{E}\to[0,1]$ is a distribution of $\mathbf{CZ}E$ such that $\mu_{e}(\mathbf{CZ})\in\Theta$ for each $e\in\mathcal{E}$. Then for fixed $\beta, \epsilon > 0$
\begin{equation}\label{eq_PEF_def}
\mathbb{P}_{\mu_{e}}\left(\mu_{e}(\mathbf{C}\lvert\mathbf{Z})\ge \left(\epsilon\prod_{i=1}^{n}F_{i}(C_{i}Z_{i})\right)^{-1/\beta}\right)\le \epsilon
\end{equation}
holds for each $e\in\mathcal{E}$, where $F_{i}(C_{i}Z_{i})$ is the probability estimation factor for the $i$'th trial.
\end{theorem}
\textit{Proof}. See Appendix~\ref{a:PEFproof_lbscme}. 

The inequality~\eqref{eq_PEF_def} in Theorem~\ref{PEF_firstthm} can be understood, intuitively, as follows: When the trial-wise product $\prod_{i=1}^{n}F_{i}(C_{i}Z_{i})$ of the PEFs is large and so for fixed $\epsilon,\beta>0$, the quantity $(\epsilon\prod_{i=1}^{n}F_{i}(C_{i}Z_{i}))^{-1/\beta}$ is small, for each $e\in\mathcal{E}$ there is a very small probability (denoted by the outer probability $\mathbb{P}_{\mu_{e}}(\cdot)$) that the conditional probability of the sequence of outcomes $\mathbf{C}$ conditioned on the sequence of settings $\mathbf{Z}$ (denoted by $\mu_{e}(\mathbf{C}\lvert\mathbf{Z})$) is more than a small value.   

For information-theoretic purposes, it is useful to translate the bound in \eqref{eq_PEF_def} into a statement about min-entropy with respect to an adversary. An adversary's goal is to predict $C$. Conditioned on a particular realisation of the settings sequence $\mathbf{z}\in\mathcal{Z}^{n}$ and side information $e\in\mathcal{E}$, one can measure the ``predictability'' of the sequence of outcomes $\mathbf{C}$ with the following maximum probability: 
\begin{equation*}
\max_{\mathbf{c}\in\mathcal{C}^{n}}\mu(\mathbf{c}\lvert\mathbf{z}e).
\end{equation*}
It quantifies the best guess of the adversary. The $\mathbf{z}e$-conditional min-entropy of $\mathbf{C}$, corresponding to that particular realisation $\mathbf{z}e\in\mathcal{Z}^{n}\times\mathcal{E}$, is the following negative logarithm:
\begin{equation*}
    \mathbb{H}_{\infty,\mu}(\mathbf{C}\lvert\mathbf{z}e) \coloneqq -\log_{2}\left(\max_{\mathbf{c}\in\mathcal{C}^{n}}\mu(\mathbf{c}\lvert\mathbf{z}e)\right).
\end{equation*}
The subscript $\mu$ in the notation $\mathbb{H}_{\infty,\mu}(\cdots)$ refers to the distribution $\mu(\mathbf{CZ}E)$. The average $\mathbf{Z}E$-conditional min-entropy is then defined as follows:
\begin{equation*}
\mathbb{H}_{\infty,\mu}^{\mathsf{avg}}(\mathbf{C}\lvert\mathbf{Z}E) \coloneqq -\log_{2}\bigg[\sum_{\mathbf{z}e\in\mathcal{Z}^{n}\times\mathcal{E}}\Big(\max_{\mathbf{c}\in\mathcal{C}^{n}}\mu(\mathbf{c}\lvert\mathbf{z}e)\Big)\mu(\mathbf{z}e)\bigg].
\end{equation*}
But, information-theoretic security of cryptographic protocols take into account a more realistic measure of average $\mathbf{Z}E$-conditional min-entropy which involves a smoothing-parameter $\epsilon$, a type of error bound, and is known as the $\epsilon$-smooth average $\mathbf{Z}E$-conditional min-entropy. This quantity is useful for our scenario in which the probability distribution is not known exactly and its characteristics can only be inferred from observed data, which introduces the possibility of error. It is defined as follows.

\begin{mydef}[Smooth Average Conditional Min-Entropy]\label{def:smooth_avg_cond_minentropy} For a distribution $\mu\colon\mathcal{C}^{n}\times\mathcal{Z}^{n}\times\mathcal{E}\to[0,1]$ of $\mathbf{C,Z},E$ define the set $\mathcal{B}^{\epsilon}(\mu)$ of distributions of $\mathbf{C,Z},E$ as 
\begin{subequations}
\begin{equation}\label{eq_scme1}
\mathcal{B}^{\epsilon}(\mu)\coloneqq \{\sigma\colon\mathcal{C}^{n}\times\mathcal{Z}^{n}\times\mathcal{E}\to[0,1]\mid d_\mathrm{TV}(\sigma,\mu)\le\epsilon\},
\end{equation}
where $\epsilon\in(0,1)$ and $d_{\mathrm{TV}}(\sigma,\mu)$ is the total variation distance between $\sigma$ and $\mu$ defined as
\begin{equation}\label{e:d_TV}
    d_{\mathrm{TV}}(\sigma,\mu) \coloneqq \frac{1}{2}\sum_{\mathbf{cz}e\in\mathcal{C}^{n}\times\mathcal{Z}^{n}\times\mathcal{E}}\abs{\mu(\mathbf{cz}e)-\sigma(\mathbf{cz}e)}.
\end{equation}
The $\epsilon$-smooth average $\mathbf{Z}E$-conditional min-entropy is then defined as follows.
\begin{equation}\label{eq_scme2}
    \mathbb H_{\infty,\mu}^{\mathsf{avg},\epsilon}(\mathbf{C}\lvert\mathbf{Z}E) \coloneqq \max_{\sigma\in\mathcal{B}^{\epsilon}(\mu)}\Bigg[-\log_{2}\bigg[\sum_{\mathbf{z}e\in\mathcal{Z}^{n}\times\mathcal{E}}\Big(\max_{\mathbf{c}\in\mathcal{C}^{n}}\sigma(\mathbf{c}\lvert\mathbf{z}e)\Big)\sigma(\mathbf{z}e)\bigg]\Bigg].
\end{equation}
\end{subequations}
\end{mydef}

The lower bound obtained on this quantity goes as one of the inputs to extractor functions in randomness extraction, whose purpose is to convert random functions with uneven distributions into shorter, close to uniformly distributed bit strings. We note that alternative definitions of $\epsilon$-smooth conditional min-entropy can be used, for instance, the $\epsilon$-smooth \textit{worst-case} conditional min-entropy of \cite{10.1007/11593447_11}. A known result from the literature, proven in Proposition~\ref{l:avg_wst_scme} in Appendix \ref{a:scme}, justifies our usage of the $\epsilon$-smooth average conditional min-entropy without having to concern with the stricter $\epsilon$-smooth worst-case conditional min-entropy (defined in~\eqref{e:smoothwst}): specifically, the two quantities converge to one another in the asymptotic limit.

The result obtained from Theorem~\ref{PEF_firstthm} can be translated into a result on smooth average conditional min-entropy formalised in Theorem~\ref{PEF_thm2} below. This theorem appears as Theorem 1 in~\cite{PhysRevA.98.040304}. We include a proof for this theorem in Appendix~\ref{a:PEFproof_lbscme} for completeness. In the notation of $\epsilon$-smooth average $\mathbf{Z}E$-conditional min-entropy in~\eqref{eq_lb_min_ent}, the semicolon followed by $\mathsf{S}$ denotes that this information-quantity is assessed with respect to the distribution $\mu$ after conditioning on the occurrence of the event $\mathsf{S}$ defined in the statement of Theorem~\ref{PEF_thm2}. It pertains to an abort criterion. The protocol succeeds only if the product of the trial-wise PEFs exceeds some threshold value, otherwise it is aborted. So we want to establish the lower bound for smooth conditional min-entropy conditioned on the event that the protocol succeeds, because it is precisely this scenario in which we extract randomness. Since a completely predictable local distribution can always have a chance of passing the protocol, however minuscule (on the order of $(3/4)^n$, where the number of trials $n$ often goes up to millions)---and $\mu(\mathbf{c}|\mathbf{z})$ will equal 1 in this case---it is necessary to assume a small but positive lower bound on the probability of not aborting to derive a useful min-entropy bound. This can be thought of as another type of error parameter. The assumed lower bound for the probability of success of the protocol is $\kappa$.

\begin{theorem}\label{PEF_thm2}
Let $\mu$ be a distribution $\mu\colon\mathcal{C}^{n}\times\mathcal{Z}^{n}\times\mathcal{E}\to[0,1]$ of $\mathbf{C,Z},E$ such that for each $e\in\mathcal{E}$,
the following holds for every $\epsilon \in (0,1)$:
\begin{equation}\label{eq_PEF_def2}
\mathbb{P}_{\mu_{e}}\left(\mu_{e}(\mathbf{C}\lvert\mathbf{Z})\le\left(\epsilon\prod_{i=1}^{n}F_{i}\right)^{-1/\beta}\right)\ge 1-\epsilon,
\end{equation}
where $F_{i}$ is a PEF with power $\beta$ for the $i$'th trial. For a fixed choice of $\epsilon \in (0,1)$ and $p\ge \abs{\mathcal{C}}^{-n}$, define the event $\mathsf{S}\coloneqq \left\{\left(\epsilon\prod_{i=1}^{n}F_{i}\right)^{-1/\beta}\le p\right\}$. Then if $\kappa$ satisfies $0<\kappa \le \mathbb{P}_{\mu}(\mathsf{S})$, the following holds:
\begin{equation}\label{eq_lb_min_ent}
\mathbb H_{\infty,\mu}^{\mathsf{avg},\epsilon/\kappa}(\mathbf{C}\lvert\mathbf{Z}E;\mathsf{S})\ge \log_{2}(\kappa) - \log_{2}(p)
\end{equation}
\end{theorem}
\textit{Proof}. See Appendix~\ref{a:PEFproof_lbscme}.

Under the same conditions of Theorem~\ref{PEF_thm2}, the main result~\eqref{eq_lb_min_ent} admits a minor reformulation as follows. 
This is the formulation that aligns with the statement of Theorem 1 in~\cite{PhysRevA.98.040304}:
\begin{corollary}\label{coro1}
Let $\mu\colon\mathcal{C}^{n}\times\mathcal{Z}^{n}\times\mathcal{E}\to[0,1]$ be a distribution of $\mathbf{CZ}E$ and $F$ be a PEF with power $\beta$ such that~\eqref{eq_PEF_def2} holds for each $e\in\mathcal{E}$. For a fixed choice of $\epsilon\in (0,1)$, $p\ge |\mathcal C|^{-n}$, and positive $\kappa \le \mathbb P_\mu (\mathsf S)$ where $\mathsf S = \left\{(\epsilon \prod_{i=1}^n F_i )^{-1/\beta}\le p\right\}$, we have
\begin{equation}\label{e:lb_min_ent_2}
    \mathbb{H}_{\infty,\mu}^{\mathsf{avg},\epsilon}(\mathbf{C}\lvert\mathbf{Z}E;\mathsf{S}) \ge \left(1+\frac{1}{\beta}\right)\log_{2}(\kappa) - \log_{2}(p).
\end{equation}
\end{corollary}
\begin{proof} Use Theorem~\ref{PEF_thm2} with $\epsilon'=\kappa\epsilon$, $p'=p/\kappa^{1/\beta}$, and $\kappa'=\kappa$, noting that since $0 < \kappa \le 1$ and $\beta>0$ hold, we have $\epsilon'\in(0,1)$ and $p'\ge |\mathcal C|^{-n}$ as required for invoking the theorem. Then notice the corresponding event $\mathsf S' = \left\{(\epsilon' \prod_{i=1}^n F_i )^{-1/\beta}\le p'\right\}$ aligns with the event $\mathsf S$. 
\end{proof}

The above results hold when we consider distributions $\mu\colon\mathcal{C}^{n}\times\mathcal{Z}^{n}\times\mathcal{E}^{n}\to[0,1]$ of $\mathbf{CZE}$, i.e., where the side information is structured as a sequence of random variables. The proof remains the same with the exception that we condition on an arbitrary sequence of realisation $\mathbf{e}\in\mathcal{E}^{n}$ of $\mathbf{E}$. We consider this scenario in Section~\ref{section3} where we define an IID attack from the adversary.

Theorem \ref{PEF_firstthm} does not indicate how to find PEFs. One way to find useful PEFs is to first notice that the success criterion of the protocol is the event $\mathsf{S}$ that the inequality $(\epsilon\prod_{i=1}^{n}F_{i})^{-1/\beta}\le p$ holds, which can be equivalently expressed as
\begin{equation}\label{e:eq_expr_log2PEFsum}
\sum_{i=1}^{n}\log_{2}(F_{i})/\beta + \log_{2}(\epsilon)/\beta \ge -\log_{2}(p),
\end{equation}
where $\epsilon,\beta\text{ and }p$ are pre-determined quantities  to be chosen in advance of running the protocol. Then considering an anticipated trial distribution $\rho(CZ)$ based on observed results and calibrations from previous trials, in the limit of sufficiently large $n$ the difference between the term on the left hand side of \eqref{e:eq_expr_log2PEFsum} (which consists of the trial-wise sum of (base-2) logarithm of PEFs) and $n\mathbb{E}_{\rho}[\log_{2}(F(CZ))/\beta]$ will be either greater or less than zero with roughly equal probability. This follows from the Central Limit Theorem if the distribution remains roughly stable from trial to trial. Since it is desirable to have the largest value of $-\log_2(p)$ possible, one can then perform the following constrained maximisation using any convex programming software owing to the concavity of the objective function and the linearity of the constraints. 
\begin{align}\label{eq_PEF_opt}
\text{Maximise: }& \mathbb{E}_{\rho}[(n\log_{2}(F(CZ)) + \log_{2}(\epsilon))/\beta]\nonumber\\
\text{Subject to: }& \mathbb{E}_{\nu}[F(CZ)\nu(C\lvert Z)^{\beta}]\le 1,\text{ for all } \nu(CZ)\in\Pi,\nonumber\\
& F(cz)\ge 0,\text{ for all } (c,z)\in\mathcal{C}\times\mathcal{Z}
\end{align} 

Since $n,\epsilon\text{ and }\beta$ are fixed, it is sufficient to maximise $\mathbb{E}_{\rho}[\log_{2}(F(CZ))]$ subject to the same constraints. In practice, one can consider a range of values of $\beta$ and perform the constrained maximisation with the objective $\mathbb{E}_{\rho}[\log_{2}(F(CZ))]$, then plug in the maximum value in the expression $\mathbb{E}_{\rho}[(n\log_{2}(F(CZ)) + \log_{2}(\epsilon))/\beta]$ and obtain a plot with respect to the considered range of $\beta$ values (see, for example, Figure 2 in~\cite{Bierhorst_2020}; a similar pattern is observed in Figure~\ref{fig:beta_n_relation} in Section~\ref{s:application}).

The following lemma (from~\cite{PhysRevA.98.040304}, see Lemma 15)---for which we provide a more direct proof---enables us to restrict the satisfiability constraints of the optimisation routine in \eqref{eq_PEF_opt} to the extremal distributions of the model $\Pi$ under the condition that the model is convex and closed. So, the first line of constraints in \eqref{eq_PEF_opt} can be replaced with $\mathbb{E}_{\nu}[F(CZ)\nu(C\lvert Z)^{\beta}]\le 1,\,\forall\nu(CZ)\in\Pi_{\text{extr}}$, where $\Pi_{\text{extr}}$ is the set of extremal distributions of $\Pi$. If the model $\Pi$ is not convex and closed, we take its convex closure. In words, the lemma states that if $F(CZ)$ is a PEF with power $\beta>0$ for the distributions $\sigma_{1}(CZ)\text{ and }\sigma_{2}(CZ)$, then it is a PEF with the same power for all distributions that can be expressed as a convex combination of $\sigma_{1}$ and $\sigma_{2}$.

\begin{lemma}\label{lemma_conv_PEF}
For distributions $\sigma_{i}(CZ)\in\Pi$ satisfying $\mathbb{E}_{\sigma_{i}}[F(CZ)\sigma_{i}(C\lvert Z)^{\beta}]\le 1$, for $i=1,2$, if $\sigma(CZ)\in\Pi$ is expressible as $\sigma(CZ)=\lambda\sigma_{1}(CZ)+(1-\lambda)\sigma_{2}(CZ)$ for $\lambda\in[0,1]$, then it satisfies $\mathbb{E}_{\sigma}[F(CZ)\sigma(C\lvert Z)^{\beta}]\le 1$.
\end{lemma}
\begin{proof}
For $z$ such that $\sigma_{1}(z),\sigma_{2}(z)>0$, we have $\sigma(z)>0$ as well, and from $\sigma(CZ)=\lambda\sigma_{1}(CZ)+(1-\lambda)\sigma_{2}(CZ)$, straightforward algebra shows that $\sigma(c\lvert z)=\delta\sigma_{1}(c\lvert z)+(1-\delta)\sigma_{2}(c\lvert z)$ for any $(c,z)\in\mathcal{C}\times\mathcal{Z}$, where $\delta=\lambda\sigma_{1}(z)/\sigma(z)\in[0,1]$. 
Since for $\alpha > 1$, $x^\alpha$ is convex for $x\ge 0$, we can write 
\begin{align}
    \sigma(c\lvert z)^{1+\beta} &\le \delta\sigma(c\lvert z)^{1+\beta} + (1-\delta)\sigma_{2}(c\lvert z)^{1+\beta}\nonumber \\
    \Rightarrow \sigma(c\lvert z)^{1+\beta}\sigma(z) &\le \lambda\sigma_{1}(c\lvert z)^{1+\beta}\sigma_{1}(z) + (1-\lambda)\sigma_{2}(c\lvert z)^{1+\beta}\sigma_{2}(z).\label{e:convexp}
\end{align}
Turning to cases where $\sigma_{1}(z)$ and/or $\sigma_{2}( z)$ may equal zero, we can also demonstrate \eqref{e:convexp} under the convention of taking $\sigma_i(c\lvert z)$ to be zero when $\sigma_{i}( z)=0$. Then the inequality holds as an equality when $\sigma_{1}(z)=\sigma_{2}(z)=0$ (which implies $\sigma (z)=0$ as well); for $0=\sigma_{2}(z)<\sigma_{1}(z)$
one can verify \eqref{e:convexp} after noting $\sigma(cz)=\lambda \sigma_1(cz)$ and $\sigma(z)=\lambda \sigma_1(z)$, and the $0=\sigma_{1}(z)<\sigma_{2}(z)$ case follows symmetrically. Now multiplying both sides of \eqref{e:convexp} by $F(cz)$ and summing over $(c,z)\in\mathcal{C}\times\mathcal{Z}$ gives
\begin{align*}
    \sum_{c,z}F(cz)\sigma(c\lvert z)^{1+\beta}\sigma(z)
    &\le \lambda\sum_{c,z}F(cz)\sigma_{1}(c\lvert z)^{1+\beta}\sigma_{1}(z) + (1-\lambda)\sum_{c,z}F(cz)\sigma_{2}(c\lvert z)^{1+\beta}\sigma_{2}(z)\\
\Rightarrow \mathbb{E}_{\sigma}[F(CZ)\sigma(C\lvert Z)^{\beta}] &\le \lambda\mathbb{E}_{\sigma_{1}}[F(CZ)\sigma_{1}(C\lvert Z)^{\beta}] + (1-\lambda)\mathbb{E}_{\sigma_{2}}[F(CZ)\sigma_{2}(C\lvert Z)^{\beta}]\\
&\le \lambda + (1-\lambda) = 1.
\end{align*}
\end{proof}

We remark that the result of Lemma 1 can also be obtained through specialisation of known quantum results to classical distributions; however, this requires a more technical argument with additional machinery. To elaborate, the proof for lemma \ref{lemma_conv_PEF} involves showing the joint convexity of $\sigma(C\lvert Z)^{1+\beta}\sigma(Z)$ which can be seen as a special case of the joint convexity of sandwiched R\'enyi powers. To be more specific, it arises as a special case of the joint convexity of $e^{\beta D_{1+\beta}(\sigma\lvert\lvert\omega)}$ for $\beta > 0$ when the distribution $\omega(CZ)$  is taken to be $\omega(cz)=\sigma(z)/\abs{\mathcal{C}},\,\forall (c,z)\in\mathcal{C}\times\mathcal{Z}$. Notice that $D_{1+\beta}(\sigma\lvert\lvert\omega)$ is the (classical) R\'enyi divergence of order $(1+\beta)\in(1,\infty)$ of $\sigma(CZ)$ with respect to $\omega(CZ)$. The functional $e^{D_{1+\beta}(\sigma\lvert\lvert\omega)}$ can also be seen as a specialisation (to classical states) of the same functional, defined in terms of (quantum) density states $\sigma$ and $\omega$, whose joint convexity was proven in proposition 3 of \cite{2013JMP....54l2201F} with an extended technical argument. 

\section{Asymptotic Performance}\label{section3}

The results of the previous section give us a method for certifying randomness. In this section, we assess the asymptotic performance of the method. Our figure of merit is the amount of randomness certified per trial, as measured by the average conditional min-entropy divided by the number of trials $n$. 
We will see in this section that the PEF method is asymptotically optimal, in the following sense: given a fixed observed distribution, the PEF method can asymptotically certify an amount of per-trial conditional min-entropy that is equal to the actual per-trial conditional min-entropy generated by an adversary replicating the observed distribution with as little randomness as possible.

To elaborate on this, consider that the adversary's goal is to minimise the following quantity:
$$
\frac{1}{n} \mathbb H^{\mathsf{avg}}_{\infty,\mu}(\mathbf{C}|\mathbf{Z}E).
$$
We assume that the adversary has complete knowledge of the distribution $\mu$, and can have access to not just the realised value of $E$, but also the realised value of $\mathbf Z$ in guessing $\mathbf C$. This access to $\mathbf{Z}$ aligns with the paradigm, as discussed in \cite{Bierhorst2018}, of ``using public (settings) randomness to generate private (outcome) randomness". The adversary is constrained, however, in that the statistics when marginalised over $E$ must appear to be consistent with an expected observed trial distribution $\rho(CZ)$ for the protocol to not abort. Technically, all that is necessary for the protocol to pass is that the observed product of the PEFs must exceed some threshold value chosen by the experimenter---which could be possible with high probability with many different distributions $\mu$---but as the experimenter's threshold value will likely be chosen based on a full behaviour that they expect to observe, we study attacks that match the expected observed trial distribution exactly. We will find attacks meeting this criterion that are asymptotically optimal for minimising the conditional min-entropy.

Given an expected observed distribution, how can the adversary generate observed statistics consistent with it while yielding as little randomness as possible? She can employ a strategy of preparing multiple different states to be measured that will yield different distributions, each one consistent with the trial model $\Pi$, whose convex mixture is equal to the observed distribution. If she has an auxiliary random variable $E$ realising values from the finite-cardinality set $\mathcal{E}$ and recording which state was prepared on which trial, she can predict better the outcome conditioned on her side information $E=e$, in conjunction with the settings $Z$. Indeed, some of her $e$-conditional distributions could be deterministic, in which case she does not yield any randomness to Alice and Bob on a trial where $E$ takes that value.\footnote[2]{A deterministic distribution must be understood as the product of a fixed settings distribution and a deterministic behaviour (conditional distribution of the outcomes conditioned on settings).} But if the overall observed statistics are non-local, then she is forced to prepare at least some states that contain randomness even conditioned on $e$; this, in essence, is because the information that she possesses with $E$ is a local hidden variable.

\subsection{I.I.D.~Attacks} Given a convex decomposition of the observed distribution, the adversary's simplest form of an attack is to select $e$ from some finite-cardinality set $\mathcal E$ in an i.i.d manner on each trial according to the distribution that recovers the observed distribution $\rho(CZ)$. A more general attack would allow her to use memory of earlier trials, but we will see later that, asymptotically, this does not yield meaningful improvement. 

Operationally, we do not like to think of the adversary accessing the devices in between trials to provide a choice of $e_i$ for each trial. Instead, one can imagine her randomly sampling from the distribution of $\mathbf E$ for all trials, coming up with a choice $\mathbf e$ that encodes all the choices of $e_i$ for each trial, and then supplying this choice to the measured system, in advance, to determine its behaviour in each trial. She keeps a record of $\mathbf e$ to help her predict $C$ later. Through this sampling process there is a small chance that she will sample an atypical ``bad" $\mathbf e$ that results in statistics deviating from the observed distribution, but the probability that her $\mathbf e$ is typical is asymptotically high. Our figure of merit for the adversary now is: 
$$
\frac{1}{n} \mathbb H^{\mathsf{avg}}_{\infty,\mu}(\mathbf{C}|\mathbf{Z}\mathbf{E}),
$$
which she wants to minimise with a distribution that, marginalised over $\mathbf {E}$, is consistent with i.i.d sampling from an expected observed distribution $\rho$. We formally define the set of distributions $\omega\colon\mathcal{C}\times\mathcal{Z}\times\mathcal{E}\to[0,1]$ of $C,Z,E$ mimicking $\rho$ through such a convex decomposition as follows, where $e$ is shorthand for the event $\{E=e\}$: 
\begin{equation}\label{def:strategy}
    \Sigma_{\mathrm E}^{\rho}\coloneqq \Big\{\omega (CZE) \colon \omega(CZ\lvert e)\in\Pi\,\forall e\in\mathcal{E},\sum_{e\in\mathcal{E}}\omega(CZ|e)\omega(e)=\rho(CZ)\Big\}.
\end{equation}
Then an IID attack can be defined as follows.
\begin{mydef}[IID Attack]\label{def:IID_attack} Given a distribution $\omega(CZE)\in\Sigma_{\mathrm E}^{\rho}$, we define an IID attack (with $\omega$) to be the distribution $\phi$ consisting of $n$ independent and identical realisations of random variables $C_{i},Z_{i},E_{i}$ distributed according to $\omega$; i.e., the joint distribution of the sequence of random variables $\mathbf{C,Z,E}$ is $\phi\colon\mathcal{C}^{n}\times\mathcal{Z}^{n}\times\mathcal{E}^{n}\to[0,1]$ such that $\phi(\mathbf{CZE})=\prod_{i=1}^{n}\omega(C_{i}Z_{i}E_{i})$.
\end{mydef}
As mentioned earlier, the adversary randomly samples from the distribution of $\mathbf{E}$ which represents their knowledge of all trials; $\mathbf{e}\equiv(e_{1},e_{2},\ldots,e_{n})\in\mathcal{E}^{n}$ encodes the individual choices $e_{i}$ for trial $i\in\{1,2,\ldots,n\}$. The IID attack model satisfies the two assumptions of the experiment model discussed earlier (see~\eqref{e:Expt_model_assumptions} and the short discussion that follows immediately). Namely, the (joint) probability of the $(i+1)$'th trial outcome and input setting, conditioned on each realisation of the outcomes and settings for the first $i$ trials and each realisation $\mathbf{e}\in\mathcal{E}^{n}$ of the side information, satisfies the conditions of the trial model; and conditioned on each $\mathbf{e}\in\mathcal{E}^{n}$, the settings for the $(i+1)$'th trial are (unconditionally) independent of the outcomes and settings of the past and present trials (i.e., the first $i$ trials). This is formally stated and proved in Lemma~\ref{l:iidattack_lemma} below.

\begin{lemma}\label{l:iidattack_lemma}
The IID attack strategy as defined in Definition~\ref{def:IID_attack} satisfies the following conditions.
\begin{align}
    \phi(C_{i+1}Z_{i+1}\lvert\mathbf{c}_{\le i}\mathbf{z}_{\le i}\mathbf{e}) &\in \Pi,\,\forall\,\mathbf{c}_{\le i}\in\mathcal{C}^{i},\,\mathbf{z}_{\le i}\in\mathcal{Z}^{i},\,\mathbf{e}\in\mathcal{E}^{n}\label{e:lemma3_r1} \\
    \phi(Z_{i+1}\mathbf{C}_{\le i}\mathbf{Z}_{\le i}\lvert\mathbf{e}) &= \phi(Z_{i+1}\lvert\mathbf{e})\phi(\mathbf{C}_{\le i}\mathbf{Z}_{\le i}\lvert\mathbf{e}),\,\forall\,\mathbf{e}\in\mathcal{E}^{n}\label{e:lemma3_r2}
\end{align}
\end{lemma}

\begin{proof}
Consider the distribution $\phi(\mathbf{CZ}\lvert\mathbf{e})$ conditioned on a realisation $\mathbf{E}=\mathbf{e}$, where $\phi(\mathbf{CZE})=\prod_{i=1}^{n}\omega(C_{i}Z_{i}E_{i})$. Notice that $\phi(\mathbf{CZ}\lvert\mathbf{e})=\prod_{i=1}^{n}\omega(C_{i}Z_{i}\lvert e_{i})$. Marginalising over the random variables $C_{i+2},C_{i+3},\ldots,C_{n},Z_{i+2},Z_{i+3},\ldots,Z_{n}$ we get:
\begin{equation}\label{e:lemma3_eq1}
    \phi(C_{i+1}Z_{i+1}\mathbf{C}_{\le i}\mathbf{Z}_{\le i}\lvert\mathbf{e}) = \prod_{j=1}^{i+1}\omega(C_{j}Z_{j}\lvert e_{j})
\end{equation}
Corresponding to a particular realisation $\mathbf{c}_{\le i}\in\mathcal{C}^{i},\,\mathbf{z}_{\le i}\in\mathcal{Z}^{i}$, we then have $\phi(C_{i+1}Z_{i+1}\mathbf{c}_{\le i}\mathbf{z}_{\le i}\lvert\mathbf{e})=\omega(C_{i+1}Z_{i+1}\lvert e_{i+1})\prod_{j=1}^{i}\omega(c_{j}z_{j}\lvert e_{j})$; and since $\phi(\mathbf{c}_{\le i}\mathbf{z}_{\le i}\lvert\mathbf{e})=\prod_{j=1}^{i}\omega(c_{j}z_{j}\lvert e_{j})$, we have
\begin{equation}\label{e:lemma3_eq2}
    \frac{\phi(C_{i+1}Z_{i+1}\mathbf{c}_{\le i}\mathbf{z}_{\le i}\lvert\mathbf{e})}{\phi(\mathbf{c}_{\le i}\mathbf{z}_{\le i}\lvert\mathbf{e})} = \phi(C_{i+1}Z_{i+1}\lvert\mathbf{c}_{\le i}\mathbf{z}_{\le i}\mathbf{e}) = \omega(C_{i+1}Z_{i+1}\lvert e_{i+1}).
\end{equation}
$\omega(C_{i+1}Z_{i+1}\lvert e_{i+1})$ belongs to the set $\Pi$ for all values of $e_{i+1}\in\mathcal{E}$ (by construction of the set $\Sigma_{E}$, see~\eqref{def:strategy}). Since~\eqref{e:lemma3_eq2} is true for all realisations $\mathbf{c}_{\le i}\in\mathcal{C}^{i},\,\mathbf{z}_{\le i}\in\mathcal{Z}^{i},\,\mathbf{e}\in\mathcal{E}^{n}$ we conclude~\eqref{e:lemma3_r1} holds. Next, marginalising~\eqref{e:lemma3_eq1} over $C_{i+1}$ we have:
\begin{equation}\label{e:lemma_eq3}
    \phi(Z_{i+1}\mathbf{C}_{\le i}\mathbf{Z}_{\le i}\lvert\mathbf{e}) = \omega(Z_{i+1}\lvert e_{i+1})\prod_{j=1}^{i}\omega(C_{j}Z_{j}\lvert e_{j})=\phi(Z_{i+1}\lvert\mathbf{e})\phi(\mathbf{C}_{\le i}\mathbf{Z}_{\le i}\lvert\mathbf{e})
\end{equation}
In~\eqref{e:lemma_eq3}, $\omega(Z_{i+1}\lvert e_{i+1})=\phi(Z_{i+1}\lvert\mathbf{e})$ can be observed by marginalising~\eqref{e:lemma3_eq1} over the random variables $C_{1},\ldots,C_{i},C_{i+1},Z_{1},\ldots,Z_{i}$ and $\phi(\mathbf{C}_{\le i}\mathbf{Z}_{\le i}\lvert\mathbf{e})=\prod_{j=1}^{i}\omega(C_{j}Z_{j}\lvert e_{j})$ (from marginalising \eqref{e:lemma3_eq1} over $C_{i+1},Z_{i+1}$);~\eqref{e:lemma_eq3} is true for all $\mathbf{e}\in\mathcal{E}^{n}$, hence we conclude~\eqref{e:lemma3_r2}.
\end{proof}

Next, the adversary would like to implement an attack that ``generates as little randomness as possible". One measure of the randomness is the conditional Shannon entropy of the outcomes $C$ conditioned on the inputs $Z$ and the side information $E$.

\begin{mydef}[Conditional Shannon Entropy]
For a distribution $\mu\colon\mathcal{C}\times\mathcal{Z}\times\mathcal{E}\to[0,1]$ of $C,Z,E$ the conditional Shannon entropy of the outcomes $C$ conditioned on the settings $Z$ and the side information $E$ is defined as
\begin{equation}\label{e:con_Sh_entropy}
    \mathbb{H}_{\mu}(C\lvert ZE) = -\sum_{cze}\log_{2}\mu(c\lvert ze)\mu(cze)
    = \mathbb{E}_{\mu}[-\log_{2}\mu(C\lvert ZE)].
\end{equation}
\end{mydef}
The Greek letter $\mu$ in the subscript of $\mathbb{H}_{\mu}(\cdot\lvert\cdot)$ refers to the distribution $\mu(CZE)$ with respect to which the conditional Shannon entropy is defined.

Theorem \ref{thm_AEP} below shows that $H_\omega(C|ZE)$ is an asymptotic upper bound on the per-trial conditional min-entropy that the adversary generates with an IID attack employing a trial distribution $\omega$ that is consistent with the observed distribution $\rho$. This result was discussed but not demonstrated explicitly in~\cite{PhysRevResearch.2.033465}. The proof of Theorem~\ref{thm_AEP} involves one of the fundamental technical tools from information theory, the (classical) Asymptotic Equipartition Property (AEP), or equivalently the notion of typical sequences which has the weak law of large numbers at its core. 

Suppose $\mu$, the distribution of all trials, is obtained as $n$ i.i.d. copies of a single-trial distribution $\omega$. Then for $\epsilon_{\mathrm{a}}\in(0,1)$, $\delta\ge 0$ there exists $N(\epsilon_{\mathrm{a}},\delta)$ such that $n\ge N(\epsilon_{\mathrm{a}},\delta)$ ensures $\mathbb{E}_{\mu(\mathbf{ZE})}[\mathbb{P}_{\mu(\mathbf{C}\lvert\mathbf{ZE})}(\mu(\mathbf{C}\lvert\mathbf{ZE})\ge\gamma)]\ge 1-\epsilon_{\mathrm{a}}$, where $\gamma = 2^{-nH_{\omega}(C\lvert ZE)-n\delta}$ and $H_{\omega}(C\lvert ZE)$ is the conditional Shannon entropy. We refer to this as the AEP condition; it holds by a conditional form of the classical AEP (see, for instance, Section 14.6 in~\cite{wilde_2013}). The set $B^{\epsilon_{\mathrm{s}}}(\mu)$ of distributions of $\mathbf{C,Z,E}$ that are within a $\mathrm{TV}$-distance of $\epsilon_{\mathrm s}$ from $\mu$ and the sets $\mathcal{A}_{\mathbf{ze}}$ are as defined below:
    \begin{align}
    \mathcal{B}^{\epsilon_{\mathrm{s}}}(\mu) &\coloneqq \{\sigma\colon\mathcal{C}^{n}\times\mathcal{Z}^{n}\times\mathcal{E}^{n}\to[0,1] \mid d_\mathrm{TV}(\mu,\sigma)\le \epsilon_{\mathrm{s}}\}\label{eq15},\\
    \mathcal{A}_{\mathbf{ze}} &\coloneqq \{\mathbf{c}\in\mathcal{C}^{n}\mid \mu(\mathbf{c}\lvert\mathbf{ze})\ge \gamma\},\label{eq16}
    \end{align}
where $\mathcal{A}_{\mathbf{ze}}$ is defined for any $\mathbf{ze}$ for which $\mu(\mathbf{ze})>0$. Note that the case $\epsilon_{\mathrm s}=0$ reduces to a bound on the standard (non-smooth) average conditional min-entropy. We now state the result as follows.

\begin{theorem}\label{thm_AEP}
Let $\mu$ be an IID attack with $\omega$. For $\epsilon_{\mathrm s}\ge 0$, $\epsilon_{\mathrm a},\delta >0$ and $\epsilon_{\mathrm a}+2\epsilon_{\mathrm s}<1$, there exists $N(\epsilon_{\mathrm a},\epsilon_{\mathrm s},\delta)$ such that for $n\ge N(\epsilon_{\mathrm a},\epsilon_{\mathrm s},\delta)$
\begin{equation}\label{thm_AEP_ineq}
    \frac{1}{n}\mathbb H_{\infty,\mu}^{\mathsf{avg},\epsilon_{\mathrm s}}(\mathbf{C}\lvert\mathbf{ZE}) \le \mathbb H_{\mu}(C\lvert ZE) + \frac{1}{n}\log_{2}\frac{1}{1-\epsilon_{\mathrm a}-2\epsilon_{\mathrm s}} + \delta.
\end{equation}
\end{theorem}

\begin{proof}
Throughout, we follow the convention that $\sigma(\mathbf{c}\lvert\mathbf{ze})=0$ for all $\mathbf{c}\in\mathcal{C}^{n}$ for any $\mathbf{ze}\in\mathcal{Z}^{n}\times\mathcal{E}^{n}$ with $\sigma(\mathbf{ze})=0$. We begin with the inequality $d_\mathrm{TV}(\sigma,\mu)\le\epsilon_{\mathrm{s}}$ that any $\sigma\in\mathcal{B}^{\epsilon_{\mathrm s}}(\mu)$ must satisfy and proceed as follows:

\begin{align}
    2\epsilon_{\mathrm s} &\ge \norm{\mu-\sigma}_{1}=\sum_{\mathbf{cze}\in\mathcal{C}^{n}\times\mathcal{Z}^{n}\times\mathcal{E}^{n}}\abs{\mu(\mathbf{cze})-\sigma(\mathbf{cze})}\nonumber\\
    &\ge \sum_
{\mathbf{ze}:\mu(\mathbf{ze})>0}\,\sum_{\mathbf{c} \in \mathcal{A}_{\mathbf{ze}}}\abs{\mu(\mathbf{cze})-\sigma(\mathbf{cze})} \label{thm3_1}\\
    &\ge \abs{\sum_
{\mathbf{ze}:\mu(\mathbf{ze})>0}\,\sum_{\mathbf{c} \in \mathcal{A}_{\mathbf{ze}}}\left(\mu(\mathbf{cze})-\sigma(\mathbf{cze})\right)}\label{thm3_2} \\
    &= \abs{\mathbb{E}_{\mu(\mathbf{ZE})}[\mathbb{P}_{\mu(\mathbf{C}\lvert\mathbf{ZE})}(\mu(\mathbf{C}\lvert\mathbf{ZE})\ge\gamma)] - \sum_
{\mathbf{ze}:\mu(\mathbf{ze})>0}\,\sum_{\mathbf{c} \in \mathcal{A}_{\mathbf{ze}}}\sigma(\mathbf{cze})}\nonumber\\
&\ge \mathbb{E}_{\mu(\mathbf{ZE})}[\mathbb{P}_{\mu(\mathbf{C}\lvert\mathbf{ZE})}(\mu(\mathbf{C}\lvert\mathbf{ZE})\ge\gamma)] - \sum_
{\mathbf{ze}:\mu(\mathbf{ze})>0}\,\sum_{\mathbf{c} \in \mathcal{A}_{\mathbf{ze}}}\sigma(\mathbf{cze}).\nonumber
\end{align}
The inequality in~\eqref{thm3_1} follows as a result of the sum containing fewer terms; the inequality in~\eqref{thm3_2} follows from the triangle inequality. Now from the AEP condition mentioned above we have the following:
\begin{equation}\label{thm3_3}
    \sum_
{\mathbf{ze}:\mu(\mathbf{ze})>0}\,\sum_{\mathbf{c} \in \mathcal{A}_{\mathbf{ze}}}\sigma(\mathbf{cze}) \ge \mathbb{E}_{\mu(\mathbf{ZE})}[\mathbb{P}_{\mu(\mathbf{C}\lvert\mathbf{ZE})}(\mu(\mathbf{C}\lvert\mathbf{ZE})\ge\gamma)] - 2\epsilon_{\mathrm s}\ge 1-\epsilon_{\mathrm a}-2\epsilon_{\mathrm s}.
\end{equation}
For any $\sigma\in\mathcal{B}^{\epsilon_{\mathrm s}}(\mu)$, we define $M_{\mathbf{ze}}^{\sigma}$ for any $\mathbf{ze}\in\mathcal{Z}^{n}\times\mathcal{E}^{n}$ as $M_{\mathbf{ze}}^{\sigma} \coloneqq \max_{\mathbf{c}\in\mathcal{C}^{n}}\sigma(\mathbf{c}\lvert\mathbf{ze})$. The average conditional maximum probability is then expressed as $\bar{M}_{\sigma}\coloneqq\sum_{\mathbf{ze}}M_{\mathbf{ze}}^{\sigma}\sigma(\mathbf{ze})$. Now, because $1\le \sum_{c\in \mathcal{A}_{\mathbf{ze}}}\mu(\mathbf{c}|\mathbf{ze})\le \gamma |\mathcal{A}_{\mathbf{ze}}|$, we have $\abs{\mathcal{A}_{\mathbf{ze}}}\le 1/\gamma$ for each $\mathbf{ze}$, and we can write:
\begin{align}\label{thm3_4}
     \sum_
{\mathbf{ze}:\mu(\mathbf{ze})>0}\,\sum_{\mathbf{c}\in\mathcal{A}_{\mathbf{ze}}}\sigma(\mathbf{cze}) &=  \sum_
{\mathbf{ze}:\mu(\mathbf{ze})>0}\,\sum_{\mathbf{c}\in\mathcal{A}_{\mathbf{ze}}}\sigma(\mathbf{c}\lvert\mathbf{ze})\sigma(\mathbf{ze})\le \sum_
{\mathbf{ze}:\mu(\mathbf{ze})>0}\,\sum_{\mathbf{c}\in\mathcal{A}_{\mathbf{ze}}}M_{\mathbf{ze}}^{\sigma}\sigma(\mathbf{ze})\nonumber\\
    &=  \sum_
{\mathbf{ze}:\mu(\mathbf{ze})>0}\abs{\mathcal{A}_{\mathbf{ze}}}M_{\mathbf{ze}}^{\sigma}\sigma(\mathbf{ze})\le\frac{1}{\gamma}\sum_{\mathbf{ze}}M_{\mathbf{ze}}^{\sigma}\sigma(\mathbf{ze})=\frac{\bar{M}_{\sigma}}{\gamma}.
\end{align}
Using~\eqref{thm3_3} and~\eqref{thm3_4} we obtain $\bar{M}_{\sigma}\ge\gamma(1-\epsilon_{\mathrm a}-2\epsilon_{\mathrm s})$ from which~\eqref{thm_AEP_ineq} follows using the definition of smooth average conditional min-entropy.
\end{proof}

Having shown that the per-trial min-entropy generated by an IID attack is asymptotically bounded by the conditional Shannon entropy, we give the following definition of an \emph{optimal} attack.

\begin{mydef}[Optimal IID Attack]\label{def:optimal_IID_attack}
The distribution $\mu(\mathbf{CZE})$ of the sequence of random variables $\mathbf{C,Z,E}$ is an optimal IID attack strategy if $\mu$ is obtained through an IID attack based on a single-trial distribution $\omega$ whose conditional Shannon entropy achieves the infimum defined below:
\begin{equation}\label{e:def_IIDattack}
    h_{\min}(\rho) \coloneqq \inf_{\omega(CZE)\in\Sigma_{\mathrm{E}}^{\rho}}\mathbb H_{\omega}(C\lvert ZE)
\end{equation}
\end{mydef}

Additional motivation for naming the attack of Definition~\ref{def:optimal_IID_attack} \textit{optimal} is provided by later results in this section, which show that the adversary must generate \textit{at least} $h_{\min}(\rho)$ of per-trial conditional min-entropy asymptotically with any attack that replicates the observed distribution $\rho$. 

In the theorem that follows, we formalise the claim that the infimum in \eqref{e:def_IIDattack} is achieved. This theorem corresponds to Theorem 43 in \cite{PhysRevResearch.2.013016}; in comparison, the  comprehensive proof provided here explicitly works out more of the steps. Crucially, this explicit approach also allowed us to provide an improvement upon the result of Theorem 43 in~\cite{PhysRevResearch.2.013016}, decreasing the required value of $ \abs{\mathcal E}$ by one, thereby better characterising the adversary's optimal attack. Results in Section~\ref{s:robustandoptimal} will illustrate that no further improvement, i.e., a decrease in $\abs{\mathcal E}$, is possible.

\begin{theorem}\label{thm_hmin_achieved} Suppose $\Pi$ is closed and equal to the convex hull of its extreme points. Then there is a distribution $\mu(CZE)\in\Sigma_{\mathrm E}^{\rho}$ with $\abs{\mathcal{E}}=1+\dim \Pi$ such that 
$\mathbb H_{\mu}(C\lvert ZE)=h_{\min}(\rho)$.
\end{theorem}

\noindent \textit{Proof.} See Appendix \ref{a:convex}.

\medskip

Theorem~\ref{thm_hmin_achieved}, in conjunction with the bound in Theorem~\ref{thm_AEP}, sets a benchmark for how well the adversary can do with an IID attack that replicates the observed distribution $\rho(CZ)$. Specifically, the adversary's goal is to minimise the amount of per trial conditional min-entropy, and this shows there exists a strategy to replicate the observed statistics while conceding no more min-entropy per trial than $h_{\min} (\rho)$, asymptotically.

\subsection{Optimal PEFs}We now show that PEFs can asymptotically certify a min-entropy of $h_{\min} (\rho)$ per trial from an observed distribution $\rho$ . This is notable since it shows that an IID attack can be asymptotically optimal: since the PEF method certifies the presence of $h_{\min}(\rho)$ min-entropy per trial against \textit{any} attack, this means no attack can generate observed statistics consistent with $\rho$ while conceding a smaller amount of randomness. This furthermore demonstrates that there is nothing to be gained (asymptotically) by the adversary employing a more sophisticated memory-based attack, since the PEF method allows for the possibility of memory attacks. Conversely, the below results  show that the PEF method is asymptotically optimal: no (correct) method can certify more min-entropy per trial from $\rho$ than the amount that is present in an explicit attack.

To formalise and prove these claims, we use the following technical tool, called an ``entropy estimator" as in \cite{PhysRevResearch.2.033465}. 

\begin{mydef}[Entropy Estimator]\label{def:EntropyEstimator}
An entropy estimator of the model $\Pi$ is a function $K(CZ)$ of the random variables $C,Z$ such that $\mathbb{E}_{\sigma}[K(CZ)]\le\mathbb{E}_{\sigma}[-\log_{2}(\sigma(C\lvert Z))]$ holds for all $\sigma(CZ)\in\Pi$.
\end{mydef}
Given an entropy estimator $K(CZ)$, we say that its \emph{entropy estimate} at a distribution $\sigma(CZ)$ is $\mathbb{E}_{\sigma}[K(CZ)]$. We will see below that an entropy estimator can be used to construct PEFs certifying per-trial min-entropy arbitrarily close to its entropy estimate, underlying the significance of the following result:

\begin{theorem}\label{t:est}
Suppose $\Pi$ satisfies the conditions of Theorem~\ref{thm_hmin_achieved} and $\rho$ is in the interior of $\Pi$. Then there exists an entropy estimator whose entropy estimate at $\rho$ is equal to $h_{\min}(\rho)$.
\end{theorem}

\noindent \textit{Proof.} See Appendix~\ref{a:convex}.

\medskip

\noindent The assumption that $\rho$ is in the interior of $\Pi$ will generally hold if $\rho$ is estimated from real data, as the boundary of $\Pi$ is a measure zero set. If the assumption is removed, a weaker version of the theorem can still be obtained, which is discussed in the proof in Appendix~\ref{a:convex}.

The entropy estimator $K(CZ)$ whose existence is guaranteed by the above theorem can be used to show the existence of a family of PEFs that can get arbitrarily close to certifying $h_{\min}(\rho)$ amount of per-trial min-entropy. However, for a precise formulation of this claim we need a way to measure the asymptotic rate of min-entropy using PEFs. Recall from~\eqref{e:lb_min_ent_2} that we can lower-bound the per-trial min-entropy certified by a PEF as: 
\begin{equation}\label{e:pertrialPEF}
\frac{1}{n}\mathbb H_{\infty,\mu}^{\mathsf{avg},\epsilon}(\mathbf{C}\lvert\mathbf{ZE};\mathsf S)\ge \frac{1}{n}\left(1+\frac{1}{\beta}\right)\log_{2}(\kappa)-\frac{1}{n}\log_{2}(p).
\end{equation}
As in~\cite{PhysRevResearch.2.013016}, we ignore the $\log_2(\kappa)$ term in the asymptotic regime, as the completeness parameter $\kappa$ can be thought of as a ``reasonable" lower bound on the probability that the protocol does not abort, a type of error parameter that one might try to decrease somewhat for longer experiments but not at the exponential decay rate required to make this term asymptotically significant. Focusing then on the $-(1/n)\log_2(p)$ term, recall that success of the protocol is determined by the occurrence of the event $\mathsf S \coloneqq \Big\{\Big(\epsilon\prod_{i=1}^{n}F_{i}\Big)^{-1/\beta}\le p\Big\}$, the inequality in which can be expressed equivalently as:
\begin{equation*}
    \frac{1}{n\beta}\sum_{i=1}^{n}\log_{2}(F_{i}) + \frac{1}{n\beta}\log_{2}(\epsilon) \ge \textcolor{blue}{-}\frac{1}{n}\log_{2}(p). 
\end{equation*}
The expression on the left hand side of the above inequality is the negative base-2 logarithm of the upper bound on $\mu_{\mathbf{e}}(\mathbf{C}\lvert\mathbf{Z})$ for each $\mathbf{e}\in\mathcal{E}^{n}$ (refer to~\eqref{eq_PEF_def} and  the comments following Corollary \ref{coro1}) and so is a rough measure of the amount of randomness, up to an error probability of $\epsilon$, present in the outcome data. More concretely, since $p$ will be chosen to make $-(1/n)\log(p)$ as large as reasonably possible to optimise min-entropy certified by \eqref{e:pertrialPEF}, the anticipated value of the left hand side quantity can be used as a measure of certifiable randomness. For a stable experiment (i.e., one with each trial having the same distribution $\sigma$ belonging to the same model $\Pi$), the quantity $(1/n)\sum_{i=1}^{n}\log_{2}(F_{i})/\beta$ approaches $\mathbb{E}_{\sigma}[\log_{2}(F(CZ))]/\beta$ in the limit $n\to\infty$, while the term $(1/n\beta)\log_2(\epsilon)$ goes to zero for any fixed value of $\beta$ and $\epsilon$. Hence we introduce the following quantity as a measure of per-trial min-entropy certified by a PEF.

\begin{mydef}[Log-Prob Rate] The log-prob rate of a PEF $F(CZ)$ with power ${\beta}$ at a distribution $\rho(CZ)$ is defined as $\mathcal{O}_{\rho}(F;\beta)=\mathbb{E}_{\rho}[\log_{2}(F(CZ))]/\beta$.
\end{mydef}

We say that a PEF certifies randomness at a distribution $\rho$ if the quantity $\mathcal{O}_{\rho}(F;\beta)$ is positive. We note that  this definition is consistent with our expectation that only non-local distributions allow the certification of randomness, as the log-prob rate for a local distribution is a non-positive number, i.e., $\mathcal{O}_{\sigma_{\mathrm{L}}}(F;\beta)\le 0$: A local behaviour is a convex mixture of (finitely many) local deterministic behaviours $\sigma_{\mathrm{LD}}(C\lvert Z)$. Hence, with a fixed settings distribution $\pi(z)>0$, the defining condition $\mathbb{E}_{\sigma}[F(CZ)\sigma(C\lvert Z)^{\beta}]\le 1$ of a PEF for a distribution defined as $\sigma(cz)=\sigma_{\mathrm{LD}}(c\lvert z)\pi(z)$, for all $c,z$, is equivalently expressed as $\mathbb{E}_{\sigma}[F(CZ)]\le 1$, since $\sigma_{\mathrm{LD}}(c\lvert z)$ is either 0 or 1 for all $c,z$. Due to the concavity of $\log$ function, we then have $\mathbb{E}_{\sigma}[\log_{2}(F(CZ))]\le\log_{2}(\mathbb{E}_{\sigma}[F(CZ)])\le 0$ using Jensen's inequality. Hence, no device-independent randomness can be certified at a local-realistic distribution.

\begin{theorem}\label{thm:EntEst_PEF}
Given an entropy estimator $K(CZ)$ and an observed distribution $\rho(CZ)$, for any $\epsilon \in (0,1/2)$ there is a PEF whose log-prob rate at $\rho$ is greater than $\mathbb{E}_{\rho}[K(CZ)]-\epsilon$.
\end{theorem}
\noindent Our proof follows the general approach of Theorem 41 in \cite{PhysRevResearch.2.033465}, though we are able to shorten the argument. 
\begin{proof}
Given an entropy estimator $K(CZ)$ and $\epsilon \in (0,1/2)$ from the statement of the theorem, for any  $\gamma> 0$ we can define a function 
\begin{equation}\label{e:Fdef}
F(CZ) = 2^{(K(CZ)-\epsilon)\gamma}
\end{equation}
We will show that there exists a (small) positive value of $\gamma$ for which $F(CZ)$ is a PEF with power $\beta=\gamma$; the asymptotic log-prob rate of this PEF at $\rho$ will then be $\mathbb E_\rho [\log_{2}(F(CZ))]/\beta = \mathbb E_\rho[K(CZ)] - \epsilon$
as desired. So our task is to find a value of $\gamma$ such that the following inequality holds for all $\sigma \in \Pi$: 
$$
\mathbb E_\sigma [F(CZ) \sigma (C|Z)^\gamma] \le 1
$$
We study the left side of the above expression as a function of $\gamma$; specifically, define a function 
\begin{equation*}
f_\sigma (\gamma) = \mathbb{E}_{\sigma}[F(CZ) \sigma(C|Z)^{\gamma}]
=\sum_{c,z:\sigma(cz)> 0} \left[2^{K(cz)-\epsilon} \sigma(c|z)\right]^\gamma \sigma(cz)\\
\end{equation*}
which is, for any fixed choice of $\sigma$ and $K(CZ)$, a convex combination of positive constants raised to the power of $\gamma$ and so is infinitely differentiable at all $\gamma \in \mathds{R}$. (Note that we never encounter the problematic form $0^0$ because the argument of $[\cdot]^\gamma$ will always be strictly positive, as the sum defining $f_\sigma$ extends only over values of $c,z$ for which $\sigma(cz)$ is positive, and hence $\sigma(c|z)>0$.) We can thus Taylor-expand $f_\sigma$ about $\gamma = 0$, obtaining via the Lagrange remainder theorem
that for any positive $\gamma$, there exists a $k\in (0,\gamma)$ making the following equality hold:
\begin{equation}\label{e:taylorexpand}
f_\sigma(\gamma) = f_\sigma(0) + f'_\sigma(0)\gamma + \frac{f''_\sigma(k)}{2}\gamma^2
\end{equation}
The first term in the expansion satisfies $f_\sigma(0)=\sum_{cz} 1\cdot \sigma(cz) =1$. The coefficient of $\gamma$ in~\eqref{e:taylorexpand} satisfies:
\begin{align*}
f'_\sigma(0) &=\sum_{c,z:\sigma(cz)> 0} \left[2^{K(cz)-\epsilon} \sigma(c|z)\right]^0 \sigma(cz)\ln\left[2^{K(cz)-\epsilon} \sigma(c|z)\right]\\
&= \sum_{c,z:\sigma(cz)>0}\sigma(cz) [K(cz) - \epsilon +\log_{2}(\sigma(c|z))]\ln(2)\\
&= \ln(2)\left(\mathbb E_\sigma [K(CZ)] - \mathbb E[-\log_2 (\sigma(c|z))] - \epsilon\right) \le -\epsilon\ln(2)
\end{align*}
where the inequality follows from the condition $\mathbb E_\sigma [K(CZ)] \le \mathbb E_{\sigma}[-\log_2(\sigma(C\lvert Z))]$ in the definition of an entropy estimator. Hence \eqref{e:taylorexpand} yields

\begin{equation}\label{e:taylorineq}
f_\sigma(\gamma) \le 1 - \epsilon\gamma\ln(2) + \frac{f''_\sigma(k)}{2}\gamma^2
\end{equation}
for some $k\in (0,\gamma)$. Now, given a fixed $\gamma$, $k$  may be different in \eqref{e:taylorineq} for different choices of $\sigma$; however, it must always lie in the interval $(0,\gamma)$, so if we can show that there is a choice of $\gamma$ such that for \textit{any} $\sigma$ the following inequality holds for all $k\in(0,\gamma)$
\begin{equation}\label{e:taylorgoal}
\frac{f''_\sigma(k)}{2}\gamma^2\le \epsilon\gamma\ln(2)
\end{equation}
then for that value of $\gamma$, we will know that $F(CZ)$ as defined in \eqref{e:Fdef} is a valid PEF satisfying the conditions of the theorem. To find the needed value of $\gamma$ making \eqref{e:taylorgoal} hold and complete the proof, we calculate
\begin{align*}
f''_{\sigma}(k) &= \ln^{2}(2)\sum_{c,z:\sigma(cz)>0}\left[2^{K(cz)-\epsilon}\sigma(c\lvert z)\right]^{k}\left[\log_{2}\left(2^{K(cz)-\epsilon}\sigma(c\lvert z)\right)\right]^{2}\sigma(cz) \\
&\le \ln^{2}(2)M^{k}\sum_{cz:\sigma(cz)>0}\sigma(c\lvert z)^{k+1}\left[K(cz)-\epsilon+\log_{2}(\sigma(c\lvert z))\right]^{2}\sigma(z)
\end{align*}
where $M=\max_{cz}2^{K(cz)}$. We now assert that each quantity $\sigma(c|z)^{k+1}\left[K(cz) - \epsilon + \log_{2}(\sigma(c|z))\right]^2$ is bounded above by a constant $N_{cz}$ for all $k>0$, and $N_{cz}$ is independent of $\sigma$. This follows because for any fixed choice of $c$ and $z$, this quantity is strictly smaller than the expression $g_{cz}(x) = x\left[K(cz) - \epsilon + \log_{2}(x)\right]^2$ for the choice of $x=\sigma(c|z) \in (0,1]$ (note that since $\sigma(c\lvert z)\in(0,1]$, $\sigma(c|z)^{k+1} \le \sigma (c|z)$ holds for any $k>0$). Then two applications of l'H\^{o}pital's rule demonstrate that $\lim_{x\to 0} g_{cz}(x)$ exists and so $g_{cz}$ can be extended to a continuous function on $[0,1]$ where it has a maximum by the extreme value theorem.\footnote[2]{Invocation of the extreme value theorem, rather than computing an explicit bound, is what primarily allows us to shorten the proof compared to the argument proving Theorem 41 in~\cite{PhysRevResearch.2.033465}.} 
Referring to this maximum as $N_{cz}$ and letting $N = \max_{cz} N_{cz}$, we get the desired bound as shown below.

\begin{equation}\label{eq:bound_f''}
f''_\sigma(k)\le \ln^{2}(2)M^k\sum_{z:\sigma(z)>0}\sigma(z)\sum_{c:\sigma(c,z)>0} N
\le \ln(2)M^k\sum_{z:\sigma(z)>0}\sigma(z) \abs{\mathcal C} N = \ln(2)M^k\abs{\mathcal{C}}N. 
\end{equation}
This shows that if $M^k\gamma \le 2\epsilon/\abs{\mathcal C} N$ holds, then \eqref{e:taylorgoal} holds, from which it follows that a sufficiently small choice of $\gamma>0$ makes~\eqref{e:taylorgoal} hold for all $k\in (0,\gamma)$.
\end{proof}

The combination of Theorem~\ref{t:est}, which shows the existence of an entropy estimator with entropy estimate $h_{\min} (\rho)$, and Theorem~\ref{thm:EntEst_PEF}, which enables the construction of a family of PEFs with log-prob rate arbitrarily close to this entropy estimate, demonstrates of the asymptotic optimality of the PEF method.

\subsection{Robustness of PEFs}\label{s:Robustness}

We want to consider a question not considered in the previous PEF papers: can a PEF optimised for $\rho(CZ)$ certify randomness for a distribution different from $\rho$, where the difference is measured in terms of the total variation distance between them; in other words, how \textit{robust} is the PEF? We will see in the next section that in the $(2,2,2)$ Bell scenario, for any behaviour corresponding to $\rho$ violating the CHSH-Bell inequality, PEFs can be (up to any desired $\epsilon$-tolerance) asymptotically optimal in terms of log prob rate at $\rho$ while also generating randomness at a positive rate for any behaviour (corresponding to a distribution of outcomes and settings) that violates the CHSH-Bell inequality by a fixed positive amount, which can be chosen to be as small as desired.

The following theorem gives a useful sufficient condition for a distribution different from $\rho$ to have positive log-prob rate, and demonstrates that any nontrivial (i.e., non-constant) PEF will have at least some degree of robustness.

\begin{theorem}\label{thm:robustPEF}
Let $F(CZ)=G(CZ)^{\beta}$ be a non-constant positive PEF with power $\beta>0$ for $\Pi$. The log-prob rate $\mathcal{O}_{\sigma}(F;\beta)$ at a distribution $\sigma(CZ)\in\Pi$ is related to the log-prob rate $\mathcal{O}_{\rho}(F;\beta)$ at $\rho(CZ)\in\Pi$ and the total variation distance between $\rho$ and $\sigma$ as
\begin{equation}\label{eq5}
\abs{\mathcal{O}_{\rho}(F;\beta)-\mathcal{O}_{\sigma}(F;\beta)} \le (L-l)d_\mathrm{TV}(\rho,\sigma),
\end{equation}
where $L=\max_{cz}\log_{2}(G(cz))$ and $l=\min_{cz}\log_{2}(G(cz))$. Consequently, assuming that $\mathcal{O}_{\rho}(F;\beta)$ is positive, the following upper bound on the total variation distance between $\rho(CZ)$ and $\sigma(CZ)$ is a sufficient condition for $F$ to have a positive log-prob rate at $\sigma(CZ)$
\begin{equation}\label{e:suff_cond_robustPEF}
d_\mathrm{TV}(\rho,\sigma) < \mathbb{E}_{\rho}[\log_{2}(G)]/(L-l).
\end{equation} 
\end{theorem}

\begin{proof}
Using the definition of log-prob rate at a given distribution we have
\begin{align*}
&\abs{\mathcal{O}_{\rho}(F;\beta)-\mathcal{O}_{\sigma}(F;\beta)}
= \abs{\sum_{cz}\frac{1}{\beta}\left[\log_{2}(G(cz)^{\beta})\big(\rho(cz)-\sigma(cz)\big)\right]}\\
&=\abs{\sum_{cz}\left(\log_{2}(G(cz)) + \frac{L+l}{2}-\frac{L+l}{2}\right)(\rho(cz)-\sigma(cz))}\\
&=\abs{\sum_{cz}\left(\log_{2}(G(cz))-\frac{L+l}{2}\right)\big(\rho(cz)-\sigma(cz)\big)+\frac{L+l}{2}\sum_{cz}\big(\rho(cz)-\sigma(cz)\big)}\\
&\le \sum_{cz}\abs{\log_{2}(G(cz))-\frac{L+l}{2}}\abs{\rho(cz)-\sigma(cz)}\\
&\le (L-l)\frac{1}{2}\sum_{cz}\abs{\rho(cz)-\sigma(cz)}=(L-l)d_\mathrm{TV}(\rho,\sigma)
\end{align*}
Hence, we have
\begin{equation*}
\mathcal{O}_{\rho}(F;\beta) - (L-l)d_\mathrm{TV}(\rho,\sigma) \le \mathcal{O}_{\sigma}(F;\beta) \le \mathcal{O}_{\rho}(F;\beta) + (L-l)d_\mathrm{TV}(\rho,\sigma).
\end{equation*}
Assuming that $\mathcal{O}_{\rho}(F;\beta)$ is positive, a sufficient condition for $\mathcal{O}_{\sigma}(F;\beta)$ to be positive is $\mathcal{O}_{\rho}(F;\beta)>\abs{L-l}d_{\mathrm{TV}}(\rho,\sigma)$, or equivalently, the following bound on $d_{\mathrm{TV}}(\rho,\sigma)$:
\begin{equation*}
d_\mathrm{TV}(\rho,\sigma) < \mathcal{O}_{\rho}(F;\beta)/(L-l)=\mathbb{E}_{\rho}[\log_{2}(G)]/(L-l).
\end{equation*}
\end{proof}

\noindent We will see in Section \ref{s:robustandoptimal} that the bound \eqref{eq5} can be saturated, and so is tight.

\section{Application to the (2,2,2) Bell scenario} \label{s:application}

Here, we explore the application of the results of the previous section to the (2,2,2) Bell scenario (that of two parties, two measurement settings, and two outcomes). First, working within the trial model of no-signalling distributions $\Pi_{\mathrm{NS}}$,  we show that PEFs can be simultaneously  asymptotically optimal and robust by means of an explicit construction of a sequence of PEFs that approaches the optimal log-prob rate for the target distribution while simultaneously generating randomness at a positive rate for any other distribution violating the CHSH inequality.

In the course of this exercise, we will observe that the optimal adversarial attack---one generating the observed statistics (consistent with an expected trial distribution $\rho$) while asymptotically yielding $h_{\min}(\rho)$ amount of per-trial randomness---is always achieved through a single-trial distribution that marginalises to $\rho$ through a convex combination of a single extremal no-signalling non-local distribution and a local realistic distribution (which itself consists of  a convex mixture of up to eight extremal local deterministic distributions). This is a notable feature, revealing that the adversary never needs to prepare more than one non-local distribution to simulate the observed distribution with as little min-entropy as possible. Later in this section, we explore the potential for generalisation of this feature to the (2,2,2) scenario restricted to quantum distributions ($\Pi_{\mathrm{Q}}$); if true, this would be an important finding, outlining the optimal approach of a (more realistic) quantum-limited adversary attacking the PEF protocol. The general observation that preparing a single non-local state is preferable to preparing multiple underlies the significance of the answer to this question. We find some evidence that the feature---only requiring one extremal non-local distribution in the convex combination attack---may hold for the $\Pi_{\mathrm{Q}}$ in the (2,2,2) Bell scenario, but this may be a difficult question to resolve due to the complicated geometry of the quantum set. We also explore possible generalisations of this feature to no-signalling trial models for $(n,m,k)$ Bell scenarios where $n$, $m$, or $k$ exceed 2, and find that it \textit{does not} hold in any of these cases---so the question of whether this holds in a given Bell scenario and trial model is non-trivial in general.

We begin with a brief review of the (2,2,2) Bell scenario and some features of the set $\Pi_{\mathrm{NS}}$ of no-signalling distributions in this scenario.

\subsection{A brief review of the (2,2,2) Bell scenario}\label{e:brief_review} The (2,2,2) Bell scenario is the minimal Bell scenario, comprising of two spatially separated parties Alice and Bob, each having two measurement settings and two possible outcomes corresponding to each setting. The measurement settings for Alice and Bob are represented by the RVs $X,Y$ realising values $x,y\in\{0,1\}$ and the measurement outcomes are represented by the RVs $A,B$ realising values $a,b\in\{0,1\}$. With $\sigma_{\mathrm s}(XY)$ representing a fixed settings distribution, we refer to the sets $\Pi_{\mathrm{NS}},\Pi_{\mathrm{Q}}\text{ and }\Pi_{\mathrm{L}}$ as no-signalling, quantum and local models, respectively, when they comprise of distributions $\mu(ABXY)\coloneqq\mu(AB\lvert XY)\sigma_{\mathrm s}(XY)$, where the conditional probabilities $\mu(AB\lvert XY)$, referred to as behaviours, are constrained by the no-signalling, quantum and local realism principle, respectively. Henceforth, all distributions $\mu(ABXY)$ belonging to a model are defined as $\mu(ABXY)\coloneqq\mu(AB\lvert XY)\sigma_{\mathrm s}(XY)$, and we associate a model with its constituent behaviour $\mu(AB\lvert XY)$ or distribution $\mu(ABXY)$, indistinctively, since the settings distribution is fixed. Recall that the model $\Pi_{\mathrm{NS}}$ is a polytope, the extremal points of which consist of the behaviours $\mu_{\mathrm{extr}}(AB\lvert XY)\equiv \big\{\mu_{\mathrm{extr}}(ab\lvert xy)\colon a,b,\linebreak[0]x,y\in\{0,1\}\big\}$ defined below.
\begin{align}
    \mu_{\text{PR}}^{\alpha\beta\gamma}(ab\lvert xy) &\coloneqq\begin{cases} \frac{1}{2} &: \quad a\oplus b=x y\oplus\alpha x\oplus\beta y\oplus\gamma\\
    0 &: \quad \text{otherwise}\end{cases}\label{eq_PRbox}\\
    \mu_{\text{LD}}^{\alpha\beta\gamma\delta}(ab\lvert xy) &\coloneqq \begin{cases} 1 &: \quad a=\alpha x\oplus\beta, b=\gamma y\oplus\delta\\
    0 &: \quad \text{otherwise}\end{cases}\label{eq_LDbox}
\end{align}
with $\alpha,\beta,\gamma,\delta\in\{0,1\}$ and $\oplus$ denotes addition modulo $2$. \eqref{eq_PRbox} and~\eqref{eq_LDbox} are known as the Popescu-Rohrlich (PR) behaviours~\cite{Popescu1994} and the Local Deterministic (LD) behaviours, respectively. The CHSH-Bell inequalities shown below are known to be the only non-trivial facet inequalities delimiting the local polytope which is the convex hull of the LD behaviours~\cite{PhysRevLett.48.291}. Corresponding to each choice of $\alpha,\beta,\gamma\in\{0,1\}$, the inequalities represent a version of the canonical CHSH-Bell inequality.  
\begin{equation}\label{eq_Bell}
    B^{\alpha\beta\gamma} \coloneqq (-1)^{\gamma}E_{00}+(-1)^{\beta+\gamma}E_{01}+(-1)^{\alpha+\gamma}E_{10}+(-1)^{\alpha+\beta+\gamma+1}E_{11}\le 2,
\end{equation}
where $E_{xy}\coloneqq \sum_{a,b=0}^{1}(-1)^{a+b}\mu(ab\lvert xy)$ for $x,y\in\{0,1\}$. The non-local algebraic maximum for the expression $B^{\alpha\beta\gamma}$ is 4. The local maximum is obtained by eight $\mu_{\mathrm{LD}}^{\alpha\beta\gamma\delta}(AB\lvert XY)$ behaviours for each $B^{\alpha\beta\gamma}$. The sets $\mathsf{LD}_{i},\,i\in\{1,2,\ldots,8\}$, each comprising of eight LD behaviours saturating---i.e., achieving the value of 2---exactly one $B^{\alpha\beta\gamma}$ are shown in Table~\ref{tab:1PR_8LD}. A result proven in~\cite{Bierhorst_2016} (see Theorems 2.1 and 2.2 therein) states that any behaviour violating~\eqref{eq_Bell} can be represented as a convex combination of one PR box achieving the non-local maximum for $B^{\alpha\beta\gamma}$ and (up to) eight LD behaviours of the corresponding $\mathsf{LD}_{i}$ set saturating it. In fact, the geometry of the no-signalling polytope in this Bell scenario is such that there is a one-to-one correspondence between the non-local no-signalling extremal points, the PR boxes, in~\eqref{eq_PRbox} and the non-trivial facets of the local polytope described by~\eqref{eq_Bell}, with exactly one extremal point violating it up to the algebraic maximum of $4$ for each choice of $(\alpha,\beta,\gamma)\in\{0,1\}^{3}$. Hence, any non-local behaviour---that violates a given version of the CHSH-Bell inequality---is contained in a \emph{non-local $8$-simplex} whose vertices are the one PR box that maximally violates that particular version and the eight LD behaviours that saturate it. Recall that a $p$-simplex is a $p$-dimensional polytope which is the convex hull of its $p+1$ vertices. More formally, if the set $C\coloneqq\{\vec{a}_{0},\vec{a}_{1},\ldots,\vec{a}_{p}\}\subset\mathds{R}^{n}$ of $p+1$ points are affinely independent, then the $p$-simplex determined by them is the following set of points:
\begin{equation*}
    \Delta^{p} \coloneqq \left\{\sum_{k=0}^{p}\theta_{k}\vec{a}_{k}\,\bigg|\, \sum_{k=0}^{p}\theta_{k}=1,\,\theta_{k}\ge 0\text{ for }k=0,1,\ldots,p \right\}.
\end{equation*}
The affine independence condition means that the only admissible choice of $\theta_{k}\in\mathds{R}$ such that $\sum_{k=0}^{p}\theta_{k}\vec{a}_{k}=\vec{0}$ and $\sum_{k=0}^{p}\theta_{k}=0$ are satisfied is $\theta_{k}=0$ for all $k$; this holds if and only if the vectors $\vec{a}_{k}-\vec{a}_{0}$ are linearly independent for $k=1,2,\ldots,p$. 

One can check that the PR box that achieves the non-local maximum for a given version of the CHSH-Bell expression $B^{\alpha\beta\gamma}$ and the eight LD behaviours that achieve the local maximum for it are affinely independent. Since $a,b,x,y\in\{0,1\}$ and $|\{0,1\}^{4}|=16$, we can represent the behaviours $\mu(ab\lvert xy)$ in this Bell scenario as vectors $\vec{\mu}\in\mathds{R}^{16}$ as shown in Table~\ref{tab:Simplex_PR1}. Then the affine independence is apparent: letting the PR box behaviour be $\vec{a}_0$ and the LD behaviours be the other $\vec{a}_k$, each $\vec{a}_k-\vec{a}_0$ term has a unique column where it contains a ``1" while all of the other terms contain ``0", ensuring linear independence.

\begin{table}[H]
    \centering
    \begin{tabular}{c|cccc|cccc|cccc|cccc|}
    \cline{2-17}
     & \multicolumn{16}{c|}{$xy$} \\
     & \multicolumn{4}{c}{$00$} & \multicolumn{4}{c}{$01$} & \multicolumn{4}{c}{$10$} & \multicolumn{4}{c|}{$11$} \\
     \cline{2-17}
     & \multicolumn{4}{c|}{$ab$} & \multicolumn{4}{c|}{$ab$} & \multicolumn{4}{c|}{$ab$} & \multicolumn{4}{c|}{$ab$} \\
     & $00$ & $01$ & $10$ & $11$ & $00$ & $01$ & $10$ & $11$ & $00$ & $01$ & $10$ & $11$ & $00$ & $01$ & $10$ & $11$ \\
     \cline{2-17}
   $\vec{\mu}_{\mathrm{PR},1}$ & $1/2$ & $0$ & $0$ & $1/2$ & $1/2$ & $0$ & $0$ & $1/2$ & $1/2$ & $0$ & $0$ & $1/2$ & $0$ & $1/2$ & $1/2$ & $0$ \\
   $\vec{\mu}_{\mathrm{LD},1}$ & $1$ & $0$ & $0$ & $0$ & $1$ & $0$ & $0$ & $0$ & $1$ & $0$ & $0$ & $0$ & $1$ & $0$ & $0$ & $0$ \\
   $\vec{\mu}_{\mathrm{LD},2}$ & $0$ & $0$ & $0$ & $1$ & $0$ & $0$ & $0$ & $1$ & $0$ & $0$ & $0$ & $1$ & $0$ & $0$ & $0$ & $1$ \\
   $\vec{\mu}_{\mathrm{LD},3}$ & $1$ & $0$ & $0$ & $0$ & $0$ & $1$ & $0$ & $0$ & $1$ & $0$ & $0$ & $0$ & $0$ & $1$ & $0$ & $0$ \\
   $\vec{\mu}_{\mathrm{LD},4}$ & $0$ & $0$ & $0$ & $1$ & $0$ & $0$ & $1$ & $0$ & $0$ & $0$ & $0$ & $1$ & $0$ & $0$ & $1$ & $0$ \\
   $\vec{\mu}_{\mathrm{LD},5}$ & $1$ & $0$ & $0$ & $0$ & $1$ & $0$ & $0$ & $0$ & $0$ & $0$ & $1$ & $0$ & $0$ & $0$ & $1$ & $0$ \\
   $\vec{\mu}_{\mathrm{LD},6}$ & $0$ & $0$ & $0$ & $1$ & $0$ & $0$ & $0$ & $1$ & $0$ & $1$ & $0$ & $0$ & $0$ & $1$ & $0$ & $0$ \\
   $\vec{\mu}_{\mathrm{LD},7}$ & $0$ & $1$ & $0$ & $0$ & $1$ & $0$ & $0$ & $0$ & $0$ & $0$ & $0$ & $1$ & $0$ & $0$ & $1$ & $0$ \\
   $\vec{\mu}_{\mathrm{LD},8}$ & $0$ & $0$ & $1$ & $0$ & $0$ & $0$ & $0$ & $1$ & $1$ & $0$ & $0$ & $0$ & $0$ & $1$ & $0$ & $0$ \\
   \cline{2-17}
    \end{tabular}
    \caption{These probability vectors in $\mathds{R}^{16}$ are the PR box $\vec{\mu}_{\mathrm{PR},1}\equiv \mu_{\mathrm{PR}}^{000}$ that achieves the non-local maximum of $4$ and the eight LD behaviours $\vec{\mu}_{\mathrm{LD},1},\ldots,\vec{\mu}_{\mathrm{LD},8}$ that achieve the local maximum of $2$ for the standard CHSH-Bell expression $B^{000}$, with the LD behaviors corresponding to the eight probability tables numbered 1, 4, 5, 8, 9, 12, 14 and 15 in Table A2 of~\cite{Bierhorst_2016}, and also given in the first row of Table \ref{tab:1PR_8LD}. One can verify the affine independence of the nine vectors above by verifying that the eight vectors obtained by subtracting the first vector from the remaining eight are linearly independent.}
    \label{tab:Simplex_PR1}
\end{table}
It is known that a behaviour belonging to $\Pi_{\mathrm{NS}}\setminus\Pi_{\mathrm{L}}$ violates exactly one of the eight CHSH-Bell inequalities. The impossibility of simultaneously violating a specific pair of CHSH-Bell inequalities can be seen as presented in~\cite{Le2023quantumcorrelations}: Suppose a behaviour in $\Pi_{\mathrm{NS}}\setminus\Pi_{\mathrm{L}}$ violates both inequalities corresponding to $(\alpha,\beta,\gamma)=(0,0,0)$ and $(\alpha,\beta,\gamma)=(1,0,0)$, then $E_{00}+E_{01}+E_{10}-E_{11}>2$ and $E_{00}+E_{01}-E_{10}+E_{11}>2$ holds for the same behaviour. Adding these two inequalities we have $2(E_{00}+E_{01})>4$, i.e., $E_{00}+E_{01}>2$, which is not possible to satisfy since the correlations $E_{xy}$ satisfy $\abs{E_{xy}}\le 1$. 

Table~\ref{tab:1PR_8LD} lists the eight versions of the Bell expression $B^{\alpha\beta\gamma}$ and the eight non-local $8$-simplices $\Delta_{\mathrm{PR},i}^{8}$ containing points that violate the corresponding CHSH-Bell inequality. Any non-local no-signalling behaviour ultimately belongs to exactly one such simplex.

\begin{table}[h]
    \centering
    {\renewcommand{\arraystretch}{1.32}%
    \begin{tabular}{|c|c|}
    \hline
    \hline
    $B^{\alpha\beta\gamma}$ & $\Delta_{\mathrm{PR},i}^{8}$ \\
    \hline
    \hline
    $B^{000}$ & $\Delta_{\mathrm{PR},1}^{8}\coloneqq\mathrm{conv}\left\{\mu_{\mathrm{PR}}^{000},\mu_{\mathrm{LD}}^{0000},\mu_{\mathrm{LD}}^{0101},\mu_{\mathrm{LD}}^{0010},\mu_{\mathrm{LD}}^{0111},\mu_{\mathrm{LD}}^{1000},\mu_{\mathrm{LD}}^{1101},\mu_{\mathrm{LD}}^{1011},\mu_{\mathrm{LD}}^{1110}\right\}$\\
    $B^{001}$ & $\Delta_{\mathrm{PR},2}^{8}\coloneqq\mathrm{conv}\left\{\mu_{\mathrm{PR}}^{001},\mu_{\mathrm{LD}}^{0001},\mu_{\mathrm{LD}}^{0011},\mu_{\mathrm{LD}}^{0100},\mu_{\mathrm{LD}}^{0110},\mu_{\mathrm{LD}}^{1001},\mu_{\mathrm{LD}}^{1010},\mu_{\mathrm{LD}}^{1100},\mu_{\mathrm{LD}}^{1111}\right\}$\\
    $B^{010}$ & $\Delta_{\mathrm{PR},3}^{8}\coloneqq\mathrm{conv}\left\{\mu_{\mathrm{PR}}^{010},\mu_{\mathrm{LD}}^{0000},\mu_{\mathrm{LD}}^{0010},\mu_{\mathrm{LD}}^{0101},\mu_{\mathrm{LD}}^{0111},\mu_{\mathrm{LD}}^{1001},\mu_{\mathrm{LD}}^{1010},\mu_{\mathrm{LD}}^{1100},\mu_{\mathrm{LD}}^{1111}\right\}$\\
    $B^{011}$ & $\Delta_{\mathrm{PR},4}^{8}\coloneqq\mathrm{conv}\left\{\mu_{\mathrm{PR}}^{011},\mu_{\mathrm{LD}}^{0001},\mu_{\mathrm{LD}}^{0011},\mu_{\mathrm{LD}}^{0100},\mu_{\mathrm{LD}}^{0110},\mu_{\mathrm{LD}}^{1000},\mu_{\mathrm{LD}}^{1011},\mu_{\mathrm{LD}}^{1101},\mu_{\mathrm{LD}}^{1110}\right\}$\\
    $B^{100}$ & $\Delta_{\mathrm{PR},5}\coloneqq\mathrm{conv}\left\{\mu_{\mathrm{PR}}^{100},\mu_{\mathrm{LD}}^{0000},\mu_{\mathrm{LD}}^{0011},\mu_{\mathrm{LD}}^{0101},\mu_{\mathrm{LD}}^{0110},\mu_{\mathrm{LD}}^{1000},\mu_{\mathrm{LD}}^{1010},\mu_{\mathrm{LD}}^{1101},\mu_{\mathrm{LD}}^{1111}\right\}$\\
    $B^{101}$ & $\Delta_{\mathrm{PR},6}^{8}\coloneqq\mathrm{conv}\left\{\mu_{\mathrm{PR}}^{101},\mu_{\mathrm{LD}}^{0001},\mu_{\mathrm{LD}}^{0010},\mu_{\mathrm{LD}}^{0100},\mu_{\mathrm{LD}}^{0111},\mu_{\mathrm{LD}}^{1001},\mu_{\mathrm{LD}}^{1011},\mu_{\mathrm{LD}}^{1100},\mu_{\mathrm{LD}}^{1110}\right\}$\\
    $B^{110}$ & $\Delta_{\mathrm{PR},7}^{8}\coloneqq\mathrm{conv}\left\{\mu_{\mathrm{PR}}^{110},\mu_{\mathrm{LD}}^{0001},\mu_{\mathrm{LD}}^{0010},\mu_{\mathrm{LD}}^{0100},\mu_{\mathrm{LD}}^{0111},\mu_{\mathrm{LD}}^{1000},\mu_{\mathrm{LD}}^{1010},\mu_{\mathrm{LD}}^{1101},\mu_{\mathrm{LD}}^{1111}\right\}$\\
    $B^{111}$ & $\Delta_{\mathrm{PR},8}^{8}\coloneqq\mathrm{conv}\left\{\mu_{\mathrm{PR}}^{111},\mu_{\mathrm{LD}}^{0000},\mu_{\mathrm{LD}}^{0011},\mu_{\mathrm{LD}}^{0101},\mu_{\mathrm{LD}}^{0110},\mu_{\mathrm{LD}}^{1001},\mu_{\mathrm{LD}}^{1011},\mu_{\mathrm{LD}}^{1100},\mu_{\mathrm{LD}}^{1110}\right\}$\\
    \hline
    \hline
    \end{tabular}}
    \caption{The eight non-local $8$-simplices containing behaviours that violate the corresponding version of the CHSH-Bell inequality. We identify each $8$-simplex $\Delta_{\mathrm{PR},i}^{8}$ with a PR box which solely contributes to the non-locality of the behaviour violating the CHSH-Bell inequality.}
    \label{tab:1PR_8LD}
\end{table}

\subsection{Robust PEFs and optimal adversarial attacks in the (2,2,2) Bell Scenario}
\label{s:robustandoptimal}

We now examine the \textit{robustness} of PEFs that are optimal for an anticipated distribution $\rho$ and a fixed number of planned trials $n$.  We first review how we find optimal PEFs in this scenario. The constrained maximisation routine in~\eqref{eq_PEF_opt} provides a method to find useful PEFs with respect to an anticipated trial distribution, with Lemma~\ref{lemma_conv_PEF} showing that the feasibility constraints in \eqref{eq_PEF_opt}  
 can be restricted to only the distributions corresponding to the eight PR and sixteen LD behaviours (with a fixed settings distribution $\sigma_{\mathrm s}(XY)>0$).  
In practice, the number of trials $n$ will affect the choice of $\beta$ and the PEF that optimises the quantity $\mathbb{E}_{\rho}[(n\log_{2}(F(CZ)) + \log_{2}(\epsilon))/\beta]$, a quantity which (per the discussion surrounding \eqref{eq_PEF_opt}) can be thought of the anticipated amount of raw randomness from running the experiment whose trial distribution is expected to be $\rho$. 
If we divide this quantity by $n$, we arrive at a measure of expected randomness per trial for the optimal PEF at a given value of $\beta$, called the \textit{net log-prob rate}: the function $(\max_{F}\mathcal{O}_{\rho}(F;\beta))+\log_{2}(\epsilon)/n\beta$. Figure~\ref{fig:beta_n_relation} shows a plot of the net log-prob rates corresponding to two different values of $n$, as well as the supremum of the log-prob rate, for $\beta$ varying from $0.001$ to $0.1$ and $\epsilon$ fixed at the value $10^{-4}$. The value of $\beta$, and the corresponding PEF that maximises the curve, is then the best choice for the given planned number of trials $n$.

The plot illustrates some notable features of PEFs.  First, it was proved in Appendix D of \cite{PhysRevA.98.040304} that assuming a stable experiment (with each trial distribution $\rho$) the function $\sup_{F}\mathcal{O}_{\rho}(F;\beta)$ is monotonically non-increasing in $\beta>0$\footnote[2]{The proof that $\beta' < \beta$ implies $\sup_{F}\mathcal{O}_{\rho}(F;\beta')\le \sup_{F}\mathcal{O}_{\rho}(F;\beta)$ is straightforward: write $\beta'=\gamma \beta$ with $0<\gamma<1  $; then for any $F$ in the scope of $\sup_{F}\mathcal{O}_{\rho}(F;\beta)$, it turns out $F^\gamma$ is a PEF with power $\beta'$, for which the equality $\mathcal{O}_{\rho}(F^\gamma;\beta')=\mathcal{O}_{\rho}(F;\beta)$ follows immediately from the definition of log-prob rate -- hence the supremum of log-prob rates cannot be smaller at $\beta'$. $F^\gamma$ is a PEF with power $\beta'$ as $\mathbb E_\rho(F^\gamma \sigma(c|z)^{\beta\gamma})\le \mathbb E_\rho(F \sigma(c|z)^\beta)^\gamma\le 1^\gamma = 1$, with the first inequality holding by Jensen's inequality ($f(x)=x^\gamma$ is concave) and the second because $F$ is a PEF with power $\beta$.} which implies that the global supremum of the log-prob rates $\sup_{\beta>0}\sup_{F}\mathcal{O}_{\rho}(F;\beta)$, for all PEFs with positive powers, is achieved in the limit $\beta\to 0$. We observe this with the top curve. For a fixed $\epsilon$, the net log-prob rate converges upwards to $\sup_{F}\mathcal{O}_{\rho}(F;\beta)$ for each $\beta$ as $n\to 0$ , but for any fixed value of $n$, $\log_2 (\epsilon)/n\beta$ diverges to $-\infty$ as $\beta \to 0$. 
Hence in a finite trial regime the supremum of the log-prob rates (attainable by PEFs with positive powers) is not achieved---the maximum value of the \emph{net} log-prob rate is achieved at a $\beta$ away from $0$.  The general trend is that for a value of $n$ the net log-prob rate achieves a higher value corresponding to a lower value of $\beta$; the net log-prob rate is improved by a reduction in power and an increase in the number of trials. This is observed in Figure \ref{fig:beta_n_relation} for the two choices of $n=1.5 \times 10^5$ and $n=2.4 \times 10^5$. 

The arguments above illustrate how it is necessary to consider a range of $\beta$ values to find the optimal choice. We remark there is an upper limit to the range of $\beta$ values that must be considered: it was noted in in~\cite{PhysRevA.98.040304} (see Appendix F therein) that there exists a certain threshold value $\beta_{\text{th}}^{\text{NS}}$ such that for all $\beta \ge \beta_{\text{th}}^{\text{NS}}$, the optimisation problem in~\eqref{eq_PEF_opt} will return the same PEF independent of the choice of $\beta$, and \cite{PhysRevA.98.040304} cites numerical evidence that this bound is $\beta_{\text{th}}^{\text{NS}} \simeq 0.4151$. The following result, whose proof we give in the appendix, derives this threshold analytically, finding it to have the exact value $\log_2(4/3)$.
\begin{proposition}\label{beta_threshold}
For the set of behaviours $\Pi_{\mathrm{NS}}$, the PEF optimisation in~\eqref{eq_PEF_opt} is independent of the power $\beta$ for $\beta\ge\log_{2}(4/3)$.
\end{proposition}
\begin{proof}
    See Appendix~\ref{s:Proof_beta_threshold}.
\end{proof}

We now ask how optimal PEFs for lower and lower values of $\beta$ (and correspondingly higher values of $n$) compare on the question of \textit{robustness}, in the following sense: can a PEF optimised with respect to a distribution $\rho$ violating the standard CHSH-Bell inequality be used to certify randomness of distributions that are different from $\rho$, provided they violate the same CHSH-Bell inequality? This question is relevant because in practice, the observed experimental distribution will never be exactly the same as the anticipated one, and may be somewhat different depending on many potential factors. Figure~\ref{fig:HeatMap} gives an illustration of the matter of robustness. Comparing the two plots of the log-prob rate for quantum-realisable distributions on the two-dimensional slice (shown in Figure~\ref{fig:SliceAboveCHSHfacet}) above the standard CHSH-Bell facet, we observe that the level set denoting a zero amount of certified randomness in the right hand plot (which corresponds to a lower value of $\beta$ than that on the left) is pushed further down to (almost touching) the standard CHSH-Bell facet.

\begin{figure}[h]
\centering
\includegraphics[scale = 0.7]{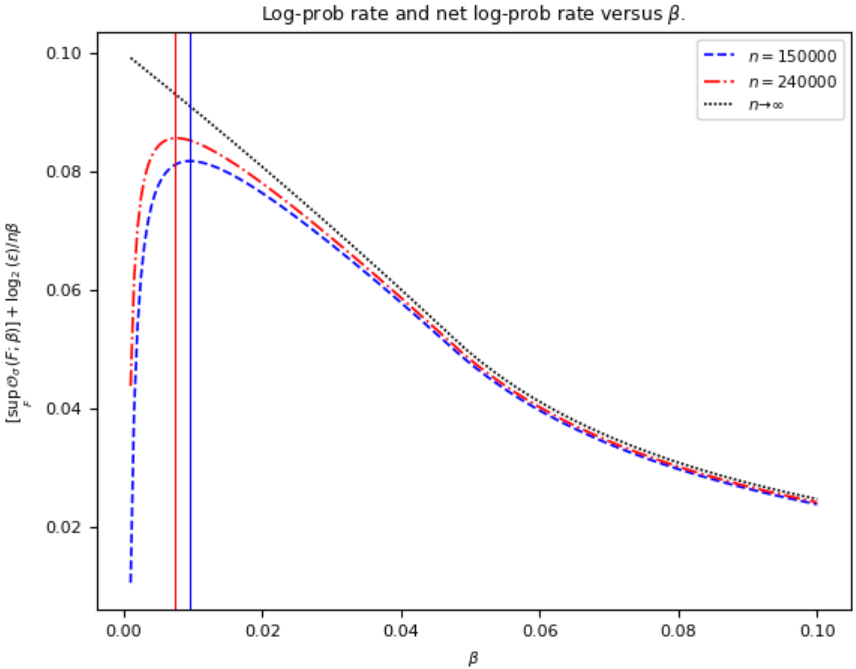}
\caption{A plot showing the net log-prob rates for $n=1.5\times 10^{5}$ (the dashed curve) and $n=2.4\times 10^{5}$ (the dash-dotted curve) with $\epsilon=10^{-4}$ and $\beta$ varying in the interval $(0.001,0.1)$. The dotted curve is the log-prob rate $\sup_{F}\mathcal{O}_{\rho}(F;\beta)$, an upper bound for the net log-prob rate in the limit as $n\to\infty$. We selected 200 equally spaced points in the interval $(0.001,0.1)$ for $\beta$ and performed the maximisation $\max_{F}\mathbb{E}_{\rho}[\log_{2}(F(ABXY))]$ constrained by: (1) the non-negativity of PEFs and (2) the defining condition $\mathbb{E}_{\mu}[F(ABXY)\mu(AB\lvert XY)^{\beta}]\le 1$ at all distributions $\mu$ corresponding to the eight PR and sixteen LD behaviours with a fixed uniform settings distribution $\mu(xy)=1/4$ for all $x,y\in\{0,1\}$. The anticipated distribution $\rho$ used here was the one corresponding to the behaviour given in Table I in~\cite{PhysRevResearch.2.033465}. We observe that the maximum value for the net log-prob rate---indicated by the solid vertical lines---is achieved at a lower value of $\beta$ for a higher value of $n$.}
\label{fig:beta_n_relation}
\end{figure}

\begin{figure}[h]
    \centering
    \includegraphics[scale=0.62]{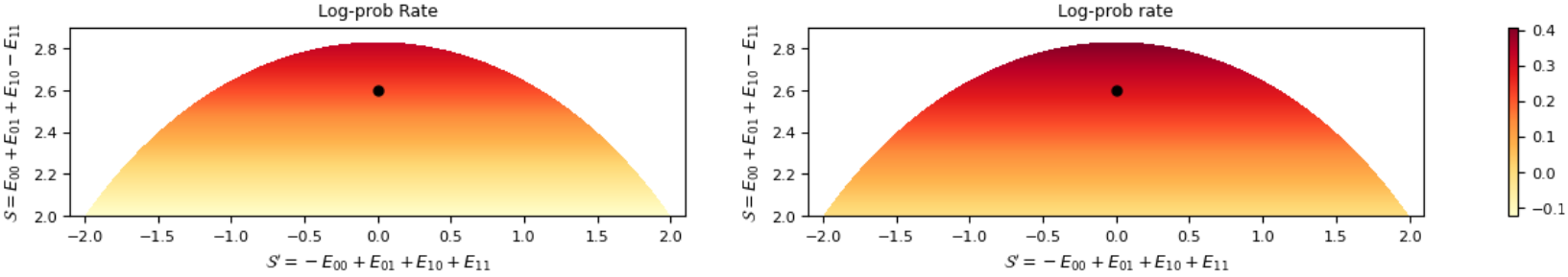}
    \caption{A heat map illustrating the robustness of PEF with log-prob rate as the figure of merit, evaluated for behaviours $\sigma(ab\lvert xy)$ on the two-dimensional slice of the set of quantum behaviours (shown in Figure~\ref{fig:SliceAboveCHSHfacet}) above the standard CHSH-Bell facet. The behaviours on the 2-dimensional slice shown above are parametrised by $\mathcal{S}$ and $\mathcal{S}'$ as shown in~\eqref{e:NSB2Dslice} with the added restrictions $\mathcal{S}^{2}+(\mathcal{S}^{\prime})^{2}\le 8$ and $2\le\mathcal{S}\le 2\sqrt{2},\,-2\le\mathcal{S}^{\prime}\le 2$ (see also Table~\ref{tab:QB2DSlice}). Assuming a uniform distribution for the settings, $\sigma_{\mathrm s}(xy)=1/4$ for all $x,y$, we plot the log-prob rate $\sum_{abxy}[\log_{2}F_{\ast}(abxy)\sigma(ab\lvert xy)\sigma_{\mathrm s}(xy)]/\beta$ for all distributions in the slice. The black dot corresponds to the behaviour (and hence the distribution) with respect to which we perform the PEF optimisation for a fixed $n$ and $\epsilon$ to obtain $F_{\ast}$. The coordinates for the black dot are $(\mathcal{S}^{\prime},\mathcal{S})\equiv(0,2.6)$.\\
    (a) \emph{Figure on the left}: Heat map with $F_{\ast}$ obtained from the PEF optimisation in~\eqref{eq_PEF_opt} with respect to the fixed distribution (corresponding to the black dot in the figures), fixed $n,\,\epsilon$ and $\beta=0.1$. Below $\mathcal{S}\simeq 2.22145$ no device-independent randomness can be certified.\\ 
]    (b) \emph{Figure on the right}: Heat map with $F_{\ast}$ obtained from the PEF optimisation in~\eqref{eq_PEF_opt} with respect to the fixed distribution (corresponding to the black dot in the figures), fixed $n,\,\epsilon$ and $\beta=0.01$. Below $\mathcal{S}\simeq 2.02072$ no device-independent randomness can be certified.}
    \label{fig:HeatMap}
\end{figure}
This suggests that the asymptotic optimality of a PEF need not entail a trade-off with its robustness; indeed we observed that in many cases, as $\beta>0$ assumes smaller and smaller values, the PEF optimised for a fixed $\rho$ violating the standard CHSH-Bell inequality gets more and more robust in the sense that it certifies randomness at a positive rate (asymptotically) for increasingly statistically different $\sigma$.

We show that this is a general feature. To this end, we define a sequence of PEFs that is both asymptotically optimal with respect to the log-prob rate and is asymptotically robust in the sense that given any distribution violating the standard CHSH-Bell inequality, all the PEFs beyond a point in the sequence certify randomness at a positive rate. To construct this PEF sequence, we first define the function $K_{\ast}(ABXY)$ as shown below:
\begin{equation}\label{e:Opt_EntEst}
    K_{\ast}(abxy)\coloneqq 
    \begin{cases}
    1 &:\quad a\oplus b=xy\\
    -3 &:\quad \text{otherwise}
    \end{cases}
\end{equation}

The function defined in~\eqref{e:Opt_EntEst} is an entropy estimator for the distributions in the no-signalling polytope when the settings are equiprobable; i.e., $\sigma_{\mathrm{s}}(xy)=1/4$ for all choices of $x$ and  $y$. To see this, recalling Definition \ref{def:EntropyEstimator} we can check---by direct evaluation---if $K_{\ast}$ satisfies the inequality $\mathbb{E}_{\sigma}[K(CZ)]\le\mathbb{E}_{\sigma}[-\log_{2}(\sigma(C\lvert Z))]$ when $\sigma$ is each of the extremal points of the no-signalling polytope. It is sufficient to check this condition for the extremal points of the no-signalling set, i.e., the PR behaviours and the LD behaviours. This is because if $\sigma$ is expressible as $\sigma=\lambda\sigma_{1}+(1-\lambda)\sigma_{2}$ then for any function $K$ satisfying $\mathbb{E}_{\sigma_{i}}[K(ABXY)]\le\mathbb{H}_{\sigma_{i}}(AB\lvert XY)$, we have $\mathbb{E}_{\sigma}[K]=\lambda\mathbb{E}_{\sigma_{1}}[K]+(1-\lambda)\mathbb{E}_{\sigma_{2}}[K]\le\lambda\mathbb{H}_{\sigma_{1}}(AB\lvert XY)+(1-\lambda)\mathbb{H}_{\sigma_{2}}(AB\lvert XY)\le\mathbb{H}_{\sigma}(AB\lvert XY)$. Hence if the condition holds for the extremal points, it will hold for all points in the set. To see that it does, we confirm by inspection that $\mathbb{E}_{\sigma}[K_{\ast}]$ attains the value 1 for the PR behaviour achieving the no-signalling maximum for the standard CHSH function, the value $-3$ for the PR behaviour achieving $-4$, and the value $-1$ for each of the PR behaviours that achieve the value 0, which are all less than or equal to the conditional Shannon entropy of the respective PR behaviours, which is 1. Likewise, we can check that $K_{\ast}$ is a valid entropy estimator for all the LD behaviours; it takes the value zero for the eight local deterministic distributions appearing in Table \ref{tab:Simplex_PR1} and $-2$ for the other eight, while $\mathbb{H}(AB\lvert XY)=0$ for these distributions. Hence, we have verified that $K_{\ast}$ satisfies the entropy estimator condition for all the extremal behaviours, and by extension all behaviours in the no-signalling polytope. 

Having shown $K_{\ast}$ is a is an entropy estimator, we next consider a sequence of functions $\{F_{k}\}_{k=1}^\infty$ where $F_{k}$ is defined according to the construction in Theorem \ref{thm:EntEst_PEF}:
\begin{equation}\label{e:PEF_sequence}
    F_{k}(ABXY) = 2^{(K_{\ast}(ABXY)-e^{-k})\beta_{k}},
\end{equation}
where we choose a positive $\beta_k$ making $F_k$ a PEF for each $k$, whose existence is guaranteed by the theorem. 
By construction, for each $k$ the function $F_{k}$ is a valid PEF with power $\beta_{k}>0$ for the set of no-signalling distributions.
The log-prob rate of $F_{k}$ at $\sigma$ is:
\begin{equation}\label{e:logprobrate_F_k}
\mathcal{O}_{\sigma}(F_{k};\beta_{k})=\frac{1}{\beta_k}\mathbb{E}_{\sigma}\left[\log_{2}\left(2^{(K_{\ast}-e^{-k})\beta_{k}}\right)\right]=\mathbb{E}_{\sigma}[K_{\ast}]-e^{-k}.
\end{equation}
We show robustness of the sequence in the following sense: for \textit{any} $\sigma\in \Pi_{\mathrm{NS}}$ violating the standard CHSH-Bell inequality, the log-prob rate of the sequence of PEFs $\{F_{k}\}_{k=1}^\infty$ is eventually positive. To see this, recall that as discussed in our brief review of the (2,2,2) Bell scenario, behaviours violating the standard CHSH-Bell inequality are contained in the non-local $8$-simplex $\Delta_{\mathrm{PR},1}^{8}$ (see Table~\ref{tab:1PR_8LD}). Hence, $\sigma$ is expressible as a convex combination of the vertices of $\Delta_{\mathrm{PR},1}^{8}$:
\begin{equation}\label{e:sigma_1PR+8LD}
\sigma(ab\lvert xy)\sigma_{\mathrm s}(xy) = \lambda_{\mathrm{PR},1}\mu_{\mathrm{PR},1}(ab\lvert xy)\sigma_{\mathrm s}(xy) + \sum_{i=1}^{8}\alpha_{i}\mu_{\mathrm{LD},i}(ab\lvert xy)\sigma_{\mathrm s}(xy),
\end{equation}
where $\lambda_{\mathrm{PR},1} + \sum_{i=1}^{8}\alpha_{i}=1$. This decomposition allows us to express the log-prob rate in terms of the standard CHSH-Bell function, which we define as $S(ABXY) = (-1)^{XY}(-1)^{A+B}/\sigma_{\mathrm s}(XY)$, where $\sigma_{\mathrm s}(XY)$ is the fixed settings distribution. We see that $\lambda_{\mathrm{PR},1}=(S_{\sigma}-2)/2$ in \eqref{e:sigma_1PR+8LD}, where $S_{\sigma}$ is the expected standard CHSH-Bell value according to the distribution $\sigma(ABXY)=\sigma(ab\lvert xy)\sigma_{\mathrm s}(xy)$. This follows by computing the expectation of $S$ according to the PR Box distribution $\mu_{\mathrm{PR},1}=\mu_{\mathrm{PR},1}(abxy)=\mu_{\mathrm{PR},1}(ab\lvert xy)\sigma_{\mathrm s}(xy)$, which is 4, and the expectation of $S$ according to the local distribution $\mu_{\mathrm{L},i}=\mu_{\mathrm{L},i}(abxy)=\mu_{\mathrm{LD},i}(ab\lvert xy)\sigma_{\mathrm s}(xy)$, which is 2. The
log-prob rate $\mathcal{O}_{\sigma}(F_{k};\beta_{k})$ for $F_{k}$ at $\sigma$ is then expressed as: 
\begin{equation}\label{e:logprobrate_sigma_1}
    \mathcal{O}_{\sigma}(F_{k};\beta_{k}) = \left(\frac{S_{\sigma}-2}{2}\right)\mathbb{E}_{\mu_{\mathrm{PR},1}}[K_{\ast}] + \sum_{i=1}^{8}\alpha_{i}\mathbb{E}_{\mu_{\mathrm{LD},i}}[K_{\ast}] - e^{-k}.
\end{equation}
Since $\mathbb{E}_{\mu_{\mathrm{LD},i}}[K_{\ast}]$ evaluates to zero for each $\mu_{\mathrm{LD},i}$ and $\mathbb{E}_{\mu_{\mathrm{PR},1}}[K_{\ast}]$ evaluates to $1$, the expression for $\mathcal{O}_{\sigma}(F_{k};\beta_{k})$ reduces to
$\mathcal{O}_{\sigma}(F_{k};\beta_{k}) = \frac{S_{\sigma}-2}{2} - e^{-k}$.
As $k\to\infty$, $\mathcal{O}_{\sigma}(F_{k};\beta_{k})=(S_{\sigma}-2)/2$ and so the quantity is eventually strictly positive provided $S_{\sigma}>2$, i.e., provided $\sigma$ violates the standard CHSH-Bell inequality.

Continuing our discussion on robustness, a different perspective on it would be to ask: given a PEF $F$ with power $\beta>0$ optimised with respect to the distribution $\rho$, how far in terms of total-variation distance can another distribution $\sigma$ be such that the same PEF (with the same power) can be used to certify randomness? Theorem~\ref{thm:robustPEF} provides a sufficient condition for the robustness of a positive, non-constant PEF $F=G^{\beta}$ 
with power $\beta$ in the following sense: assuming the log-prob rate of $F$ at $\rho$ is positive, the log-prob rate of $F$ at a different distribution $\sigma$ is positive if $d_{\mathrm{TV}}(\rho,\sigma)$ is within a certain bound (as given in~\eqref{e:suff_cond_robustPEF}). For the sequence $\{F_{k}\}_{k=1}^\infty$ of PEFs the upper-bound on $d_{\mathrm{TV}}(\rho,\sigma)$ is computed as follows: Notice that in the sequence $\{F_{k}\}_{k=1}^\infty$ of PEFs, $F_{k}$ is of the form $F_{k}=G_{k}^{\beta_{k}}$, where $G_{k}=2^{K_{\ast}-e^{-k}}$. The upper-bound on $d_{\mathrm{TV}}(\rho,\sigma)$ (as given in~\eqref{e:suff_cond_robustPEF}) is then $\mathbb{E}_{\rho}[G_{k}]/(L-l)=\frac{1}{4}\left(\frac{S_{\rho}-2}{2}-e^{-k}\right)$. It is worthwhile to observe that given a standard-CHSH Bell inequality violating distribution $\rho$, this upper-bound approaches the \emph{strength of non-locality} for $\rho$ which is expressed as $(S_{\rho}-2)/8$. The strength of non-locality is defined in terms of how far the non-local no-signalling distribution $\rho$ is from the local set $\Pi_{\mathrm{L}}$~\cite{PhysRevA.97.022111}. It is defined as follows:
\begin{equation}\label{e:Strength_Nonlocality}
    d_{\mathrm{NL}}(\rho)\coloneqq \frac{1}{\abs{\mathcal{X}}\abs{\mathcal{Y}}}\frac{1}{2}\min_{\tau\in\Pi_{L}}\sum_{abxy}\abs{\rho(ab\lvert xy)-\tau(ab\lvert xy)},
\end{equation}
where the minimum is over all distributions $\tau$ belonging to the local set $\Pi_{\mathrm{L}}$. In the definition of $d_{\mathrm{NL}}(\rho)$ in~\eqref{e:Strength_Nonlocality} we have assumed a uniform settings distribution as is evident from the factor $1/\abs{\mathcal{X}}\abs{\mathcal{Y}}$, where $\abs{\mathcal{X}}$ and $\abs{\mathcal{Y}}$ denote the number of the measurement settings choices for Alice and Bob, respectively (which for the (2,2,2) Bell scenario is $2$ for Alice and $2$ for Bob). A theorem in~\cite{Bierhorst_2016} (see Theorem 3.1) provides a condition for the local distribution $\tau$ such that the minimum $(1/2)\min_{\tau\in\Pi_{\mathrm{L}}}\sum_{abxy}\abs{\rho(ab\lvert xy)-\tau(ab\lvert xy)}$ in~\eqref{e:Strength_Nonlocality} is achieved and that the minimum comes out to be the weight $(S_{\rho}-2)/2$ on the PR-box in the expression of $\rho$ as the convex combination of the vertices of $\Delta_{\mathrm{PR},1}^{8}$; and so per the definition in~\eqref{e:Strength_Nonlocality} $d_{\mathrm{NL}}(\rho)=(S_{\rho}-2)/8$. 
Thus, the bound $\frac{1}{4}\left(\frac{S_{\rho}-2}{2}-e^{-k}\right)$ from Theorem~\ref{thm:robustPEF} approaches $\frac{S_{\rho}-2}{8}$ which is the strength of non-locality $d_{\mathrm{NL}}(\rho)$ for $\rho$. This illustrates that a bound of this form cannot be improved, in the sense that increasing the total variation distance from $\rho$ by any positive amount will encompass local distributions which cannot certify randomness.

Thus $\{F_k\}_{k=1}^\infty$ is fully robust as $k\to \infty$. Next, we confirm that $\{F_k\}_{k=1}^\infty$ is asymptotically optimal in terms of min-entropy per trial (i.e., log-prob rate), for any
distribution $\sigma$ violating the standard CHSH inequality. Since $\Pi_{\mathrm{NS}}$ is closed and equal to the convex hull of its extremal points,
Theorem~\ref{thm_hmin_achieved} implies that given such a $\sigma$, the adversary has a strategy obtained through an IID attack based on a single-trial distribution whose conditional Shannon entropy is equal to the infimum defined in~\eqref{e:def_IIDattack}. We can identify this attack. 
The optimisation in~\eqref{e:def_IIDattack} can be expressed as follows:
\begin{equation}\label{e:hminagain}
    h_{\min}(\sigma) = \min\Big\{H_{\nu}(AB\lvert XYE)\colon\nu_{e}\in\Pi_{\mathrm{NS}},\sum_{e}\nu(e)\nu_{e}=\sigma\Big\},
    \end{equation}
where $\nu_{e}=\nu(ABXY\lvert e)$. We compute $H_\nu(AB|XYE)$ for the decomposition of $\sigma$ given in~\eqref{e:sigma_1PR+8LD}, where we have noted $\lambda_{\mathrm{PR},1}=(S_{\sigma}-2)/2$. 
Since the conditional Shannon entropy is one for PR boxes and zero for LD behaviours, we obtain  $H_\nu(AB|XYE)= (S_{\sigma}-2)/2$, and hence $h_{\min}(\sigma)$ is no larger than this value. But since this expression is same as that of the asymptotic log-prob rate of the sequence $\{F_{k}\}_{k=1}^{\infty}$ of valid PEFs, we can say $h_{\min}(\sigma)$ is also no smaller than this value, and so $h_{\min}(\sigma)= (S_{\sigma}-2)/2$.
This demonstrates the asymptotic optimality of the sequence $\{F_{k}\}_{k=1}^\infty$ in the sense that the PEFs in the sequence get arbitrarily close to certifying an asymptotic randomness rate of $h_{\min}(\sigma)$.

In our proof of the asymptotic optimality of the sequence $\{F_k\}_{k=1}^\infty$, we identified the optimal attack by an adversary: it is to prepare the decomposition in \eqref{e:sigma_1PR+8LD} with each $e$ corresponding to one of the (up to) nine extremal behaviours, with respective $\nu(e)$ weights of $\lambda_{\mathrm{PR},1}$ and $\alpha_{i}$. This can be seen to be the \textit{unique} attack achieving  $h_{\min}(\sigma)$, through an argument we sketch as follows: (1) any $\nu$-decomposition of $\sigma$ can be improved upon (i.e., reducing $H_\nu(AB|XYE)$) by considering only extremal $\nu_e$, by the concavity of conditional Shannon entropy; (2) any decomposition including positive weights on more than one PR box can be \textit{strictly} improved upon by one with weights on a single PR box, by Theorem 2.1 of~\cite{Bierhorst_2016}, which shows how to replace equal mixtures of two PR boxes with mixtures of a single PR box and local deterministic distributions; (3) this decomposition can be further strictly improved via Theorem 2.2 of~\cite{Bierhorst_2016} by removing any local deterministic distributions not saturating the CHSH-Bell inequality with those that do (the improvement being obtained by decreasing the weight on the sole remaining PR box). The resulting decomposition---that of~\eqref{e:sigma_1PR+8LD}---is thus the unique optimiser of~\eqref{e:hminagain}. It witnesses the bound of $1+\mathrm{dim}(\Pi_{\mathrm{NS}})=1+8=9$ on the set $\mathcal{E}$ (as shown in Theorem~\ref{thm_hmin_achieved}).  In general, positive weight on all 9 extremal boxes may be necessary, due to their affine independence which was noted in Section \ref{e:brief_review}. One can confirm this visually from Table~\ref{tab:Simplex_PR1}: weight on the (only) non-local distribution, the PR box, is necessary to violate the CHSH-Bell inequality, and any distribution with non-zero probabilities for each possible outcome (a property possessed by, for example, the quantum distribution saturating Tsirelson's bound) will require positive weight on all the local deterministic behaviours, as each LD behaviour corresponds to a distinct sole appearance of the number ``1" in a column otherwise populated by zeroes in Table~\ref{tab:Simplex_PR1}. This witnesses that further reduction of the $1+\mathrm{dim}(\Pi_{\mathrm{NS}})$ bound on $|\mathcal E|$ in Theorem~\ref{thm_hmin_achieved} is impossible, and so this bound is optimal.

It is an important observation that the adversary needs to prepare only one non-classical state in her realisation of the optimal attack, since the preparation of a non-classical state is likely the most difficult aspect of the attack. We now explore possible generalisations of this feature to other trial models.

\subsection{Characterising the optimal attack in different scenarios}

We start by exploring the possibility of arriving at a similar analytic characterisation of the optimal adversarial attack when the adversary is limited to only quantum-realisable distributions.
Suppose now that our trial model is the set $\Pi_{\mathrm Q}$ of quantum-achievable distributions for the (2,2,2) scenario. The adversary is still constrained to performing probabilistic attacks to simulate the trial statistics, while generating the least amount of randomness possible; however, she now tries to mimic the trial statistics using quantum-achievable distributions. The optimisation routine depicting this goal is:
\begin{equation}\label{e:Eve_objective_Q}
    \tilde{h}_{\min}(\sigma)=\min\Big\{H_{\omega}(AB\lvert XYE)\colon\omega_{e}\in\Pi_{\mathrm{Q}},\sum_{e}\omega(e)\omega_{e}=\sigma\Big\},
\end{equation}
where $\omega_{e}=\omega(ABXY\lvert e)$. The set $\Pi_{\mathrm Q}$ is compact and convex, but unlike $\Pi_{\mathrm{NS}}$, is not a polytope and so there is a continuum of extremal points.

We conjecture that the minimum in~\eqref{e:Eve_objective_Q} is achieved at a distribution that marginalises to the observed trial distribution through a convex combination of (only) one quantum extremal distribution violating the standard CHSH-Bell inequality and no more than eight local deterministic distributions that saturate the same inequality.

An attempt to prove this will require an understanding of the geometry of the quantum set, and in particular its extremal points. We do not yet have a complete characterisation of the set of behaviours $\Pi_{\mathrm{Q}}$ (in the true $\mathbb{R}^{8}$ space), although a recent work has conjectured an analytic criterion for extremality in the CHSH scenario~\cite{mikosnuszkiewicz2023extremal}. However, a characterisation does exist when we make the assumption of unbiased marginals: $\mu(A=0\lvert x)=\mu(A=1\lvert x)=1/2$ for all $x\in\{0,1\}$ and $\mu(B=0\lvert y)=\mu(B=1\lvert y)=1/2$ for all $y\in\{0,1\}$, in which case the set of behaviours is four dimensional. The unbiased marginal case has been completely characterised, a detailed exposition of which can be found in~\cite{Le2023quantumcorrelations} (see Theorem 1 therein).

\begin{figure}[H]
    \centering
    \subfloat[]{\includegraphics[width=0.486\textwidth]{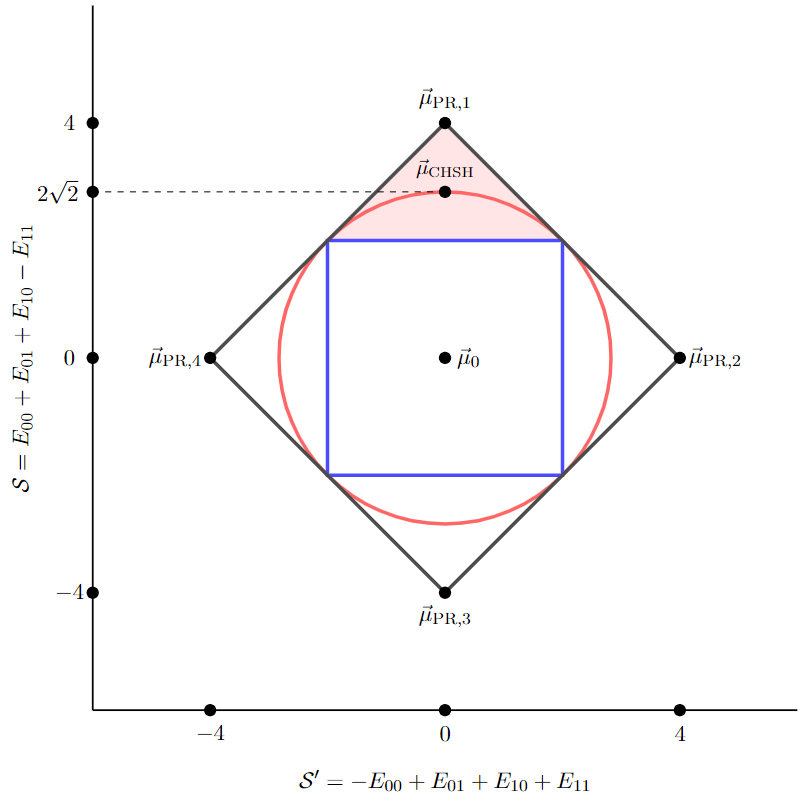}\label{fig:2DSlice}}
    \hfill
    \subfloat[]{\includegraphics[width=0.486\textwidth]{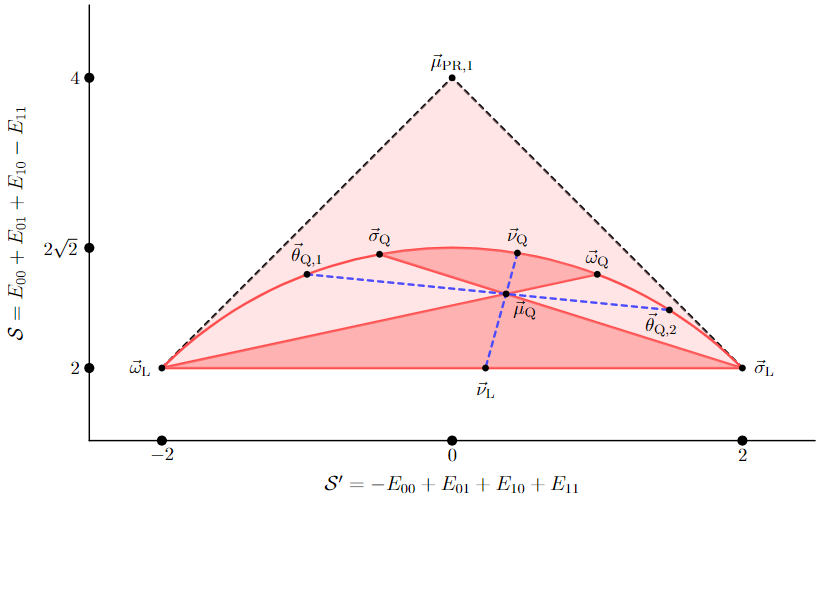}\label{fig:SliceAboveCHSHfacet}}
    \captionsetup{singlelinecheck=off}
    \caption[.]{(a) A 2-dimensional slice of the set of no-signalling behaviours (containing the quantum and the local set). The behaviours can be parametrised as the CHSH-Bell values $\mathcal{S}$ and $\mathcal{S}^{\prime}$ obtained by two different versions of the CHSH-Bell expression in~\eqref{eq_Bell}. Any behaviour on the slice can be represented as:
\begin{equation}\label{e:NSB2Dslice}
    \vec{\mu} = \frac{\mathcal{S}}{4}\vec{\mu}_{\mathrm{PR},1} + \frac{\mathcal{S}'}{4}\vec{\mu}_{\mathrm{PR},2} + \left(1 - \frac{\mathcal{S}+\mathcal{S}'}{4}\right)\vec{\mu}_{0},\quad \mathcal{S,S'}\in[-4,4],
\end{equation}
    where $\vec{\mu}_{0}$ is the maximally random behaviour obtained as the equal mixtures of all 16 local deterministic behaviours. The disk $\mathcal{S}^{2}+(\mathcal{S}')^{2}\le 8$ represents the set of quantum behaviours.
    (b) The portion of the 2-dimensional slice containing the no-signalling (including quantum-achievable) behaviours above the standard CHSH-Bell facet. For a fixed behaviour $\vec\mu_\mathrm{Q}$ in the interior of the quantum region, the darker shaded region corresponds to possible ways of expressing $\vec\mu_\mathrm{Q}$ as a convex combination of a behaviour on the quantum boundary and a behaviour on the local boundary (for example, $\vec\mu_\mathrm{Q}=\lambda\vec\nu_\mathrm{Q}+(1-\lambda)\vec\nu_\mathrm{L},\,\lambda\in(0,1)$). For the same behaviour $\vec\mu_\mathrm{Q}$, the lighter shaded region represents possible ways of expressing it as a convex combination of two behaviours on the quantum boundary (for example, $\vec\mu_\mathrm{Q}=\delta\vec\theta_\mathrm{Q,1}+(1-\delta)\vec\theta_\mathrm{Q,2},\,\delta\in(0,1)$).} 
\end{figure}

A key enabling step in the direction of characterising the optimal attack in the unbiased marginals case would be to see if the following two conditions hold simultaneously: first, a convex combination of any two extremal quantum behaviours can be expressed equivalently as a different convex combination of one extremal quantum behaviour (different from the previous two) and classical noise (mixtures of the local deterministic behaviours), i.e., for extremal quantum behaviours $\vec{\mu}_{1},\vec{\mu}_{2}$, the convex combination $\lambda\vec{\mu}_{1}+(1-\lambda)\vec{\mu}_{2}$ can be re-expressed as the convex combination $\delta\vec{\mu}_{3}+(1-\delta)\vec{\mu}_{0}$, where $\lambda,\delta\in(0,1)$, $\vec{\mu}_{3}$ is a third extremal quantum behaviour,  and $\vec{\mu}_{0}$ is a mixture of the local deterministic behaviours; and second, $\lambda\mathbb{H}_{\mu_{1}}(AB\lvert XY)+(1-\lambda)\mathbb{H}_{\mu_{2}}(AB\lvert XY) \ge \delta\mathbb{H}_{\mu_{3}}(AB\lvert XY)$, where the term $(1-\delta)\mathbb H_{\mu_0}(AB|XY)$ that might be expected to appear on the right vanishes due to the concavity of conditional Shannon entropy and the fact that it is zero for local deterministic behaviours, into which $\vec{\mu}_0$ can be decomposed. 

A numerical inspection to check---by means of an exhaustive search---if these two conditions hold simultaneously (in the uniform marginals case) introduces a lot of free variables. If we add more symmetry to the behaviours with uniform marginals and constrain ourselves to the 2-dimensional slice as shown in Figure~\ref{fig:2DSlice},\footnote[2]{This can be done as follows: A behaviour with uniform marginals can be completely specified by the correlators $(E_{00},E_{01},E_{10},E_{11})$, where $-1\le E_{xy}\le 1,\,\forall x,y$ (see the line following~\eqref{eq_Bell} for the definition of $E_{xy}$). To obtain behaviours in the 2-dimensional slice as shown in Figure~\ref{fig:2DSlice} one can restrict attention to distributions of the form $\mu(ab\lvert xy)=\frac{1}{4}(1+(-1)^{a+b}C_{xy})$ where $C_{00}=-C_{11}=\frac{E_{00}-E_{11}}{2}$ and $C_{01}=C_{10}=\frac{E_{01}+E_{10}}{2}$.} where the behaviours are given by the formula \eqref{e:NSB2Dslice} and are of the form displayed in Table \ref{tab:QB2DSlice}, then one can perform numerical search to see if the two conditions mentioned above hold simultaneously, and we did observe it to hold in some initial numerical investigations comparing the $\vec \theta$ decompositions against $\vec \nu$ decompositions as depicted in Figure~\ref{fig:SliceAboveCHSHfacet}.

\begin{table}[h]
\centering
{\renewcommand{\arraystretch}{1.24}%
\begin{tabular}{|r|r|c|c|c|c|}
\cline{3-6}
\multicolumn{1}{r}{} & \multicolumn{1}{r}{} & \multicolumn{4}{|c|}{$ab$}\\
\cline{3-6}
\multicolumn{1}{r}{} & \multicolumn{1}{r}{} & \multicolumn{1}{|c}{$00$} & \multicolumn{1}{c}{$01$} & \multicolumn{1}{c}{$10$} & \multicolumn{1}{c|}{$11$}\\
\hline
\multirow{4}{*}{$xy$} & $00$ & $s_{1}$ & $s_{2}$ & $s_{2}$ & $s_{1}$ \\
\cline{3-6}
 & $01$ & $s_{3}$ & $s_{4}$ & $s_{4}$ & $s_{3}$ \\
\cline{3-6}
 & $10$ & $s_{3}$ & $s_{4}$ & $s_{4}$ & $s_{3}$ \\
\cline{3-6}
 & $11$ & $s_{2}$ & $s_{1}$ & $s_{1}$ & $s_{2}$ \\
\hline
\end{tabular}}
\caption{Tabular representation of the no-signalling behaviours on the 2-dimensional slice shown in Figure~\ref{fig:2DSlice}. The behaviours have uniform marginals, i.e., the probability of observing an outcome conditioned on a measurement setting is $1/2$ for each party for all outcomes and settings. The behaviours are further constrained in having the third and fourth row completely determined by the first and second, which need not hold in general for uniform marginal distributions, and brings the dimensionality down from four to two. Any behaviour represented as above is parameterised as the values $\mathcal{S}$ and $\mathcal{S}^{'}$ of the two versions of the CHSH-Bell expression $E_{00}+E_{01}+E_{10}-E_{11}$ and $-E_{00}+E_{01}+E_{10}+E_{11}$, respectively: $s_{1} = (4+\mathcal{S}-\mathcal{S}^{\prime})/16$, $s_{2} = (4+\mathcal{S}^{\prime}-\mathcal{S})/16$, $s_{3} = (4+\mathcal{S}+\mathcal{S}^{\prime})/16$, $s_{4} = (4-\mathcal{S}-\mathcal{S}^{\prime})/16$, where for the no-signalling set $-4\le\mathcal{S}'+\mathcal{S}\le 4,\,-4\le\mathcal{S}'-\mathcal{S}\le 4$ and for the quantum set $\mathcal{S}^{2}+(\mathcal{S}')^{2}\le 8$.}\label{tab:QB2DSlice}
\end{table}

Going beyond the minimal Bell scenario, we considered the possibility of a similar characterisation of optimal \emph{no-signalling} adversarial attack in higher $(n,m,k)$ Bell scenarios. In the (2,2,2) Bell scenario the analytical characterisation of the optimal adversarial attack crucially relied upon the geometric features of the no-signalling polytope, namely Theorems 2.1 and 2.2 in~\cite{Bierhorst_2016}: that equal mixtures of two PR behaviours are expressible as equal mixtures of four distinct LD behaviours and consequently, a behaviour violating any of the eight versions (up to local relabelling of the outcomes and settings) of the CHSH-Bell inequality is expressible as a convex combination of the one PR behaviour achieving the non-local maximum and (up to) eight LD behaviours achieving the local maximum of the corresponding CHSH-Bell expression. These geometric features, however, do not extend to the no-signalling polytopes of higher $(n,m,k)$ Bell scenarios. Membership of equal mixtures of extremal no-signalling non-local behaviours in the local polytope holds solely in the (2,2,2) Bell scenario.

Below we provide examples of equal mixtures of no-signalling non-local extremal behaviours in the $(2,2,3)$, $(2,3,2)$ and $(3,2,2)$ Bell scenarios that do not belong to the local polytope. One can use linear programming to check non-locality of the such examples. Assessment of locality of a behaviour is an instance of the \emph{membership problem of the local polytope}. Since the local deterministic (LD) behaviours are the extremal points of the local polytope, we can formulate our problem as a \emph{feasibility linear program}. Suppose $\left\{\vec{\mu}_{\mathrm{LD},1},\vec{\mu}_{\mathrm{LD},2},\ldots,\vec{\mu}_{\mathrm{LD},\#_{\mathrm{LD}}}\right\}$ is the set of LD behaviours for some Bell scenario. The vector $\vec{\mu}_{\mathrm{LD},i}\in\mathds{R}^{d}$ denotes the joint probability of outcomes conditioned on the input choices and $d$ is the dimension of the ambient space in which the vector lies. The feasibility linear program has the variable $\vec{x}\in\mathds{R}^{\#_{\mathrm{LD}}}$. The inequality constraints comprise of $x_{i}\ge 0,\,i\in[\#_{\mathrm{LD}}]$ and the equality constraints are $\sum_{i=1}^{\#_{\mathrm{LD}}}x_{i}=1$ and the following:
\begin{equation}\label{e:equalMix_membership}
    \frac{1}{2}\overrightarrow{\mathsf{NS}}_{\mathrm{extr}} + \frac{1}{2}\overrightarrow{\mathsf{NS}^{\prime}}_{\mathrm{extr}} = \sum_{k=1}^{\#_{\mathrm{LD}}}x_{k}\vec{\mu}_{\mathrm{LD},k}
\end{equation}
where $\overrightarrow{\mathsf{NS}}_{\mathrm{extr}}$ is a non-local no-signalling extremal behaviour. The details on formulating the dual of this linear program can be found in section E.2.1 of the Appendix of \cite{10.1093/oso/9780198788416.001.0001}.

Before presenting the counter-examples we briefly review the $(n,m,k)$ Bell scenario: This scenario 
 consists of $n$ spatially separated parties, where each party $i\in[n]$ has a choice of $m$ different $k$-outcome measurements. For $\mathcal{X}\equiv\{0,1,\ldots,m-1\}$ and $\mathcal{A}\equiv\{0,1,\ldots,k-1\}$ the joint probability $\mu(a_{1}a_{2}\ldots a_{n}\lvert x_{1}x_{2}\ldots x_{n})$ of obtaining the outcomes $(a_{1},a_{2},\ldots,a_{n})\in\mathcal{A}^{n}$ conditioned on the inputs $(x_{1},x_{2},\ldots,x_{n})\in\mathcal{X}^{n}$ can be viewed as a probability vector $\vec{\mu}\in\mathds{R}^{d}$, where $d=\left(\abs{\mathcal{A}}\abs{\mathcal{X}}\right)^{n}$. 

The extremal points of the no-signalling polytope comprise of the local deterministic (LD) behaviours and the non-local extremal behaviours. The LD behaviours consist of all possible assignments $\Lambda_{\mathrm{LD}} = \{\{\lambda_{1x_{1}}\}_{x_{1}\in\mathcal{X}};\{\lambda_{2x_{2}}\}_{x_{2}\in\mathcal{X}};\ldots;\{\lambda_{nx_{n}}\}_{x_{n}\in\mathcal{X}}\}$, where $\lambda_{ix_{i}}\in\mathcal{A}$ for $i\in[n]$. The number of such assignments is $\#_{\mathrm{LD}}=(\abs{\mathcal{A}})^{n\abs{\mathcal{X}}}$. Corresponding to each assignment $\lambda\in\Lambda_{\mathrm{LD}}$ the LD probabilities are expressed as
\begin{equation}\label{e:def_LD}
\mu_{\mathrm{LD}, k}(a_{1}a_{2}\ldots a_{n}\lvert x_{1}x_{2}\ldots x_{n}) = [\![a_{1}=\lambda_{1x_{1}}]\!][\![a_{12}=\lambda_{2x_{2}}]\!]\cdots[\![a_{n}=\lambda_{nx_{n}}]\!]
\end{equation}
where $[\![\cdot]\!]$ is the function that evaluates to $1$ if the condition within holds, $0$ otherwise. A behaviour $\vec{\mu}_{\mathcal{L}}$ is local if it can be expressed as $\vec{\mu}_{\mathcal{L}}=\sum_{k=1}^{\#_{\mathrm{LD}}}q_{k}\vec{\mu}_{\mathrm{LD},k}$, where $q_{k}\ge 0$ and $\sum_{k=1}^{\#_{\mathrm{LD}}}q_{k}=1$.\\

\emph{$(2,2,3)$ Bell scenario:}
This scenario is an instance of the more general $(2,2,k)$ scenario, also known in the literature as the CGLMP scenario~\cite{PhysRevLett.88.040404}, for $k=3$. In this bipartite scenario the parties have two $3$-output choices of settings. The extremal behaviours for the no-signalling polytope for the CGLMP scenario have been fully described in~\cite{PhysRevA.71.022101}. The non-local no-signalling extremal behaviours for the $(2,2,3)$ scenario, up to relabelling of inputs and outcomes, are given by the following formula:
\begin{equation}\label{NLbox_223}
\mathsf{NL}_{\mathrm{ext}}(ab\lvert xy) \coloneqq \begin{cases}
\frac{1}{3} &:\quad b-a\equiv xy\Mod{3} \\
0 &: \quad \text{otherwise}
\end{cases} 
\end{equation}
where $a,b\in\{0,1,2\}$ and $x,y\in\{0,1\}$ are the outputs and inputs for the parties, respectively. We found that~\eqref{e:equalMix_membership} does not necessarily hold for all equal mixtures of a pair of distinct non-local extremal behaviours. Among the several examples we found that violate~\eqref{e:equalMix_membership}, Table~\ref{table:NL_extr_CGLMP} shows one such example. 
\begin{table}[H]
\centering
{\renewcommand{\arraystretch}{1.15}%
\begin{tabular}{c|c c c c c c c c c|}
\multicolumn{1}{r}{} & \multicolumn{1}{c}{$ab$} & \multicolumn{1}{c}{$ab^{\prime}$} & \multicolumn{1}{c}{$ab^{\prime\prime}$} & \multicolumn{1}{c}{$a^{\prime}b$} & \multicolumn{1}{c}{$a^{\prime}b^{\prime}$} & \multicolumn{1}{c}{$a^{\prime}b^{\prime\prime}$} & \multicolumn{1}{c}{$a^{\prime\prime}b$} & \multicolumn{1}{c}{$a^{\prime\prime}b^{\prime}$} & \multicolumn{1}{c}{$a^{\prime\prime}b^{\prime\prime}$} \\
\cline{2-10}
$xy$ & $1/3$ & & & & $1/3$ & & & & $1/3$ \\
$xy^{\prime}$ & $1/3$ & & & & $1/3$ & & & & $1/3$ \\
$x^{\prime}y$ & $1/3$ & & & & $1/3$ & & & & $1/3$ \\
$x^{\prime}y^{\prime}$ & & $1/3$ & & & & $1/3$ & $1/3$ & & \\
\cline{2-10}
\end{tabular}}
{\renewcommand{\arraystretch}{1.15}%
\begin{tabular}{c|c c c c c c c c c|}
\multicolumn{1}{r}{} & \multicolumn{1}{c}{$ab$} & \multicolumn{1}{c}{$ab^{\prime}$} & \multicolumn{1}{c}{$ab^{\prime\prime}$} & \multicolumn{1}{c}{$a^{\prime}b$} & \multicolumn{1}{c}{$a^{\prime}b^{\prime}$} & \multicolumn{1}{c}{$a^{\prime}b^{\prime\prime}$} & \multicolumn{1}{c}{$a^{\prime\prime}b$} & \multicolumn{1}{c}{$a^{\prime\prime}b^{\prime}$} & \multicolumn{1}{c}{$a^{\prime\prime}b^{\prime\prime}$} \\
\cline{2-10}
$xy$ & & $1/3$ & & & & $1/3$ & $1/3$ & & \\
$xy^{\prime}$ & $1/3$ & & & & $1/3$ & & & & $1/3$ \\
$x^{\prime}y$ & $1/3$ & & & & $1/3$ & & & & $1/3$ \\
$x^{\prime}y^{\prime}$ & $1/3$ & & & & $1/3$ & & & & $1/3$ \\
\cline{2-10}
\end{tabular}}
\caption{Two non-local extremal behaviours for the CGLMP scenario with $3$ outcomes whose equal mixtures is non-local. The inputs $x,y\in\{0,1\}$ and the outcomes $a,b\in\{0,1,2\}$ with $x^{\prime}=x\oplus 1,\,y^{\prime}=y\oplus 1$ and $a^{\prime}=a\oplus_{3}1,\,a^{\prime\prime}=a\oplus_{3}2, b^{\prime}=b\oplus_{3}1,\,b^{\prime\prime}=b\oplus_{3}2$. The symbol $\oplus$ denotes addition modulo $2$ and $\oplus_{3}$ denotes addition modulo $3$. The missing entries correspond to $0$. The top behaviour comes directly from \eqref{NLbox_223} while the bottom behaviour is obtained through the relabelling $x\leftrightarrow x'$ and $y \leftrightarrow y'$. An equal mixture of these two boxes lies outside the local polytope.}
\label{table:NL_extr_CGLMP}
\end{table}

\emph{$(2,3,2)$ Bell scenario}:
More generally, the extremal behaviours of $(2,k,2)$ no-signalling polytope, with $k>2$, have been completely characterised in~\cite{PhysRevLett.95.140401} and~\cite{PhysRevA.72.052312}, of which the $(2,3,2)$ is an instance. Following Table II of ~\cite{PhysRevA.72.052312} we can obtain Tables~\ref{t:Bell232_NL1} and~\ref{t:Bell_232_NL2} which are two representative examples of non-local no-signalling extremal behaviours, equal mixtures of which lie outside the local polytope. In Table~\ref{t:Bell232_NL1} all input choices, $x,y\in\{0,1,2\}$, for Alice and Bob have uniform probabilities of outcomes; in Table~\ref{t:Bell_232_NL2} all inputs for Alice and inputs $y\in\{0,1\}$ for Bob have uniform probabilities of outcomes, with the exception that Bob's outcome for $y=2$ is deterministic.

\begin{align*}
    {\renewcommand{\arraystretch}{1.15}%
    \begin{tabular}{|c c|}
    \hline
    $p(00\lvert xy)$ & $p(01\lvert xy)$ \\
    $p(10\lvert xy)$ & $p(11\lvert xy)$ \\
    \hline
    \end{tabular}}
    \qquad { \text{with} } \qquad
    {\renewcommand{\arraystretch}{1.15}%
    \begin{tabular}{|c c|}
    \hline
    ? &  \\
      &  \\
    \hline
    \end{tabular}} \equiv {\renewcommand{\arraystretch}{1.15}%
    \begin{tabular}{|c c|}
    \hline
    $1/2$ & $0$ \\
    $0$ & $1/2$ \\
    \hline
    \end{tabular}} \text{ or } {\renewcommand{\arraystretch}{1.15}%
    \begin{tabular}{|c c|}
    \hline
    $0$ & $1/2$ \\
    $1/2$ & $0$ \\
    \hline
    \end{tabular}}
\end{align*}
\begin{table}[H]
\begin{minipage}{0.45\linewidth}
\centering
{\renewcommand{\arraystretch}{1.15}%
\begin{tabular}{r|r|c c|c c|c c|}
\multicolumn{1}{r}{} & \multicolumn{1}{r}{} & \multicolumn{6}{c}{$y$} \\
\cline{3-8}
\multicolumn{1}{r}{} & \multicolumn{1}{r|}{} & \multicolumn{2}{c|}{$0$} & \multicolumn{2}{c|}{$1$} & \multicolumn{2}{c|}{$2$} \\
\cline{2-8}
\multirow{6}{*}{$x$} & \multirow{2}{*}{$0$} & $1/2$ & $0$ & $1/2$ & $0$ & $1/2$ & $0$ \\
 &  & $0$ & $1/2$ & $0$ & $1/2$ & $0$ & $1/2$ \\
\cline{2-8}
 & \multirow{2}{*}{$1$} & $1/2$ & $0$ & $0$ & $1/2$ & ? & \\
 & & $0$ & $1/2$ & $1/2$ & $0$ & & \\
 \cline{2-8}
 & \multirow{2}{*}{$2$} & $1/2$ & $0$ & ? &  & ? & \\
 &  & $0$ & $1/2$ &  &  &  & \\
 \cline{2-8}
\end{tabular}}
\caption{Non-local no-signalling extremal behaviour with all input choices $x,y\in\{0,1,2\}$ for Alice and Bob having uniform probabilities of outcomes.}
\label{t:Bell232_NL1}
\end{minipage}
\hfill%
\begin{minipage}{0.45\linewidth}
\centering
{\renewcommand{\arraystretch}{1.15}%
\begin{tabular}{r|r|c c|c c|c c|}
\multicolumn{1}{r}{} & \multicolumn{1}{r}{} & \multicolumn{6}{c}{$y$} \\
\cline{3-8}
\multicolumn{1}{r}{} & \multicolumn{1}{r|}{} & \multicolumn{2}{c|}{$0$} & \multicolumn{2}{c|}{$1$} & \multicolumn{2}{c|}{$2$} \\
\cline{2-8}
\multirow{6}{*}{$x$} & \multirow{2}{*}{$0$} & $1/2$ & $0$ & $1/2$ & $0$ & $1/2$ & $0$ \\
 &  & $0$ & $1/2$ & $0$ & $1/2$ & $1/2$ & $0$ \\
\cline{2-8}
 & \multirow{2}{*}{$1$} & $1/2$ & $0$ & $0$ & $1/2$ & $1/2$ & $0$ \\
 & & $0$ & $1/2$ & $1/2$ & $0$ & $1/2$ & $0$ \\
 \cline{2-8}
 & \multirow{2}{*}{$2$} & $1/2$ & $0$ & ? &  & $1/2$ & $0$ \\
 &  & $0$ & $1/2$ &  &  & $1/2$ & $0$ \\
 \cline{2-8}
\end{tabular}}
\caption{All inputs for Alice and inputs $y\in\{0,1\}$ for Bob have uniform probabilities of outcomes, while Bob's outcome for $y=2$ is deterministic for this non-local no-signalling extremal behaviour.}
\label{t:Bell_232_NL2}
\end{minipage}
\end{table}
There are 16 possible mixtures of the two behaviours in Tables~\ref{t:Bell232_NL1} and~\ref{t:Bell_232_NL2} corresponding to each `?' in each table being a perfect correlation or a perfect anti-correlation, all of which represent mixtures of extremal non-local boxes \cite{PhysRevA.72.052312} and all lie outside the local polytope. The non-locality of the mixtures is confirmed by noting that the four cells in the upper left corner, corresponding to restricting the settings choices to $x,y\in\{0,1\}$, is the PR box distribution which is of course non-local.\\

\emph{$(3,2,2)$ Bell scenario:}
This is a tripartite scenario with each party having binary input choices and outcomes. The no-signalling polytope consists of 46 inequivalent classes of extremal behaviours, of which one is the class comprising of 64 LD behaviours. A complete characterisation can be found in~\cite{Pironio_2011}. As an example violating \eqref{e:equalMix_membership} we can refer to the observation made in Section 2.3 of~\cite{Pironio_2011} that equal mixtures of two behaviours in Class 46 (see Table 1 of \cite{Pironio_2011}) is a GHZ correlation which is expressed (entirely in terms of correlators ]$\langle A_{x}B_{y}C_{z}\rangle$) as $P_{\mathrm{GHZ}}(abc\lvert xyz)=\frac{1}{8}(a+abc\langle A_{x}B_{y}C_{z}\rangle)$. $\vec{P}_{\mathrm{GHZ}}$ is a non-local behaviour which is obtained by measuring $\frac{1}{\sqrt{2}}(\ket{000}+\ket{111})$ in suitable local bases \cite{https://doi.org/10.48550/arxiv.0712.0921}.

\section{Conclusion}
In this work, we revisited the probability estimation framework with the goal of presenting a complete and self-contained proof of its optimality in the asymptotic regime and obtaining a better characterisation of optimal adversarial attack strategies on the protocol. We obtained in Theorem~\ref{thm_hmin_achieved} an improved and tight upper bound on the cardinality of the set of states needed in the optimal attack, and studied the implications of this result for specific scenarios in Section~\ref{s:application}. We also considered the question of \textit{robustness} for the PEF method, finding that asymptotic optimality of PEFs (in terms of randomness generation rate) need not entail a trade-off with robustness to small deviations from expected experimental behaviour.

In proving the optimality of the framework, our results show that there remains nothing to be gained, asymptotically, for an adversary implementing memory attacks---an i.i.d. attack is asymptotically optimal. However, in real world applications this may not hold. The number of trials in a Bell experiment are finite, albeit large, and there are unavoidable correlations between the successive trials (referred to as memory effects). We leave to future work considerations of side-channel attacks in the non-asymptotic (finite trials) regime for the probability estimation framework.

\section*{Acknowledgements}

We acknowledge helpful discussions with Jitendra Prakash and Mark Wilde. This work was partially supported by AFOSR Grant FA9550-20-1-0067, NSF Award 1839223, and Louisiana Board of Regents Award LEQSF (2019-22)-RD-A-27.

\appendix

\section{Proofs for Theorems~\ref{PEF_firstthm} and~\ref{PEF_thm2}}\label{a:PEFproof_lbscme}
First, we present the proof for Theorem~\ref{PEF_firstthm}.
\begin{theorem*}
Suppose $\mu\colon\mathcal{C}^{n}\times\mathcal{Z}^{n}\times\mathcal{E}\to[0,1]$ is a distribution of $\mathbf{CZ}E$ such that $\mu_{e}(\mathbf{CZ})\in\Theta$ for each $e\in\mathcal{E}$. Then for fixed $\beta, \epsilon > 0$
\begin{equation}\label{eq_PEF_def_1}
\mathbb{P}_{\mu_{e}}\left(\mu_{e}(\mathbf{C}\lvert\mathbf{Z})\ge \left(\epsilon\prod_{i=1}^{n}F_{i}(C_{i}Z_{i})\right)^{-1/\beta}\right)\le \epsilon
\end{equation}
holds for each $e\in\mathcal{E}$, where $F_{i}(C_{i}Z_{i})$ is the probability estimation factor for the $i$'th trial.
\end{theorem*}
\begin{proof}
The sequence of random variables $\mathbf{C,Z}$ represent the time-ordered sequence of $n$ trial results. For the remainder of the proof we omit conditioning on $E=e$ since the result holds for each realisation. Hence, $\mu(\cdots)$, $\mathbb{P}_{\mu}(\cdots)$ and $\mathbb{E}_{\mu}[\cdots]$ must be understood to mean $\mu_{e}(\cdots)$, $\mathbb{P}_{\mu_{e}}(\cdots)$ and $\mathbb{E}_{\mu_{e}}[\cdots]$.

Observe that for any $i\in \{1, ..., n-1\}$ we have
\begin{equation}
\label{e:PEF_thm1_1}
\mu(\mathbf{C}_{\le i+1}\lvert\mathbf{Z}_{\le i+1}) = \mu(C_{i+1}\lvert\mathbf{C}_{\le i}\mathbf{Z}_{\le i+1})\mu(\mathbf{C}_{\le i}\lvert Z_{i+1}\mathbf{Z}_{\le i})=\mu(C_{i+1}\lvert\mathbf{C}_{\le i}\mathbf{Z}_{\le i+1})\mu(\mathbf{C}_{\le i}\lvert\mathbf{Z}_{\le i})
\end{equation}
where the first equality is an elementary manipulation of conditional probabilities and the second equality follows from  \begin{align*}\mu(\mathbf{C}_{\le i}\lvert Z_{i+1}\mathbf{Z}_{\le i})&=\mu(\mathbf{C}_{\le i} Z_{i+1}\mathbf{Z}_{\le i})/\mu( Z_{i+1}\mathbf{Z}_{\le i})\\
&=\mu(\mathbf{C}_{\le i} \mathbf{Z}_{\le i})\mu(Z_{i+1})/\mu( Z_{i+1})\mu(\mathbf{Z}_{\le i})\\
&=\mu(\mathbf{C}_{\le i}\lvert\mathbf{Z}_{\le i}),
\end{align*}
with the second step above following from from the second condition in  \eqref{e:Expt_model_assumptions}, applied directly in the numerator and in the denominator via 
$$
\mu(z_{j+1}\mathbf{z}_{\le j})=\sum_{\mathbf{c}_{\le j}}\mu(z_{j+1}\mathbf{c}_{\le j}\mathbf{z}_{\le j})=\mu(z_{j+1})\sum_{\mathbf{c}_{\le j}}\mu(\mathbf{c}_{\le j}\mathbf{z}_{\le j})=\mu(z_{j+1})\mu(\mathbf{z}_{\le j}).
$$
Now consider the sequence $Q_{i} = \mu(\mathbf{C}_{\le i}\lvert\mathbf{Z}_{\le i})^{\beta}\prod_{j=1}^{i}F_{j}$, for $i\ge 1$, where we note $Q_i$ is a random variable that is determined by $\mathbf{C}_{\le i},\mathbf{Z}_{\le i}$. 
We begin by showing that conditioned on $\mathbf{C}_{\le i},\mathbf{Z}_{\le i}$ the expectation of $Q_{i+1}$ is at most $Q_{i}$ for all $i \in \{1, ..., n-1\}$. Applying \eqref{e:PEF_thm1_1}, we can write
\begin{align}\label{e:PEF_thm1_7}
Q_{i+1} &= F_{i+1}\mu(C_{i+1}\lvert Z_{i+1}\mathbf{C}_{\le i}\mathbf{Z}_{\le i})^{\beta}\mu(\mathbf{C}_{\le i}\lvert\mathbf{Z}_{\le i})^{\beta}\prod_{j=1}^{i}F_{j} \nonumber  \\
&= F_{i+1}\mu(C_{i+1}\lvert Z_{i+1}\mathbf{C}_{\le i}\mathbf{Z}_{\le i})^{\beta}Q_{i}\nonumber,\\
\Rightarrow \mathbb{E}_{\mu}[Q_{i+1}\lvert\mathbf{C}_{\le i}\mathbf{Z}_{\le i}] &= Q_{i}\mathbb{E}_{\mu}\left[F_{i+1}\mu(C_{i+1}\lvert Z_{i+1}\mathbf{C}_{\le i}\mathbf{Z}_{\le i})^{\beta}\bigg\lvert\mathbf{C}_{\le i}\mathbf{Z}_{\le i}\right]\le Q_{i}
\end{align}
where the fact that $Q_{i}$ is determined by $\mathbf{C}_{\le i},\mathbf{Z}_{\le i}$ allows us to pull it out of the conditional expectation, and the inequality follows from the fact that $\mathbb{E}_{\mu}[F_{i+1}\mu(C_{i+1}\lvert Z_{i+1}\mathbf{c}_{\le i}\mathbf{z}_{\le i})^{\beta}]\le 1$ for all realisations $\mathbf{c}_{\le i},\mathbf{z}_{\le i}$ of $\mathbf{C}_{\le i},\mathbf{Z}_{\le i}$, as ensured by Definition \ref{PEF_def}. We remark that $Q_{i}$ is a super-martingale as indicated by the inequality in \eqref{e:PEF_thm1_7}.\footnote[2]{ The term $F_{i}\mu(C_{i}\lvert Z_{i}\mathbf{Z}_{\le i-1}\mathbf{C}_{\le i-1})^{\beta}$ is non-negative, is determined by $\mathbf{C_{\le i},\mathbf{Z}_{\le i}}$ and satisfies $\mathbb{E}_{\mu}[F_{i}\mu(C_{i}\lvert Z_{i}\mathbf{C}_{\le i-1}\mathbf{Z}_{\le i-1})^{\beta}\lvert \mathbf{C}_{\le i-1}\mathbf{Z}_{\le i-1}]\le 1$. }  Now, using the law of iterated expectation we obtain:
\begin{equation}\label{e:PEF_thm1_8}
\mathbb{E}_{\mu}[Q_{i+1}]=\mathbb{E}_{\mu}\left[\mathbb{E}_{\mu}[Q_{i+1}\lvert\mathbf{C}_{\le i}\mathbf{Z}_{\le i}]\right] \le \mathbb{E}_{\mu}[Q_{i}]
\end{equation}
Since $Q_{1}$ equals $\mu(C_1|Z_1)^\beta F(C_1Z_1)$, it satisfies $\mathbb{E}_{\mu}[Q_{1}]\le 1$ directly from Definition \ref{PEF_def}, and so repeated applications of \eqref{e:PEF_thm1_8} yield $\mathbb{E}_{\mu}[Q_{n}]\le \mathbb{E}_{\mu}[Q_{n-1}]\le\cdots\le \mathbb{E}_{\mu}[Q_{1}]\le 1$. Since $Q_{n}=\mu(\mathbf{C}\lvert\mathbf{Z})^{\beta}\prod_{i=1}^{n}F_{i}$ is non-negative, we can use Markov's inequality and obtain the required result as shown below. 
\begin{align*}
\mathbb{P}_{\mu}\left(\mu(\mathbf{C}\lvert\mathbf{Z})^{\beta}\prod_{i=1}^{n}F_{i}\ge 1/\epsilon\right) \le \epsilon\mathbb{E}_{\mu}\left[\mu(\mathbf{C}\lvert\mathbf{Z})^{\beta}\prod_{i=1}^{n}F_{i}\right]&\le\epsilon \nonumber \\
    \Rightarrow \mathbb{P}_{\mu}\left(\mu(\mathbf{C}\lvert\mathbf{Z}) \ge \left(\epsilon\prod_{i=1}^{n}F_{i}\right)^{-1/\beta}\right) &\le \epsilon.
\end{align*}
\end{proof}

Next, we present the proof for Theorem~\ref{PEF_thm2}.
\begin{theorem*}
Let $\mu$ be a distribution $\mu\colon\mathcal{C}^{n}\times\mathcal{Z}^{n}\times\mathcal{E}\to[0,1]$ of $\mathbf{CZ}E$ such that for each $e\in\mathcal{E}$,
the following holds for every $\epsilon \in (0,1)$:
\begin{equation}\label{eq_PEF_def2_1}
\mathbb{P}_{\mu_{e}}\left(\mu_{e}(\mathbf{C}\lvert\mathbf{Z})\le\left(\epsilon\prod_{i=1}^{n}F_{i}\right)^{-1/\beta}\right)\ge 1-\epsilon,
\end{equation}
where $F_{i}$ is a PEF with power $\beta$ for the $i$'th trial. For a fixed choice of $\epsilon \in (0,1)$ and $p\ge \abs{\mathcal{C}}^{-n}$, define the event $\mathsf{S}\coloneqq \left\{\left(\epsilon\prod_{i=1}^{n}F_{i}\right)^{-1/\beta}\le p\right\}$. Then if $\kappa$ is a positive number for which $\mathbb{P}_{\mu}(\mathsf{S})\ge \kappa$, the following holds:
\begin{equation}\label{eq_lb_min_ent_1}
\mathbb H_{\infty,\mu}^{\mathsf{avg},\epsilon/\kappa}(\mathbf{C}\lvert\mathbf{Z}E;\mathsf{S})\ge \log_{2}(\kappa) - \log_{2}(p)
\end{equation}
\end{theorem*}
\begin{proof}
The goal is to construct a distribution $\omega$ of $\mathbf{CZ}E$
such that it is within $\epsilon/\kappa$ TV-distance from $\mu(\mathbf{CZ}E\lvert\mathsf S)$, and such that the average conditional maximum probability of $\mathbf{C}$ conditioned on (and averaged over) $\mathbf{Z}E$ is bounded below by $p/\kappa$. We will construct $\omega$ to satisfy $\omega(\mathbf{Cz}e|\mathsf S)=0$ for all values of $\mathbf{z}$ and $e$ for which $\mu(\mathbf{z}e|\mathsf S)=0$. Hence for the rest of the construction, we will restrict attention to cases where $\mu(\mathbf{z}e|\mathsf S)>0$. 
We will use expressions such as $\mathbb{P}_{\mu_{e}}(\mathsf S)$ and $\mu_{e}(\mathsf S)$ interchangeably.

We start by defining the event $$
\mathsf R\coloneqq\left\{\mu_{e}(\mathbf{C}\lvert\mathbf{Z})\le\left(\epsilon\prod_{i=1}^{n}F_{i}\right)^{-1/\beta}\right\},
$$
whose occurrence or non-occurrence is determined by the particular realisation of $e$, $\mathbf{c}$, and $\mathbf{z}$. The event $\mathsf R$ corresponds to the desired probability bound holding; \eqref{eq_PEF_def2_1} ensures that this event occurs with high probability, and we will construct our distribution $\omega$ to, in an intuitive sense, extend this desirable behaviour from $\mathsf R\cap \mathsf S$ to all of $\mathsf S$.

We begin the construction by defining, for each fixed $e$ satisfying $\mu_e(\mathsf S)>0$, a non-negative function $f\colon\mathcal{C}^{n}\times\mathcal{Z}^{n}\to\mathbb{R}^{+}$ as shown below.
\begin{align}\label{e:PEF_thm2_1}
    f(\mathbf{cz}) = \begin{cases}
    \mu_{e}(\mathbf{cz})/\mathbb{P}_{\mu_{e}}(\mathsf S), & \mathsf{S}\cap\mathsf{R}\text{ holds}\\
    0, & \text{otherwise}
    \end{cases}
\end{align}
The weight $w$ of $f$, defined as $w(f)=\sum_{\mathbf{c,z}}f(\mathbf{cz})$, satisfies $w(f)\le 1$ as shown below:
\begin{equation*}
w(f)=\sum_{\mathbf{c,z}}f(\mathbf{cz})=\sum_{\mathbf{c,z}}\mu_{e}(\mathbf{cz})[\![\mathsf{S}\cap\mathsf{R}]\!]/\mathbb{P}_{\mu_{e}}(\mathsf S)\le\sum_{\mathbf{c,z}}\mu_{e}(\mathbf{cz})[\![\mathsf S]\!]/\mathbb{P}_{\mu_{e}}(\mathsf S)=\mathbb{P}_{\mu_{e}}(\mathsf S)/\mathbb{P}_{\mu_{e}}(\mathsf S)=1,
\end{equation*}
where $[\![\cdot]\!]$ is equal to $1$, if the condition or expression within holds, $0$ otherwise. (Note that $f$ is a \textit{sub-probability distribution} on $\mathbf{cz}$: a set of non-negative numbers whose sum is less than or equal to 1. Defining a sub-probability distribution is a standard trick to construct a distribution by invoking certain lemmas.) 

Below we show that $w$ satisfies $w(f)\ge 1-\epsilon/\mathbb{P}_{\mu_{e}}(\mathsf S)$.
\begin{align}\label{e:lwr_bnd_wf}
    w(f) &= 1-1+\sum_{\mathbf{c,z}}\mu_{e}(\mathbf{cz})[\![\mathsf S]\!][\![\mathsf R]\!]/\mathbb{P}_{\mu_{e}}(\mathsf S) = 1 - \sum_{\mathbf{c,z}}\mu_{e}(\mathbf{cz})[\![\mathsf S]\!]/\mathbb{P}_{\mu_{e}}(\mathsf S) + \sum_{\mathbf{c,z}}\mu_{e}(\mathbf{cz})[\![\mathsf S]\!][\![\mathsf R]\!]/\mathbb{P}_{\mu_{e}}(\mathsf S) \nonumber \\
    &= 1 - \sum_{\mathbf{c,z}}\left(\mu_{e}(\mathbf{cz}) - \mu_{e}(\mathbf{cz})[\![\mathsf R]\!]\right)[\![\mathsf S]\!]/\mathbb{P}_{\mu_{e}}(\mathsf S)\ge 1 - \sum_{\mathbf{c,z}}\left(\mu_{e}(\mathbf{cz})-\mu_{e}(\mathbf{cz})[\![\mathsf R]\!]\right)/\mathbb{P}_{\mu_{e}}(\mathsf S) \nonumber \\ 
    &= 1-\left(1 - \mathbb{P}_{\mu_{e}}(\mathsf R)\right)/\mathbb{P}_{\mu_{e}}(\mathsf S)\ge 1-\epsilon/\mathbb{P}_{\mu_{e}}(\mathsf S)
\end{align}
where in \eqref{e:lwr_bnd_wf} we have used the fact that $\mathbb{P}_{\mu_{e}}(\mathsf R)\ge 1-\epsilon$ holds for each $e\in\mathcal{E}$, as follows from \eqref{eq_PEF_def2_1}. Next, we define a non-negative function $\tilde{f}_{\mathbf{z}}\colon\mathcal{C}^{n}\to\mathbb{R}^{+}$ for each $\mathbf{z}\in\mathcal{Z}^{n}$ for which $\mu_e(\mathbf{z}\lvert \mathsf S)>0$:
\begin{equation}\label{e:PEF_thm2_2}
\tilde{f}_{\mathbf{z}}(\mathbf{c}) =  f(\mathbf{cz})/\mu_{e}(\mathbf{z}\lvert\mathsf S)
\end{equation}

We show below that for each such $\mathbf{z}$, $\tilde{f}_{\mathbf{z}}(\mathbf{c})$ is bounded by $\mu_{e}(\mathbf{c}\lvert\mathbf{z},\mathsf S),\,\forall\mathbf{c}\in\mathcal{C}^{n}$. We have:
\begin{align}\label{e:bnd_ftilde}
    \tilde{f}_{\mathbf{z}}(\mathbf{c}) &= \mu_{e}(\mathbf{cz})[\![\mathsf S]\!][\![\mathsf R]\!]/(\mathbb{P}_{\mu_{e}}(\mathsf S)\mu_{e}(\mathbf{z}\lvert\mathsf S)) 
    = \mu_{e}(\mathbf{cz},\mathsf S)[\![\mathsf R]\!]/\mu_{e}(\mathbf{z}, \mathsf S)\notag\\
    &\le \mu_{e}(\mathbf{cz},\mathsf S)/\mu_{e}(\mathbf{z},\mathsf S)=\mu_{e}(\mathbf{c}\lvert\mathbf{z},\mathsf S), 
\end{align}
where the equality $\mu_{e}(\mathbf{cz})[\![\mathsf S]\!]=\mu_{e}(\mathbf{cz},\mathsf S)$ makes sense because whether or not $\mathsf S$ holds is determined by $\mathbf{cz}$. Since \eqref{e:bnd_ftilde} holds for all $\mathbf{c}$, we conclude $\tilde{f}_{\mathbf{z}}(\mathbf{C})\le\mu_{e}(\mathbf{C}\lvert\mathbf{z},\mathsf S)$. This proves that $\tilde{f}_{\mathbf{z}}(\mathbf{C})$ is dominated by $\mu_{e}(\mathbf{C}\lvert\mathbf{z},\mathsf{S})$. From the definition of $\tilde{f}_{\mathbf{z}}$ we also have another upper bound for all $\mathbf{c}$: 
$$
\tilde{f}_{\mathbf{z}}(\mathbf{c})=\mu_{e}(\mathbf{c}\lvert\mathbf{z})\mu_{e}(\mathbf{z})[\![\mathsf S \cap \mathsf R]\!]/\mu_{e}(\mathbf{z},\mathsf S) \le p\mu_{e}(\mathbf{z})/\mu_{e}(\mathbf{z,\mathsf S}).
$$
Above, we have used the fact that the event $\mathsf S \cap \mathsf R$ implies $\mu_{e}(\mathbf{C}\lvert\mathbf{Z})\le\left(\epsilon\prod_{i=1}^{n}F_{i}\right)^{-1/\beta}\le p$. The bound $p\mu_{e}(\mathbf{z})/\mu_{e}(\mathbf{z},\mathsf S)\ge p\ge\abs{\mathcal{C}}^{-n}$ also holds, since $\mu_{e}(\mathbf{z})/\mu_{e}(\mathbf{z},\mathsf S)\ge 1$. Hence, using the lemmas in Section \ref{s:Lemmas} we can construct, for each $\mathbf z$ under consideration, a distribution $\mu^{\prime}_{\mathbf{z}}(\mathbf{C})$ such that $\mu^{\prime}_{\mathbf{z}}(\mathbf{C})\ge \tilde{f}_{\mathbf{z}}(\mathbf{C})$, $\mu^{\prime}_{\mathbf{z}}(\mathbf{C})\le p\mu_{e}(\mathbf{z})/\mu_{e}(\mathbf{z},\mathsf S)$ and $\mathrm{TV}(\mu^{\prime}_{\mathbf{z}}(\mathbf{C}),\mu_{e}(\mathbf{C}\lvert\mathbf{z},\mathsf S))\le 1-w(\tilde{f}_{\mathbf{z}})$, where $w(\tilde{f}_{\mathbf{z}})\le 1$ is the weight of $\tilde{f}_{\mathbf{z}}(\mathbf{C})$. Now we are ready to define the distribution $\omega(\mathbf{CZ}E)$ as 
\begin{equation*}
\omega(\mathbf{cz}e)=\begin{cases}\mu^{\prime}_{\mathbf{z}}(\mathbf{c})\mu_{e}(\mathbf{z}\lvert\mathsf S)\mu(e\lvert\mathsf S) & \text{ if }\mu(\mathbf{z}e|\mathsf S) >0\\
0 &\text{ if }\mu(\mathbf{z}e|\mathsf S) =0
\end{cases}
\end{equation*}
We show that the total variation distance between $\omega$ and $\mu(\mathbf{CZ}E\lvert\mathsf S)$ is bounded by $\epsilon/\kappa$ and that the average $\mathbf{ze}$-conditional maximum probability of $\mathbf C$ is bounded by $p/\kappa$. First,
\begin{align}
    &d_\mathrm{TV}(\omega(\mathbf{CZ}E),\mu(\mathbf{CZ}E\lvert\mathsf S))= \frac{1}{2}\sum_{\mathbf{cz}e}\abs{\omega(\mathbf{cz}e)-\mu(\mathbf{cz}e\lvert\mathsf S)}\nonumber \\
    &= \frac{1}{2}\sum_{e:\mu(e\lvert \mathsf S)>0}\,\sum_{\mathbf{z}:\mu_e(\mathbf{z}\lvert \mathsf S)>0}\sum_{\mathbf{c}}\abs{\mu_{\mathbf{z}}^{\prime}(\mathbf{c})-\mu(\mathbf{c}\lvert\mathbf{z}e,\mathsf S)}\mu_{e}(\mathbf{z}\lvert\mathsf S)\mu(e\lvert \mathsf S)\label{e:tvdist.5} \\
    &\le \frac{1}{2}\sum_{e:\mu(e\lvert \mathsf S)>0}\,\sum_{\mathbf{z}:\mu_e(\mathbf{z}\lvert \mathsf S)>0}\sum_{\mathbf{c}}\left(\mu_{\mathbf{z}}^{\prime}(\mathbf{c})-\tilde{f}(\mathbf{c}) + \mu(\mathbf{c}\lvert\mathbf{z}e,\mathsf S)-\tilde{f}(\mathbf{c})\right)\mu_{e}(\mathbf{z}\lvert \mathsf S)\mu(e\lvert \mathsf S)\label{e:TVdist_1} \\
    &= \frac{1}{2}\sum_{e:\mu(e\lvert \mathsf S)>0}\,\sum_{\mathbf{z}:\mu_e(\mathbf{z}\lvert \mathsf S)>0}\mu_{e}(\mathbf{z}\lvert \mathsf S)\mu(e\lvert \mathsf S)\left(1-\sum_{\mathbf{c}}\tilde{f}(\mathbf{c}) + 1 -\sum_{\mathbf{c}}\tilde{f}(\mathbf{c}) \right)\label{e:TVdist_1.5} \\
    &=\frac{1}{2}\sum_{e:\mu(e\lvert \mathsf S)>0}2\Big(1-\sum_{\mathbf{cz}}f(\mathbf{cz})\Big)\mu(e\lvert\mathsf S)\label{e:Tvdist_2}\\
    &=\sum_{e:\mu(e\lvert \mathsf S)>0}(1-w(f))\mu(e\lvert\mathsf S) \le \sum_{e:\mu(e\lvert \mathsf S)>0}\frac{\epsilon}{\mathbb{P}_{\mu_{e}}(\mathsf S)}\mu(e\lvert\mathsf S) =\sum_{e:\mu(e\lvert \mathsf S)>0}\frac{\epsilon\mu(e)}{\mu(\mathsf S)}\le \epsilon/\mathbb{P}_{\mu}(\mathsf S)\le\epsilon/\kappa \label{e:TVdist_3}
\end{align}
The equality in \eqref{e:tvdist.5} follows because $\omega(\mathbf{cz}e)=\mu(\mathbf{cz}e|\mathsf S)=0$ for values of $e, \mathbf z$ removed from the sums, and $\mu'_\mathbf c (\mathbf c)$ is defined for the remaining values of $e, \mathbf z$. In \eqref{e:TVdist_1} we add and subtract with $\tilde{f}(\mathbf{c})$ inside the absolute value expression in the previous step and use the triangle inequality, following which we use the facts established above that both $\mu_{\mathbf{z}}^{\prime}(\mathbf{C})$ and $\mu(\mathbf{C}\lvert\mathbf{z}e,\mathsf S)=\mu_e(\mathbf{C}\lvert\mathbf{z},\mathsf S)$ dominate $\tilde{f}_{\mathbf{z}}(\mathbf{C})$. \eqref{e:TVdist_1.5} follows from the fact that $\mu_{\mathbf{z}}^{\prime}(\mathbf{c})\mu_{e}(\mathbf{z}\lvert\mathsf A)$ and $\mu(\mathbf{c}\lvert\mathbf{z}e,\mathsf S)$ sum to $1$ over $\mathbf{c}$ (being distributions), and \eqref{e:Tvdist_2} follows from $\tilde{f}_{\mathbf{z}}(\mathbf{c})=f(\mathbf{cz})/\mu_{e}(\mathbf{z}\lvert\mathsf S)$ and the fact that $f(\mathbf{cz})=0$ in cases where $\mu_e(\mathbf{z}\lvert\mathsf{S})=0$. Finally, the first inequality in \eqref{e:TVdist_3} follows from \eqref{e:lwr_bnd_wf} and the last inequality follows from $\mathbb{P}_{\mu}(\mathsf S)\ge\kappa$. Next, we show the upper bound on the average conditional maximum probability.
\begin{align}
\sum_{\mathbf{z}e}\left(\max_{\mathbf{c}}\omega(\mathbf{c}\lvert\mathbf{z}e)\right)\omega(\mathbf{z}e) &= \sum_{\mathbf{z}e}\left(\max_{\mathbf{c}}\mu_{\mathbf{z}}^{\prime}(\mathbf{c})\right)\mu_{e}(\mathbf{z}\lvert\mathsf S)\mu(e\lvert \mathsf S)\nonumber\\
    &\le \sum_{\mathbf{z}e}\frac{\mu_{e}(\mathbf{z})p}{\mu_{e}(\mathbf{z},\mathsf S)}\mu_{e}(\mathbf{z}\lvert\mathsf S)\mu(e\lvert \mathsf S) =  \sum_{\mathbf{z}e}\frac{\mu_{e}(\mathbf{z})p}{\mu_{e}(\mathsf S)}\mu(e\lvert\mathsf S)\nonumber \\
    &=\sum_{e}\frac{p}{\mathbb{P}_{\mu_{e}}(\mathsf S)}\mu(e\lvert\mathsf S) = \frac{p}{\mathbb{P}_{\mu}(\mathsf S)}\le\frac{p}{\kappa}\label{e:maxprob_bnd}\\
\Rightarrow -\log_{2}\left[\sum_{\mathbf{z}e}\left(\max_{\mathbf{c}}\omega(\mathbf{c}\lvert\mathbf{z}e)\right)\omega(\mathbf{z}e)\right] &\ge \log_{2}(\kappa) - \log_{2}(p) \label{e:neglog_maxprob_bnd}
\end{align}

Hence, we have obtained an upper bound on the average conditional maximum probability in \eqref{e:maxprob_bnd}. Since by definition the $\epsilon/\kappa$-smooth average conditional min-entropy involves a maximum (over the set $\mathcal{B}^{\epsilon/\kappa}(\mu)$) of the quantity on the left hand side of \eqref{e:neglog_maxprob_bnd}, the final result follows.
\end{proof}

\section{Proofs using Convex Geometry}\label{a:convex}
Here we prove Theorems~\ref{thm_hmin_achieved} and~\ref{t:est} using arguments from convex geometry.

\begin{theorem*} Suppose $\Pi$ is closed and equal to the convex hull of its extreme points. Then there is a distribution $\omega(CZE)\in\Sigma_{E}$ with $\abs{\mathcal{E}}=1+\dim \Pi$ such that 
$H_{\mu}(C\lvert ZE)=h_{\min}(\rho(CZ))$.
\end{theorem*}
\begin{proof}
We will be analysing $h_{\text{min}}(\cdot)$ as a function with domain $\Pi$. It is useful to re-write $h_{\text{min}}(\cdot)$ in the form
$$
h_{\text{min}}(\rho) = \inf_{\{\sigma_i\}_{i\in I}} \sum_i p_iH_{\sigma_i}(C|Z),
$$
where the infimum is taken over all finite subsets $\{\sigma_i\}_{i\in I}\subseteq\Pi$ for which $\sum_{i\in I} p_i \sigma_i = \rho$ for some collection of non-negative $p_i$ summing to 1.\footnote[2]{This is equivalent to the earlier definition if we set $\omega(CZ|e_i)=\sigma_i(CZ)$ and $\omega(e_i)=p_i$, yielding $H_\omega (C|ZE) = \sum_{i,j} H_{\omega(C|z_j,e_i)}(C)\omega(z_j,e_i)=\sum_i \left(\sum_j H_{\omega(C|z_j,e_i)}(C)\omega(z_j|e_i) \right)\omega(e_i)=\sum_i \left(\sum_j H_{\sigma_i(C|z_j)}(C)\sigma_i(z_j) \right)p_i=\sum_i p_iH_{\sigma_i}(C|Z).$}
We first observe that the scope of the infimum can be reduced to consider only sets of $\sigma_i$ belonging to $\Pi_{\mathrm{extr}}$, the set of extreme points of $\Pi$. This follows from the fact that conditional Shannon entropy is concave.\footnote[2]{The proof of theorem 43 in~\cite{PhysRevResearch.2.033465} correctly notes that the concavity of conditional Shannon entropy can be obtained as a specialisation of the concavity of the quantum conditional entropy. It is worth noting, however, that the classical (only) result can be obtained much more quickly and directly as shown in Appendix~\ref{s:concavity}} Hence any expression in the scope of the infimum defining $h_{\text{min}}$ can always be decreased (or at least unchanged) by replacing each $\sigma_i$ in the expression with a convex combination of extremal behaviours replicating $\sigma_i$.

$\Pi$ is a subset of $\mathbb R^N$ where $N=|Z|\times|C|$ is the number of conditional probabilities appearing in the behaviour. In general, $N$ is strictly larger than $\dim \Pi$: the constraint that certain elements of $\Pi$ need to form valid probability distributions reduces the dimension, and no-signalling equalities can reduce the dimension further. So we seek to re-parametrise the elements of $\Pi$ using only the number of coordinates necessary based on its dimension. The (affine) dimension of $\Pi$ is by definition the dimension of the smallest affine space containing it -- that is, the intersection of all affine subspaces of $\mathbb R^N$ containing $\Pi$, which is itself affine space. Let us call this smallest affine space $\mathcal A$. If $\dim \mathcal A=m$, then there is a set of $m$ linearly independent vectors $\vec v_i$ and a displacement/base vector $\vec b$ such that any $\sigma \in  \mathcal A$ has a unique representation as 
\begin{equation}\label{e:uniquerep}
\sigma = \vec b + \sum_{i=1}^n c_i \vec v_i.
\end{equation}
For any $\sigma \in \Pi \subseteq \mathcal A$, then, we can uniquely represent $\sigma$ as a vector of these coefficients, $(c_1, c_2, ..., c_m)$. 

We would like to construct an affine-linear map $g:\mathbb R^N \to \mathbb R^m$ whose restriction to $\mathcal A$ maps the $N$-coordinate vector $\sigma$ to its $m$-coordinate representation $(c_1, c_2, ..., c_m)$.\footnote[2]{Our approach here makes explicit the arguments only alluded to in the proof of Theorem 43 in~\cite{PhysRevResearch.2.033465} through general referral to existence and extension theorems in convex analysis, and takes full advantage of the fact that we are always working in a large ambient $\mathbb R^n$, allowing us to harness the strength of linear algebra.} Our affine-linear map will be represented by a matrix $M$ and a vector $\vec k$ such that $g(\sigma) = M \sigma + \vec k = (c_1, ... , c_m)$. To construct $M$ and $\vec k$, let $A$ be the $N\times m$ matrix whose $m$ columns are the vectors $\vec v_i $ appearing in \eqref{e:uniquerep}. 
Since the columns of $A$ are linearly independent, $A^TA$ is invertible as its kernel consists only of the zero vector:
\begin{equation*}
A^TA \vec v = \vec 0 \,\Rightarrow \,\vec v^T A^TA \vec v = 0 \, \Rightarrow \,(A\vec v)^T A \vec v = 0\, \Rightarrow \,\norm{A\vec v} = 0 \,\Rightarrow \, A\vec v = \vec 0 \,\Rightarrow \,\vec v = \vec 0
\end{equation*}
We can thus define $M = (A^TA)^{-1}A^T$ which will satisfy $MA = I$ ($M$ is a \textit{pseudo-inverse} of $A$), and so $M$ maps the vectors $\vec v_i$ to the standard basis vectors in $\mathbb R^m$. Setting $\vec k = -M\vec b$ yields the desired $g(\cdot)$.

We  point out a couple of properties of $g$ that we will use in our arguments. First, it commutes with convex combinations: For a set of non-negative $p_i$ satisfying $\sum_i p_i = 1$ and a collection of elements $\sigma_i$ of $\mathcal A$,  
\begin{equation}\label{e:convcomm}
\sum_i p_i g(\sigma_i)=g\left(\sum_i p_i \sigma_i  \right),
\end{equation}
 which follows directly from expressing $g$ as $M(\cdot) + \vec k$ and noticing that $\sum_i p_i \vec k = \vec k$. Second, $M$ is injective when restricted to $\mathcal A$, so consequently $g$ is a bijection between $\mathcal A$ and $R^m$ and in particular
\begin{equation}\label{e:comt}
g\left(\sum p_i \sigma_i\right ) = g(\sigma) \quad \Leftrightarrow \quad \sum p_i \sigma_i =\sigma.
\end{equation}

Now, let us consider the following subset of $\mathbb R^{m+1}$,\footnote[2]{The development here is inspired by the arguments in the appendix of \cite{uhlmann:1998}, though the assumptions and conclusions differ somewhat}
$$
\Xi_{\mathrm{extr}}=\{(g(\sigma), H_\sigma(C|Z)):\sigma \in \Pi_{\text{extr}}\},
$$
where the first $m$ coordinates of an element of $\Xi_{\mathrm{extr}}$ are the coordinates of $g(\sigma)$ and the $m+1$ coordinate is $H_\sigma (C|Z)$. Define
$$
\Xi=\text{conv}\left(\Xi_{\mathrm{extr}}\right),
$$
where `conv' denotes the convex hull. $\Xi_{\mathrm{extr}}$ and $\Xi$ are artificial constructions, but by studying their geometry we can prove the existence of a convex combination achieving the infimum defining $h_{\min}(\rho)$.

We first confirm that $\Xi_{\mathrm{extr}}$ is indeed the set of extremal points $\Xi$ (as suggested by our choice in names), i.e., we confirm that  $\Xi_{\mathrm{extr}}$ contains only trivial convex combinations of its elements. To see this, note if $
\sum_i p_i (g(\sigma_i), H_{\sigma_i}(C,Z)) = (g(\sigma), H_{\sigma}(C,Z))
$ holds for some $\sigma_i, \sigma \in \Pi_{\mathrm{extr}}$ and non-negative $p_i$ satisfying $\sum_i p_i = 1$, then we must have $\sum_i p_i g(\sigma_i) = g(\sigma)$ and so $\sum_i p_i \sigma_i = \sigma$ by~\eqref{e:convcomm} and~\eqref{e:comt}. This can only be a trivial convex combination (i.e., all $\sigma_i$ with nonzero $p_i$ coefficient must equal $\sigma$) as the $\sigma_i$ and $\sigma$ are assumed to be in $\Pi_{\mathrm{extr}}$. 

Second, we show that the point $(g(\rho), h_{\min} (\rho))$ is on the boundary of $\Xi$; i.e., that $(g(\rho), h_{\min} (\rho))$ is a limit point of $\Xi$ and also a limit point of $\Xi^C$. To see that we can converge to this point from \textit{within} the set, note that for any set of $\sigma_i \in  \Pi_{\mathrm{extr}}$ satisfying $\sum_i p_i \sigma_i = \rho$, we have by definition $\sum_i p_i (g(\sigma_i), H_{\sigma_i}(C|Z))\in \Xi$ which can be re-expressed as 
$\left(g(\rho), \sum_i p_i H_{\sigma_i}(C|Z)\right)\in \Xi$
by invoking~\eqref{e:convcomm}. By the nature of the infimum defining $h_{\min} (\rho)$, there must be a sequence of such elements of $\Xi$ whose last component forms a non-increasing sequence converging to $h_{\min}(\rho)$; since the first $m$ components are identically $g(\rho)$, this sequence converges to $(g(\rho), h_{\min}(\rho))$ as desired. Similarly, one can also converge to $(g(\rho), h_{\text{min}}(\rho))$ from \textit{outside} the set $\Xi$: $(g(\rho), h_{\text{min}}(\rho)-\epsilon) \notin \Xi$ for all $\epsilon>0$; this is because all elements of $\Xi$ take the form
\begin{eqnarray*}
\sum_i p_i(g(\sigma_i), H_{\sigma_i}(C|Z))=\left(\sum_i p_ig(\sigma_i), \sum_i p_iH_{\sigma_i}(C|Z)\right),
\end{eqnarray*}
for some collection $\sigma_i \in \Xi_{\mathrm{extr}}$ and if the first $m$ coordinates are equal to $g(\rho)$, then by~\eqref{e:convcomm} and~\eqref{e:comt} we must have $\sum_i p_i \sigma_i=\rho$ and so the $m+1$ coordinate is a term contributing to the infimum defining $h_{\text{min}}(\rho)$; it cannot be less than $h_{\text{min}}(\rho)$.

We now would like to demonstrate that $(g(\rho), h_{\text{min}}(\rho))$ is contained in $\Xi$. As a first step, we show that 
\begin{equation}\label{e:inconv}
(g(\rho), h_{\text{min}}(\rho))\in \text{conv}(\overline{ \Xi_{\mathrm{extr}}}),
\end{equation}
where $\text{conv}(\overline{ \Xi_{\mathrm{extr}}})$ denotes the convex hull of the closure of $\Xi_{\mathrm{extr}}$. To see this, first note that $\Xi_{\mathrm{extr}}$ is bounded -- for the $m+1$ coordinate, Shannon entropy is non-negative with a maximum value set by the cardinality of the value space of $C$, and for the first $m$ coordinates, these are contained in the image of the set $\Pi_{\mathrm{extr}}$ through the continuous map $g$ -- and since  $\Pi_{\mathrm{extr}}$ is contained in the compact set $\mathcal P = [0,1]^n$ ($\mathcal P$ contains all probability distributions), its image must be contained in the compact (and thus bounded) set $g(\mathcal P)$. As $\Xi_{\mathrm{extr}}$ is bounded, its closure, denoted $\overline {\Xi_{\mathrm{extr}}}$, must be bounded as well and so is compact. It is a known fact that the convex hull of a compact set in $\mathbb R^n$ is compact, 
so $\text{conv}(\overline{ \Xi_{\mathrm{extr}}})$ is compact -- and so in particular, closed. Finally $\text{conv}(\overline{\Xi_{\mathrm{extr}}})$ clearly contains $\Xi=\text{conv}(\Xi_{\mathrm{extr}})$, the convex hull of a smaller set; as a closed set containing $\Xi$, it will contain the $\Xi$-boundary point $(g(\rho), h_{\text{min}}(\rho))$.

Now we show that this implies containment in $\Xi$ proper. Since the map
$$
h(\rho) := (g(\rho),H_{\rho}(C|Z))
$$
with image in $\mathbb R^{m+1}$ is continuous on the domain of $n$-dimensional probability distributions and  $\Pi_{\mathrm{extr}}$ is bounded, we have $\overline{h(\Pi_{\mathrm{extr}})}\subseteq h(\overline{\Pi_{\mathrm{extr}}})$\footnote[2]{For any bounded subset $S$ in $\mathbb R^n$ (like $\Pi_{\mathrm{extr}}$) and continuous $h$, we have $\overline{h(S)} \subseteq h(\overline {S})$. Proof: Any $x\in \overline{h(S)}$ must be the limit of a sequence in $h(S)$; let $\{s_i\}_{i=1}^\infty \subseteq S$ satisfy $h(s_i) \to x$. Since $\overline S$ is compact, $\{s_i\}_{i=1}^\infty\subseteq S \subseteq \overline S$ has a convergent sub-sequence $\{s_j\}_{j=1}^\infty$ with limit in $\overline S$; let $s\in \overline S$ be this limit. By continuity, $h(s_j) \to h(s)$; considered as a sub-sequence of $\{h(s_i)\}_{i=1}^\infty$, we also have $h(s_j) \to x$ and so uniqueness of limits implies $x=h(s) \in h(\overline S)$.} and since by definition $\Xi_{\mathrm{extr}}= h(\Pi_{\mathrm{extr}})$, we write \begin{equation}\label{e:closcont}
\overline{\Xi_{\mathrm{extr}}}\subseteq h(\overline{\Pi_{\mathrm{extr}}}).
\end{equation}
Now using \eqref{e:inconv}, \eqref{e:closcont}, the definition of $h(\cdot)$, and finally \eqref{e:convcomm}, we can write
\begin{eqnarray}
(g(\rho), h_{\text{min}}(\rho)) &=& \sum_i p_i \vec w_i \text{ for some }\{\vec w_i\}_{i\in I} \subseteq \overline{ \Xi_{\mathrm{extr}}}\notag\\
&=& \sum_i p_i h(\tau_i) \text{ for some }\{\tau_i\}_{i\in I} \subseteq \overline{ \Pi_{\mathrm{extr}}}\notag\\
&=& \sum_i p_i (g(\tau_i), H_{\tau_i}(C|Z)) \text{ for some }\{\tau_i\}_{i\in I} \subseteq \overline{ \Pi_{\mathrm{extr}}}\notag\\
&=& \left(g\left(\sum_i p_i \tau_i\right),\sum_i p_i H_{\tau_i}(C|Z)\right) \text{ for some }\{\tau_i\}_{i\in I} \subseteq \overline{ \Pi_{\mathrm{extr}}}. \label{e:recday}
\end{eqnarray}
Comparing the first expression in the above sequence to the last and applying \eqref{e:comt} implies that $\sum_i p_i \tau_i=\rho$. Now since by assumption $\Pi$ is closed, $ \Pi_{\mathrm{extr}} \subseteq \Pi$ implies $\overline{\Pi_{\mathrm{extr}}} \subseteq \Pi$, so $\Pi  = \text{conv}(\Pi_{\mathrm{extr}})$ implies that elements of $\overline{\Pi_{\mathrm{extr}}} $ can be expressed as convex combinations of elements of $\Pi_{\mathrm{extr}}$. Thus in the expression $\sum_i p_i \tau_i$, if there are any non-extremal $\tau_i$ elements they can be replaced with convex combinations of elements of $\Pi_{\mathrm{extr}}$ to yield a convex combination $\sum_j q_j \sigma_j$ equalling $\rho$ where the concavity of conditional Shannon entropy implies that $\sum_j q_j H_{\sigma_j}(C|Z)$ is not larger than $\sum_i p_i H_{\tau_i}(C|Z)$. However, by~\eqref{e:recday}, $\sum_i p_i H_{\tau_i}(C|Z)=h_{\text{min}}(\rho)$ and since $\sum_j q_j H_{\sigma_j}(C|Z)$ cannot be smaller than $h_{\min}(\rho)$, it must equal $h_{\min}(\rho)$. As $\sum_j q_j \sigma_j=\rho$ and $\sum_j q_j H_{\sigma_j}(C|Z)=h_{\min}(\rho)$, one more application of \eqref{e:convcomm} yields
\begin{eqnarray*}
(g(\rho), h_{\text{min}}(\rho)) &=&  \left(g\left(\sum_j q_j \sigma_j\right),\sum_j q_j H_{\sigma_j}(C|Z)\right) \text{ for some }\{\sigma_j\}_{j\in J} \subseteq \Pi_{\mathrm{extr}}\\
&=& \sum_j q_j \left( g(\sigma_j), H_{\sigma_j}(C|Z)\right) \text{ for some }\{\sigma_j\}_{j\in J} \subseteq \Pi_{\mathrm{extr}},
\end{eqnarray*}
which is in $\Xi$.

The argument thus far demonstrates the existence of a convex combination of $\Pi_{\mathrm{extr}}$ elements explicitly achieving the infimum in the definition of $h_{\min}(\rho)$. We continue with our argument to further demonstrate that the number of required $ \Pi_{\mathrm{extr}}$ elements in such an optimal decomposition is not greater than $m+1$.

We first note that since $(g(\rho), h_{\text{min}}(\rho))$ is on the boundary of the convex set $\Xi$, the supporting hyperplane theorem says there is a supporting hyperplane $H_\rho$ with $(g(\rho), h_{\text{min}}(\rho)) \in H_\rho$ and $\Xi$ entirely on one side of $H_\rho$. Now, notice that if we decompose $(g(\rho), h_{\text{min}}(\rho))$ as a convex combination of $\Xi_{\mathrm{extr}}$ elements, these elements must all lie in the hyperplane $H_\rho$: this is because any elements strictly on one side of $H_\rho$ would have to be counterbalanced by elements strictly on the other side of $H_\rho$ -- but since one side of $H_\rho$ is disjoint from $\Xi$, this is not possible. Applying the same observation to any other element of $H_\rho \cap \Xi$, it follows that $H_\rho \cap \Xi$ is contained in the convex hull of $H_\rho\cap\Xi_{\mathrm{extr}}$. As the reverse inclusion follows from the convexity of $H_\rho$ and the fact that $\Xi=\mathrm{conv}(\Xi_{\mathrm{extr}})$, we can write $\mathrm{conv}(H_\rho\cap\Xi_{\mathrm{extr}})=H_\rho\cap\Xi$. Now since $H_\rho \cap \Xi$ is at most $m$ dimensional ($H_\rho$, as a hyperplane, has one fewer dimension than the ambient $(m+1)$-dimensional space), we can invoke Carath\'eodory's theorem to see that at most $m+1$ points of $H_\rho\cap\Xi_{\mathrm{extr}}$ are required to replicate $(g(\rho), h_{\text{min}}(\rho))$ as a convex combination. Thus we have 
\begin{equation*}
(g(\rho), h_{\text{min}}(\rho) )= \sum_i p_i \vec w_i \text{ for some }\{\vec w_i\}_{i\in I} \subseteq \Xi_{\mathrm{extr}}, |I| \le m+1\\
\end{equation*}
and so recalling the definition of $\Xi_{\mathrm{extr}}$ and invoking~\eqref{e:convcomm} one last time, we can write that for some integer $m^*$ satisfying $1\le m^*\le m+1$,
\begin{eqnarray}
(g(\rho), h_{\text{min}}(\rho)) 
&=& \sum_{i=1}^{m^*}p_i (g(\sigma_i), H_{\sigma_i}(C|Z)) \text{ for some }\{\sigma_i\}_{i=1}^{m^*} \subseteq  \Pi_{\mathrm{extr}}\notag\\
&=& \left(g\left(\sum_{i=1}^{m^*}p_i \sigma_i\right),\sum_{i=1}^{m^*} p_i H_{\sigma_i}(C|Z)\right) \text{ for some }\{\sigma_i\}_{i=1}^{m^*} \subseteq \Pi_{\mathrm{extr}}. \notag
\end{eqnarray}
By~\eqref{e:comt}, $\sum_{i=1}^{m^*}p_i \sigma_i=\rho$ and so $\{\sigma_i\}_{i=1}^{m^*}$ induces the desired distribution $\omega(CZE)$ by setting $\omega(CZ|e_i)=\sigma_i(CZ)$ and $\omega(e_i)=p_i$.
\end{proof}

\begin{theorem*}
Suppose $\Pi$ satisfies the conditions of Theorem~\ref{thm_hmin_achieved} and $\rho$ is in the interior of $\Pi$. Then there exists an entropy estimator whose entropy estimate at $\rho$ is equal to $h_{\min}(\rho)$.
\end{theorem*}

\begin{proof}
We continue from where we left off in the proof of Theorem~\ref{thm_hmin_achieved}, and show that the supporting hyperplane $H_\rho$ discussed in that proof can be used to construct an affine function that is the desired entropy estimator.  Recall that the dimension of $\Pi$, which is embedded in a higher dimensional vector space $\mathbb{R}^{N}$, is defined as the affine dimension of $\mathcal A$, the smallest affine subspace containing $\Pi$. Given this context, the assumption that $\rho$ is in the interior of $\Pi$ means that there exists an open $\epsilon$-ball $U$ in $\mathbb{R}^{N}$ such that $U \cap \mathcal A $, which is open in the subspace topology, is contained in $\Pi$.\footnote[2]{If this assumption is removed, a weaker form of the theorem demonstrating the existence of entropy estimators with estimate $\epsilon$-close to $h_{\min}(\rho)$ can be proved with a similar argument to that of the current proof by invoking Exercise 3.28 of \cite{BV}}

 First, we note that $g(\rho)$ is in the interior of $g(\Pi)$. To see this, consider the restriction $g\!\!\restriction_A$ of $g$ to $\mathcal A$, which is a bijection with affine-linear inverse map $(g\!\!\restriction_{A} )^{-1}:\mathbb R^m \to \mathcal A$ given by $A(\cdot) + \vec b$ (recalling the construction following~\eqref{e:uniquerep} in the proof of Theorem~\ref{thm_hmin_achieved}). This ensures that the set $g\!\!\restriction_A(U\cap \mathcal A)$ must be open, as it is equal to the inverse image of $U\cap \mathcal A$ under the map $(g\!\!\restriction_A)^{-1}$ which is equal to the inverse image of the open set $U$ under the (continuous) map $A(\cdot) + \vec b :\mathbb R^m \to \mathbb R^N$. Hence $g(\rho)$ is contained in the \textit{open} set $g(U\cap \mathcal A)$ which is a subset of $g(\Pi)$ as $U \cap \mathcal A \subseteq \Pi$.
 
Now we take a closer look at $H_\rho$, the supporting hyperplane touching $\Xi$ at $(g(\rho), h_{\text{min}}(\rho))$. As a hyperplane, $H_\rho$ will be equal to the set of $\vec x$ satisfying an equation of the form 
$
\vec a \cdot \vec x = b
$
for some fixed $\vec a \in \mathbb R^{m+1}$ and $b \in \mathbb R$, where $\cdot$ denotes the dot product, and the condition
\begin{equation}\label{e:oneside}
\xi \in \Xi \Rightarrow \vec a \cdot \xi \ge b
\end{equation}
expresses algebraically the notion that $\Xi$ is on one side of $H_\rho$.
We argue that the fact that $g(\rho)$ is in the interior of $g(\Pi)$ implies the $m+1$ component of $\vec a$, denoted $\vec a_{m+1}$, must be nonzero. Assume $\vec a_{m+1}=0$ for a proof by contradiction: since $(g(\rho), h_{\min} (\rho))$ is the point of contact of the supporting hyperplane $H_\rho$, we have $\vec a \cdot (g(\rho), h_{\min} (\rho)) = b$, which implies $\vec a_{[m]}\cdot g(\rho) = b$
where $\vec a_{[m]}\in \mathbb R^m$ denotes the vector consisting of the first $m$ coordinates of $\vec a$. Since the previous paragraph demonstrated there is an open subset of $g(\Pi)$ containing $g(\rho)$, this means $g(\rho)-c\vec a_{[m]}$ for a sufficiently small positive $c$ is equal to $g(\phi)$ for some $\phi$ in $\Pi$. By construction $\phi$ will satisfy $\vec a_{[m]}\cdot g(\phi) < b$, but since $\vec a_{m+1}=0$ this requires $\vec a \cdot (g(\phi), h_{\min} (\phi))<b$ as well. This would imply $(g(\phi), h_{\min} (\phi))\notin \Xi$; however this is a contradiction as the arguments of Theorem~\ref{thm_hmin_achieved} show that for any $\phi \in \Pi$, $(g(\phi), h_{\min} (\phi))$ belongs to $\Xi$ (the arguments of Theorem~\ref{thm_hmin_achieved} demonstrated this for $\rho$ but they apply to any element of $\Pi$).

Having demonstrated $\vec a_{m+1}\ne 0$, we can define a function $f_\rho:\mathbb R^m \to \mathbb R$ as follows:
\begin{equation}\label{e:fdef}
f_\rho(\vec x) = \frac{b-\vec a_{[m]} \cdot \vec x}{\vec a_{m+1}}.
\end{equation}
Composing this function with $g$, we find that $f_\rho \circ g(\rho) = h_{\min}(\rho)$. Furthermore, for general $\phi \in \Pi$ the fact that $(g(\phi), h_{\min} (\phi))\in\Xi$ ensures $f_\rho \circ g(\phi) \le h_{\min}(\phi)$, and $h_{\min}(\phi)\le H_\phi(C|Z)$ by concavity of conditional Shannon entropy, so the map $f_\rho \circ g$, applied to any $\phi\in \Pi$, satisfies 
$$
f_\rho \circ g (\phi) \le H_\phi (C|Z) =  E_\phi(-\log_2 [\phi(C|Z)]).
$$
We now use $f_\rho \circ g$ to construct the desired entropy estimator as follows. We have
\begin{equation*}
f_\rho \circ g (\phi) = \frac{b-\vec a_{[m]} \cdot (M\phi + \vec k)}{\vec a_{m+1}}
=\vec n \cdot \phi + d
\end{equation*}
where $d=(b-\vec a_{[m]}\cdot \vec k)/\vec a_{m+1}$ is a constant and $\vec n = -(1/\vec a_{m+1})M^T\vec a_{[m]}$ is an $N$-dimensional vector; that is, it has one component for each possible distinct outcome pair $c,z$ for the random variable pair $C,Z$. Now we can define $K(c,z) := \vec n_{cz} + d$ to obtain a function of $C,Z$ satisfying 
$$
\mathbb E_\phi[K(CZ)]=\vec n \cdot \phi + d=f_\rho \circ g (\phi)\le E_\phi(-\log_2 [\phi(C|Z)]),
$$ and thus $K$ is an entropy estimator satisfying the conditions of the theorem.
\end{proof}

\section{Concavity of Conditional Shannon Entropy}\label{s:concavity}

It is known that conditional Shannon entropy is concave. For completeness, we provide a brief proof of how this follows from the concavity of (unconditional) Shannon entropy.  Let $\nu$ be a convex combination of $\nu_1$ and $\nu_2$, so that for all $(c,z)\in\mathcal{C}\times\mathcal{Z}$ we have $\nu(c,z) = \lambda\nu_1(c,z) + (1-\lambda)\nu_2(c,z)$ for some $\lambda\in[0,1]$. Then it follows that $\nu(C|z)$ is a convex mixture of $\nu_1(C|z)$ and $\nu_2(C|z)$ for each fixed $z$ for which $\nu(z)>0$:
\begin{equation}\label{e:nu12}
\nu(C|z) = \lambda\frac{\nu_1(z)}{\nu(z)}\nu_1(C|z)+(1-\lambda)\frac{\nu_2(z)}{\nu(z)}\nu_2(C|z),
\end{equation}
where it is straightforward to check that the coefficients of $\nu_1(C|z)$ and $\nu_2(C|z)$ are non-negative numbers summing to one. Using~\eqref{e:nu12} and the concavity of (unconditional) Shannon entropy we have the following.
\begin{align}
H_{\nu}(C\lvert Z) &= \sum_{z}H_{\nu}(C\lvert z)\nu(z)
                   \ge \sum_{z}\left(\lambda\frac{\nu_{1}(z)}{\nu(z)}H_{\nu_{1}}(C\lvert z) + (1-\lambda)\frac{\nu_{2}(z)}{\nu(z)}H_{\nu_{2}}(C\lvert z)\right)\nu(z)\nonumber\\
                   &=\sum_{z}\left(\lambda H_{\nu_{1}}(C\lvert z)\nu_{1}(z)+(1-\lambda)H_{\nu_{2}}(C\lvert z)\nu_{2}(z)\right)
                   = \lambda H_{\nu_{1}}(C\lvert Z) + (1-\lambda)H_{\nu_{2}}(C\lvert Z)
\end{align}

\section{Useful Lemmas}\label{s:Lemmas}

The following lemmas are used for arguments in Appendices~\ref{a:PEFproof_lbscme} and~\ref{a:scme}.

\begin{lemma}\label{lemma1_appendix}
If the distributions $\mu(X)$ and $\mu^{\prime}(X)$ dominate the non-negative function $f\colon\mathcal{X}\to\mathbb{R}^{+}$ with weight $w(f)=\sum_{x\in\mathcal{X}}f(x)=1-\epsilon$ for $\epsilon\in[0,1]$, i.e., $\mu(x)\ge f(x),\,\mu^{\prime}(x)\ge f(x),$ for all $x\in\mathcal{X}$, then $d_\mathrm{TV}(\mu,\mu^{\prime})\le\epsilon$. 
\end{lemma}
\begin{proof}
Using the definition of TV distance we have the required result as shown below.
\begin{align*}
    d_\mathrm{TV}(\mu,\mu^{\prime}) &=\frac{1}{2}\sum_{x\in\mathcal{X}}\abs{\mu(x)-f(x)+f(x)-\mu^{\prime}(x)}\\
    &\le \frac{1}{2}\sum_{x\in\mathcal{X}}\left(\abs{\mu(x)-f(x)} + \abs{\mu^{\prime}(x)-f(x)}\right)\\
    &=\frac{1}{2}\sum_{x\in\mathcal{X}}\left(\mu(x)-f(x) + \mu^{\prime}(x)-f(x)\right)=1-w(f)=\epsilon
\end{align*}
\end{proof}
\begin{lemma}\label{lemma2_appendix}
Suppose the function $f\colon\mathcal{X}\to\mathbb{R}^{+}$ has weight $w(f)=\sum_{x\in\mathcal{X}}f(x)=1-\epsilon$ where $\epsilon\in[0,1]$, and satisfies $f(x)\le p,\,\forall x\in\mathcal{X}$ for some fixed $p\ge 1/\abs{\mathcal{X}}$. Then there exists a distribution $\mu^{\prime}(X)$ such that $f(x)\le\mu^{\prime}(x)\le p$ holds for all $x\in\mathcal{X}$.
\end{lemma}
\begin{proof} If $\epsilon=0$, it suffices to take $\mu'(x)=f(x)$. If $\epsilon>0$,
we construct a distribution satisfying the two properties as follows. Define a function $\mu_{\lambda}$ with domain $\mathcal{X}$ as
\begin{equation}\label{lemma2_apdx_1}
    \mu_{\lambda}(x) = (1-\lambda)f(x) + \lambda p
\end{equation}
Then for any fixed $\lambda \in [0,1]$, $\mu_{\lambda}(x)$ is a convex combination of non-negative numbers and thus non-negative for any choice of $x$. We show that there exists a $\lambda\in[0,1]$ for which $\sum_x\mu_{\lambda}(x)=1$, making $\mu_\lambda$ a distribution. It is easy to verify that for $\lambda'=\epsilon/(p\abs{\mathcal{X}} + \epsilon - 1)$ the above function adds up to unity when summed over $x\in\mathcal{X}$. We just need to ensure that $\epsilon/(p\abs{\mathcal{X}} + \epsilon - 1)\in[0,1]$ holds. To see this, note that $p\abs{\mathcal{X}}\ge 1$ so we have $p\abs{\mathcal{X}} +\epsilon -1\ge \epsilon$, and since $\epsilon >0$ the quotient must indeed lie in $[0,1]$. Finally, $\mu_{\lambda'}(X)$ satisfies the bounds in the Lemma: since $f(x)\le p$ for all $x\in\mathcal{X}$, for any $\lambda \in [0,1]$ we have
\begin{align}\label{lemma2_apdx_2}
    p \ge p + (1-\lambda)(f(x)-p) = (1-\lambda)f(x) + \lambda p = f(x) + \lambda(p - f(x)) \ge f(x),\,\forall x\in\mathcal{X}
\end{align}
and the middle term above is $\mu_{\lambda'}(X)$ for $\lambda = \lambda'$.
\end{proof}

\section{Inequalities relating smooth average conditional min-entropy and smooth worst-case conditional min-entropy}\label{a:scme}
Here we state and prove a known inequality that relates two notions of smooth conditional min-entropy. We present this result without structuring random variables as stochastic sequences, i.e., instead of considering distributions of $\mathbf{C,Z},E$ we consider distributions of $X,Y$. The result and its proof can be adapted to the more general case involving sequence of random variables.

For a distribution $\mu\colon\mathcal{X}\times\mathcal{Y}\to[0,1]$ of $X,Y$ and the set $\mathcal{B}^{\epsilon}(\mu)$ of distributions of $X,Y$ defined as $\mathcal{B}^{\epsilon}(\mu)\coloneqq\{\sigma\colon\mathcal{X}\times\mathcal{Y}\to[0,1]\mid d_{\mathrm{TV}}(\mu,\sigma)\le\epsilon\}$, the $\epsilon$-smooth average conditional min-entropy is:
\begin{equation*}
\mathbb{H}_{\infty,\mu}^{\mathsf{avg},\epsilon}(X\lvert Y)\coloneqq\max_{\sigma\in\mathcal{B}^{\epsilon}(\mu)}\bigg[-\log_{2}\Big(\sum_{y\in\mathcal{Y}}\max_{x\in\mathcal{X}}\sigma(x\lvert y)\sigma(y)\Big)\bigg].
\end{equation*}
A stricter definition of smooth conditional min-entropy than the one stated above is the $\epsilon$-smooth ``worst-case'' conditional min-entropy, introduced in~\cite{10.1007/11593447_11}. It reads as follows:
\begin{equation}\label{e:smoothwst}
\mathbb{H}_{\infty,\mu}^{\mathsf{wst},\epsilon}(X\lvert Y) = \max_{\sigma\in\mathcal{B}^{\epsilon}(\mu)}\left[-\log_{2}\Big(\max_{x\in\mathcal{X},y\in\mathcal{Y}}\sigma(x\lvert y)\Big)\right].
\end{equation}
For purposes of randomness extraction or scenarios involving predictability of an adversary, the smooth average conditional min-entropy suffices. One can show that the notions of \emph{average-case} and \emph{worst-case} are equivalent up to an additive factor~\cite{10.1007/978-3-540-24676-3_31}. This is formalised in Proposition~\ref{l:avg_wst_scme}. 

\begin{proposition}\label{l:avg_wst_scme}
For a distribution $\mu\colon\mathcal{X}\times\mathcal{Y}\to[0,1]$ of $X,Y$ and $1\ge\epsilon\ge0,\,1>\epsilon^{\prime}>0$, the smooth \emph{worst case} conditional min-entropy and smooth \emph{average case} conditional min-entropy are related by the following inequalities:
\begin{equation}\label{e:avgMinEntropy_wstMinEntropy}
    \mathbb{H}_{\infty,\mu}^{\mathsf{wst},\epsilon}(X\lvert Y) \le \mathbb{H}_{\infty,\mu}^{\mathsf{avg},\epsilon}(X\lvert Y) \le \mathbb{H}_{\infty,\mu}^{\mathsf{wst},\epsilon + \epsilon^{\prime}}(X\lvert Y) + \log_{2}(1/\epsilon^{\prime})
\end{equation}
\end{proposition}
\begin{proof}
The first inequality holds immediately, since for every $\sigma(X,Y)\in\mathcal{B}^{\epsilon}(\mu)$ we have
\begin{equation}\label{e:max_ge_avg}
\max_{x\in\mathcal{X},y\in\mathcal{Y}}\sigma(x\lvert y) \ge \sum_{y\in\mathcal{Y}}\big(\max_{x\in\mathcal{X}}\sigma(x\lvert y)\big)\sigma(y).
\end{equation}
Taking $-\log_{2}$ of both sides we obtain $\mathbb{H}_{\infty,\sigma}^{\mathsf{wst}}(X\lvert Y)\le\mathbb{H}_{\infty,\sigma}^{\mathsf{avg}}(X\lvert Y)$, where $\mathbb{H}_{\infty,\sigma}^{\mathsf{wst}}(X\lvert Y)$ is the bracketed quantity in~\eqref{e:smoothwst}, and since this inequality holds for every $\sigma\in\mathcal{B}^{\epsilon}(\mu)$, we have $\mathbb{H}_{\infty,\mu}^{\mathsf{wst},\epsilon}(X\lvert Y)\le \mathbb{H}_{\infty,\mu}^{\mathsf{avg},\epsilon}(X\lvert Y)$. For the second inequality of~\eqref{e:avgMinEntropy_wstMinEntropy}, we want to show that $\mathbb{H}_{\infty,\mu}^{\mathsf{wst},\epsilon + \epsilon^{\prime}}(X\lvert Y)\ge \mathbb{H}_{\infty,\mu}^{\mathsf{avg},\epsilon}(X\lvert Y)-\log(1/\epsilon^{\prime})$ holds. Suppose the distribution $\nu(X,Y)\in\mathcal{B}^{\epsilon}(\mu)$ witnesses\footnote[2]{This is permissible due to the compactness of $\mathcal{B}^{\epsilon}(\mu)$ and the continuity of $\mathbb{H}_{\infty,\mu}^{\mathsf{avg}}(X\lvert Y)$.} $\mathbb{H}_{\infty,\mu}^{\mathsf{avg},\epsilon}(X\lvert Y)$, i.e., $\mathbb{H}_{\infty,\mu}^{\mathsf{avg},\epsilon}(X\lvert Y)=\mathbb{H}_{\infty,\nu}^{\mathsf{avg}}(X\lvert Y)$. It suffices to construct a distribution $\sigma(X,Y)\in\mathcal{B}^{\epsilon^{\prime}}(\nu)$ such that $\allowbreak\max_{x,y}\sigma(x\lvert y)\le p/\epsilon^{\prime}$ holds, where $p=\sum_{y}\max_{x}\nu(x\lvert y)\nu(y)=\mathbb{E}_{\nu(Y)}[\max_{x\in\mathcal{X}}\nu(x\lvert Y)]$. We begin by defining the sub-probability distribution $\tilde{\nu}(X,Y)$ as shown below.
\begin{equation}\label{e:subprob}
    \tilde{\nu}(x,y) = \nu(x,y)[\![\max_{x\in\mathcal{X}}\nu(x\lvert y) \le p/\epsilon^{\prime}]\!],
\end{equation}
where the notation $[\![\cdots ]\!]$ represents the function that evaluates to 1 if the enclosed condition holds, and zero if it does not. Basically, the definition of $\tilde{\nu}$ in~\eqref{e:subprob} involves discarding $\nu$ corresponding to those $y\in\mathcal{Y}$ for which $\max_{x}\nu(x\lvert y) > p/\epsilon^{\prime}$ holds. An application of Markov's inequality then shows that the weight $w(\tilde{\nu})$ of $\tilde{\nu}$ is at least $1-\epsilon'$:
\begin{align}\label{e:lb_subprob}
    w(\tilde{\nu}) = \sum_{x,y}\tilde{\nu}(x,y) &= \sum_{y\in\mathcal{Y}\colon \max_{x}\nu(x\lvert y)\le p/\epsilon^{\prime}}\sum_{x}\nu(x,y)\nonumber \\
    &= \mathbb{P}_{\nu(Y)}\left(\max_{x\in\mathcal{X}}\nu(x\lvert Y)\le p/\epsilon^{\prime}\right) \ge 1-\epsilon^{\prime}
\end{align}
Since $p = \mathbb{E}_{\nu(Y)}[\max_{x}\nu(x\lvert Y)]$,~\eqref{e:lb_subprob} follows. One way to now construct a distribution $\sigma\in\mathcal{B}^{\epsilon^{\prime}}(\mu)$ satisfying $\max_{x,y}\sigma(x\lvert y)\le p/\epsilon^{\prime}$ is to scale $\tilde{\nu}$, i.e., we define $\sigma(X,Y)$ as $\sigma(x,y)=\tilde{\nu}(x,y)/w(\tilde{\nu})$. Note that since $w(\tilde{\nu})\ge 1-\epsilon^{\prime}$ and $\epsilon^{\prime}\in(0,1)$, $w(\tilde{\nu})$ is positive; also, $\sigma\ge\tilde{\nu}$ holds since $w(\tilde{\nu})\le 1$. Together with the fact that $\nu\ge\tilde{\nu}$, we can use Lemma~\ref{lemma1_appendix} to show that $d_\mathrm{TV}(\sigma,\nu)\le 1-w(\tilde{\nu})\le\epsilon^{\prime}$. By definition of $\sigma$ we have $\sigma(x,y)\le p\sigma(y)/\epsilon^{\prime}$ for all choices of $x,y$, where $\sigma(y) = \frac{\nu(y)}{w(\tilde{\nu})}[\![\max_{x}\nu(x\lvert y)\le p/\epsilon^{\prime}]\!]$ for each $y$. With the convention that $\sigma(x\lvert y)$ is assigned the value $0$ when $\sigma(y)=0$, we then have $\sigma(x\lvert y)\le p/\epsilon^{\prime}$ for all choices of $x,y$. Membership of $\sigma$ in the set $\mathcal{B}^{\epsilon + \epsilon^{\prime}}(\mu)$ follows from the triangle inequality $d_\mathrm{TV}(\mu,\sigma)\le d_\mathrm{TV}(\mu,\nu)+d_\mathrm{TV}(\nu,\sigma)\le\epsilon + \epsilon^{\prime}$. And so we have constructed a distribution in $\mathcal{B}^{\epsilon+\epsilon^{\prime}}(\mu)$ such that $\max_{x,y}\sigma(x\lvert y)\le p/\epsilon^{\prime}$. Taking $-\log_{2}$ on both sides, we get $\mathbb{H}_{\infty,\sigma}^{\mathsf{wst}}(X\lvert Y)\ge \mathbb{H}_{\infty,\nu}^{\mathsf{avg}}(X\lvert Y)-\log_{2}(1/\epsilon^{\prime})$. As mentioned earlier, $\nu\in\mathcal{B}^{\epsilon}(\mu)$ witnesses $\mathbb{H}_{\infty,\mu}^{\mathsf{avg},\epsilon}(X\lvert Y)$, hence we have shown:
\begin{equation}\label{e:wstEntsigma_lb}
    \mathbb{H}_{\infty,\sigma}^{\mathsf{wst}}(X\lvert Y)\ge \mathbb{H}_{\infty,\mu}^{\mathsf{avg},\epsilon}(X\lvert Y)-\log_{2}(1/\epsilon^{\prime})
\end{equation}
Since, by definition, the smooth \emph{worst case} conditional min-entropy involves a maximum of the left hand side of~\eqref{e:wstEntsigma_lb} over the set $\mathcal{B}^{\epsilon + \epsilon^{\prime}}(\mu)$, this shows that $\mathbb{H}_{\infty,\mu}^{\mathsf{wst},\epsilon+\epsilon^{\prime}}(X\lvert Y)\ge\mathbb{H}_{\infty,\mu}^{\mathsf{avg},\epsilon}(X\lvert Y)-\log_{2}(1/\epsilon^{\prime})$ holds, from which the second inequality in~\eqref{e:avgMinEntropy_wstMinEntropy} follows. 
\end{proof}
In the asymptotic limit of a large number $n$ of trials, constant factors vanish in measuring per-trial min entropy, and since $\epsilon'$ can be made arbitrarily small, \eqref{e:avgMinEntropy_wstMinEntropy} enables us to consider either definition when considering asymptotic performance.

\section{Proof of Proposition \ref{beta_threshold}}\label{s:Proof_beta_threshold}

\noindent\textbf{Proposition.} 
For the set of behaviours $\Pi_{\mathrm{NS}}$, the PEF optimisation in~\eqref{eq_PEF_opt} is independent of the power $\beta$ for $\beta\ge\log_{2}(4/3)$.

\begin{proof}
For a fixed value of $n$ and $\epsilon$ the optimisation problem in~\eqref{eq_PEF_opt} is equivalent to the following:
\begin{align}\label{e:PEFopt_fixed_nb}
\text{Maximise: }& \mathbb{E}_{\rho}[\log_{2}(F(ABXY))]\nonumber\\
\text{Subject to: }& \mathbb{E}_{\mu_{\mathrm{PR}}^i}[F(ABXY)\mu^i_{\mathrm{PR}}(AB\lvert XY)^{\beta}]\le 1, \text{ for all } i \in \{0,1\}^3,\nonumber\\
& \mathbb{E}_{\mu_{\mathrm{LD}}^j}[F(ABXY)\mu^j_{\mathrm{LD}}(AB\lvert XY)^{\beta}]\le 1,\text{ for all }j \in \{0,1\}^4,\nonumber\\
& F(abxy)\ge 0, \, \forall a,b,x,y\in\{0,1\},
\end{align} 
where the constraints range over the extremal points of $\Pi_{\mathrm{NS}}$ as given in \eqref{eq_PRbox} and \eqref{eq_LDbox}. We show that for $\beta\ge \log_2(4/3)$, the above constraints are equivalent to 
\begin{align}\label{e:equivconstraints}
&\mathbb{E}_{\mu_{\mathrm{LD}}^j}[F(ABXY)\mu^j_{\mathrm{LD}}(AB\lvert XY)]\le 1,\text{ for all }j \in \{0,1\}^4,\nonumber\\
& F(abxy)\ge 0, \, \forall a,b,x,y\in\{0,1\},
\end{align}
noticing that $\beta $ does not appear in \eqref{e:equivconstraints}.

It is immediate to see that the constraints of \eqref{e:PEFopt_fixed_nb} imply \eqref{e:equivconstraints}: since $\mu(AB|XY)$ is always zero or one for local deterministic distributions, in this case we have $\mu(AB|XY)^\beta= \mu(AB|XY)$ and thus for each choice of $j$ we have $\mathbb{E}_{\mu_{\mathrm{LD}}^j}[F(ABXY)\mu^j_{\mathrm{LD}}(AB\lvert XY)^{\beta}]\le 1$ implying the non-$\beta$ counterpart ${\mathbb{E}_{\mu_{\mathrm{LD}}^j}[F(ABXY)\mu^j_{\mathrm{LD}}(AB\lvert XY)]\le 1}$ in \eqref{e:equivconstraints}. Now we demonstrate the reverse implication. First, the argument just given also works in the opposite direction to show that the the non-$\beta$ constraints of \eqref{e:equivconstraints} imply the corresponding constraints (with $\beta$) in \eqref{e:PEFopt_fixed_nb}. We thus need only to show that the $\mathbb{E}_{\mu_{\mathrm{PR}}^i}[\cdots]\le 1$ in \eqref{e:PEFopt_fixed_nb} are implied as well. We give a specific argument for the PR box given in Table \ref{tab:Simplex_PR1}; symmetric arguments apply for the other PR boxes.  Since any distribution $\mu(ABXY)$ is the behaviour $\mu(AB\lvert XY)$ times a fixed settings distribution $\sigma_{\mathrm{s}}(XY)$, we can express the product $F(abxy)\sigma_{\mathrm{s}}(xy)$ as $F^{\prime}(abxy)$ for all choices of $(a,b,x,y)$ when the expectation functional $\mathbb{E}[\cdot]$ is written out in full. The constraints \eqref{e:equivconstraints} then imply, by summing over the eight of them corresponding to the eight local deterministic distributions appearing in Table \ref{tab:Simplex_PR1} (a set we denote $\mathsf{LD}_{1}$), that 
\begin{align}\label{e:PEFopt_re1}
\sum_{a,b,x,y}F^{\prime}(abxy)\sum_{\mu_{\mathrm{LD}}\in\mathsf{LD}_{1}}\mu_{\mathrm{LD}}(ab\lvert xy)^{2} &\le 8.
\end{align}
Noticing that the inner sum above is always 3 or 1 (this corresponds to the number of 1s appearing in each column of Table \ref{tab:Simplex_PR1}, with the result given in Table \ref{t:ldtab}), we can now rewrite~\eqref{e:PEFopt_re1} as $ 3M+N\le 8$, where $M = F'(0000)+F'(0001)+F'(0010)+F'(0111)+F'(1011)+F'(1100)+F'(1101)+F'(1110)$ and $N=F'(0011)+F'(1111)+F'(1001)+F'(0101)+F'(0110)+F'(1010)+F'(0100)+F'(1000)$.
\begin{table}[H]
\begin{minipage}{0.45\linewidth}
\centering
{\renewcommand{\arraystretch}{1.24}%
\begin{tabular}{|r|r|c|c|c|c|}
\cline{3-6}
\multicolumn{1}{r}{} & \multicolumn{1}{r}{} & \multicolumn{4}{|c|}{$ab$}\\
\cline{3-6}
\multicolumn{1}{r}{} & \multicolumn{1}{r}{} & \multicolumn{1}{|c}{$00$} & \multicolumn{1}{c}{$01$} & \multicolumn{1}{c}{$10$} & \multicolumn{1}{c|}{$11$}\\
\hline
\multirow{4}{*}{$xy$} & $00$ & $2^{-(1+\beta)}$ & $0$ & $0$ & $2^{-(1+\beta)}$ \\
\cline{3-6}
 & $01$ & $2^{-(1+\beta)}$ & $0$ & $0$ & $2^{-(1+\beta)}$ \\
\cline{3-6}
 & $10$ & $2^{-(1+\beta)}$ & $0$ & $0$ & $2^{-(1+\beta)}$ \\
\cline{3-6}
 & $11$ & $0$ & $2^{-(1+\beta)}$ & $2^{-(1+\beta)}$ & $0$ \\
\hline
\end{tabular}}
\caption{Tabular representation for $\displaystyle\mu_{\mathrm{PR},1}(AB|XY)^{1+\beta}$.}\label{t:prtab}
\end{minipage}
\hfill%
\begin{minipage}{0.45\linewidth}
\centering
{\renewcommand{\arraystretch}{1.24}%
\begin{tabular}{|r|r|c|c|c|c|}
\cline{3-6}
\multicolumn{1}{r}{} & \multicolumn{1}{r}{} & \multicolumn{4}{|c|}{$ab$}\\
\cline{3-6}
\multicolumn{1}{r}{} & \multicolumn{1}{r}{} & \multicolumn{1}{|c}{$00$} & \multicolumn{1}{c}{$01$} & \multicolumn{1}{c}{$10$} & \multicolumn{1}{c|}{$11$}\\
\hline
\multirow{4}{*}{$xy$} & $00$ & $3$ & $1$ & $1$ & $3$ \\
\cline{3-6}
 & $01$ & $3$ & $1$ & $1$ & $3$ \\
\cline{3-6}
 & $10$ & $3$ & $1$ & $1$ & $3$ \\
\cline{3-6}
 & $11$ & $1$ & $3$ & $3$ & $1$ \\
\hline
\end{tabular}}
\caption{Tabular representation of the values of the sum $\displaystyle\sum_{\mu_{\mathrm{LD}}\in\mathsf{LD}_{1}}\mu_{\mathrm{LD}}(ab\lvert xy)^{2}$.}\label{t:ldtab}
\end{minipage}
\end{table}
Since $M,N$ are both non-negative, we can drop $N$ to find that $3M+N\le 8$ implies $M\le8/3=2^{1+\log_2(4/3)}$ which in turn implies $M\le 2^{1+\beta}$ whenever $\beta\ge\log_{2}(4/3)$. Since $\mathbb{E}_{\mu_{\mathrm{PR},1}}[F(ABXY)\mu_{\mathrm{PR},1}(AB\lvert XY)^{\beta}]$ is equal to $M(1/2)^{1+\beta}$ (see Table \ref{t:prtab}) the constraint $\mathbb{E}_{\mu_{\mathrm{PR},1}}[\cdots]\le 1$ follows. \end{proof}

We remark that this inequality condition $\beta\ge\log_{2}(4/3)$ is tight in the following sense: there exists a non-negative function $F(abxy)$ violating the PR box constraint of $\mu_{\mathrm{PR},1}$ appearing in \eqref{e:PEFopt_fixed_nb} for any $\beta<\log_{2}(4/3)$, while satisfying \eqref{e:equivconstraints} for any positive $\beta$ -- and consequently satisfying all the constraints of \eqref{e:PEFopt_fixed_nb} for $\beta\ge\log_{2}(4/3)$ per the argument in the above proof. Thus the feasible set of \eqref{e:PEFopt_fixed_nb} always excludes this particular choice of $F$ for $\beta<\log_{2}(4/3)$ and includes it for $\beta\ge\log_{2}(4/3)$. This function is $F(abxy)=(1/3)[\![a\oplus b = xy]\!]\sigma_s(xy)^{-1}$; fixing $\beta=\log_{2}(4/3)-\epsilon$ for some choice of $\epsilon$ in the interval $(0,\log_{2}(4/3)$, we can check that all the LD boxes satisfy the inequality $\sum_{a,b,x,y}F^{\prime}(abxy)\mu_{\mathrm{LD}}(ab\lvert xy)^{1+\beta}\le 1$; the value of the expression is always either 1/3 or 1. However, for the PR box $\mu_{\mathrm{PR},1}$ in Table~\ref{tab:1PR_8LD} we obtain $\sum_{a,b,x,y}F^{\prime}(abxy)\mu_{\mathrm{PR},1}(ab\lvert xy)^{1+\beta}=2^{\epsilon}>1$, which is a violation.  

\newpage
\printbibliography

\end{document}